\pdfoutput=1 
\documentclass[format=acmsmall]{acmart}
\renewcommand\footnotetextcopyrightpermission[1]{} 
\usepackage{booktabs} 

\usepackage[linesnumbered,vlined, boxed, ruled]{algorithm2e}

\SetAlFnt{\small}
\SetAlCapFnt{\small}
\SetAlCapNameFnt{\small}
\SetAlCapHSkip{0pt}
\IncMargin{-\parindent}

\usepackage{xspace}
\usepackage{graphicx}
\usepackage{subfigure}
\usepackage{mathtools}
\usepackage{color}

\usepackage{amsmath,amssymb}
\usepackage{epsfig}
\usepackage{epstopdf, url}

\usepackage[T1]{fontenc}
\usepackage[utf8]{inputenc}
\usepackage[english]{babel}
\usepackage[font=small,labelfont=bf]{caption}
\usepackage{booktabs}
\usepackage{varwidth}

\theoremstyle{remark}
\newtheorem{remark}{Remark}

\DeclarePairedDelimiter\ceil{\lceil}{\rceil}
\DeclarePairedDelimiter\floor{\lfloor}{\rfloor}

\SetKwInOut{Input}{input}\SetKwInOut{Output}{output}
\SetKwComment{Comment}{//}{}
\SetEndCharOfAlgoLine{}
\SetKw{True}{true}
\SetKw{False}{false}
\SetKw{Not}{not}

\DeclareMathOperator{\var}{Var}

\DeclareMathOperator{\mean}{E}
\DeclareMathOperator{\prb}{Pr}
\newcommand{\prob}[1]{\prb[#1]}
\newcommand{\xE}{\mathcal{E}}

\newcommand{\xB}{\mathcal{B}}

\newcommand{\xS}{\mathcal{S}}

\newcommand{\xU}{\mathcal{U}}
\newcommand{\bbar}{\bar{b}}
\newcommand{\alphabar}{\bar{\alpha}}
\let\eps\epsilon
\DeclareMathOperator{\frc}{frac}
\newcommand{\bS}{\mathbb{S}}
\newcommand{\bL}{\mathbb{L}}
\newcommand{\bC}{\mathbb{C}}

\newcommand{\fa}{{\tilde f}}
\newcommand{\yA}{{\tilde A}}
\newcommand{\ypi}{{\tilde \pi}}
\newcommand{\yC}{{\tilde C}}

\newenvironment{romlist}{

  \begin{enumerate}
  }{\end{enumerate}}
  
\newcommand{\eat}[1]{}
\newcommand{\revision}[1]{#1}

\begin{document}

\title[Temporally-Biased Sampling Schemes]{Temporally-Biased Sampling Schemes\\ for Online Model Management}

\author{Brian Hentschel}
\affiliation{%
  \institution{Harvard University}
  \city{Cambridge} 
  \state{Massachusetts} 
  \country{USA}
}
\email{bhentschel@g.harvard.edu}

\author{Peter J. Haas}
\orcid{0000-0001-5694-3065}
\affiliation{%
  \institution{University of Massachusetts Amherst}
  \city{Amherst} 
  \state{Massachusetts} 
  \country{USA}
}
\email{phaas@cs.umass.edu}

\author{Yuanyuan Tian}
\affiliation{%
  \institution{IBM Research -- Almaden}
  \city{San Jose} 
  \state{California} 
  \country{USA}
}
\email{ytian@us.ibm.com}


\begin{abstract}
To maintain the accuracy of supervised learning models in the presence of evolving data streams, we provide temporally-biased sampling schemes that weight recent data most heavily, with inclusion probabilities for a given data item decaying over time according to a specified ``decay function''. We then periodically retrain the models on the current sample. This approach speeds up the training process relative to training on all of the data. Moreover, time-biasing lets the models adapt to recent changes in the data while---unlike in a sliding-window approach---still keeping some old data to ensure  robustness in the face of temporary fluctuations and periodicities in the data values. In addition, the sampling-based approach allows existing analytic algorithms for static data to be applied to dynamic streaming data essentially without change. We provide and analyze both a simple sampling scheme (T-TBS) that probabilistically maintains a target sample size and a novel reservoir-based scheme (R-TBS) that is the first to provide both control over the decay rate and a guaranteed upper bound on the sample size. If the decay function is exponential, then control over the decay rate is complete, and R-TBS maximizes both expected sample size and sample-size stability. For general decay functions, the actual item inclusion probabilities can be made arbitrarily close to the nominal probabilities, and we provide a scheme that allows a trade-off between sample footprint and sample-size stability. R-TBS rests on the notion of a ``fractional sample'' and allows for data arrival rates that are unknown and time varying (unlike T-TBS). The R-TBS and T-TBS schemes are of independent interest, extending the known set of unequal-probability sampling schemes.
We discuss distributed implementation strategies; experiments in Spark illuminate the performance and scalability of the algorithms, and show that our approach can increase machine learning robustness in the face of evolving data. 
\end{abstract}

\maketitle
\thispagestyle{empty}


\section{Introduction}\label{sec:intro}

A key challenge for machine learning (ML) is to keep ML models from becoming stale in the presence of evolving data. In the context of the emerging Internet of Things (IoT), for example, the data comprises dynamically changing sensor streams~\cite{WhitmoreAX15}, and a failure to adapt to changing data can lead to a loss of predictive power.

One way to deal with this problem is to re-eng\-i\-neer existing static supervised learning algorithms to become adaptive. \revision{Some parametric methods---such as support-vector machines (SVMs) without the ``kernel trick'', hidden Markov models, and regression models---can indeed be re-engineered so that the parameters are time-varying, but for many popular non-parametric algorithms such as k-nearest neighbors (kNN) classifiers, decision trees, random forests, gradient boosted machines, and so on, it is not at all clear how re-engineering can be accomplished. The 2017 Kaggle Data Science Survey~\cite{kaggleSurvey} indicates that a substantial portion of the models that developers use in industry are non-parametric.} We therefore consider alternative approaches in which we periodically retrain ML models, allowing static ML algorithms to be used in dynamic settings essentially as-is. There are several possible retraining approaches.

\textbf{Retraining on cumulative data:} Periodically retraining a model on all of the data that has arrived so far is clearly infeasible because of the huge volume of data involved. Moreover, recent data is swamped by the massive amount of past data, so the retrained model is not sufficiently adaptive.

\textbf{Sliding windows:} A simple sliding-window approach would be to, e.g., periodically retrain on the data from the last two hours. If the data arrival rate is high and there is no bound on memory, then one must deal with long retraining times caused by large amounts of data in the window. The simplest way to bound the window size is to retain the last $n$ items. Alternatively, one could try to subsample within the time-based window~\cite{GemullaL08}. The fundamental problem with all of these bounding approaches is that old data is completely forgotten; the problem is especially severe when the data arrival rate is high. This can undermine the robustness of an ML model in situations where old patterns can reassert themselves. For example, a singular event such as a holiday, stock market drop, or terrorist attack can temporarily disrupt normal data patterns, which will reestablish themselves once the effect of the event dies down. Periodic data patterns can lead to the same phenomenon. Another example, from \cite{XieTSBH15}, concerns influencers on Twitter: a prolific tweeter might temporarily stop tweeting due to travel, illness, or some other reason, and hence be completely forgotten in a sliding-window approach. Indeed, in real-world Twitter data, almost a quarter of top influencers were of this type, and were missed by a sliding window approach.

\textbf{Temporally-biased sampling:} An appealing alternative is a temporally biased sampling-based approach, i.e., maintaining a sample that heavily emphasizes recent data but also contains a small amount of older data, and periodically retraining a model on the sample. By using a time-biased sample, the retraining costs can be held to an acceptable level while not sacrificing robustness in the presence of recurrent patterns. This approach was proposed in \cite{XieTSBH15} in the setting of graph analysis algorithms, and has recently been adopted in the MacroBase system~\cite{BailisGMNRS17}. The orthogonal problem of choosing when to retrain a model is also an important question, and is related to, e.g.,  the literature on ``concept drift''~\cite{GamaZBPB14}; in this paper we focus on the problem of how to efficiently maintain a time-biased sample.

In more detail, our time-biased sampling algorithms ensure that the ``appearance probability'' for a given data item---i.e., the probability that the item appears in the current sample---decays over time at a controlled rate. Specifically, we assume that items arrive in \emph{batches} $\xB_1,\xB_2,\ldots$, at time points $t_1,t_2,\cdots$, where each batch contains 0 or more items and $t_k\to\infty$ as $k\to\infty$. Our goal is to generate a sequence $\{S_k\}_{k\ge 1}$, where $S_k$ is a sample of the items that have arrived at or prior to time~$t_k$, i.e., a sample of the items in $\xU_k=\bigcup_{i=1}^k \xB_i$. These samples should be biased towards recent items, in the following sense. For $1\le i\le k$, denote by $\alpha_{i,k}=t_k-t_i$ the \emph{age} at time $t_k$ of an item belonging to batch $\xB_i$. Then for arbitrary times $t_i\le t_j$ and items $x\in\xB_i$ and $y\in\xB_j$,
\begin{equation}\label{eq:expratio}
\prob{x\in S_k}/\prob{y\in S_k}=f(\alpha_{i,k})/f(\alpha_{j,k}),
\end{equation}
for any batch arrival time $t_k\ge t_j$, where $f$ is a nonnegative and nonincreasing \emph{decay function} such that $f(0)=1$. Thus items with a given timestamp are sampled uniformly, and items with different timestamps are handled in a carefully controlled manner, such that the appearance probability for an item of age $\alpha$ is proportional to $f(\alpha)$. The criterion in \eqref{eq:expratio}, which is expressed in terms of wall-clock time, is natural and appealing in applications and, importantly, is interpretable and understandable to users.

\revision{\textbf{Choosing a decay function:} Although our primary focus is on developing sampling methods that can support a variety of decay functions, the question of how to choose a good decay function $f$ is important, and a topic of ongoing research. Cohen and Strauss~\cite{CohenS06} discuss the choice of decay functions in the setting of time-decayed aggregates in telecommunications networks, and argue that the proper choice of a decay function depends on domain knowledge. In an example involving a reliability comparison between two telecommunications links, they conclude that a polynomial decay function best matches their intuition on how the comparison should evolve over time. More generally, the authors assert that there is a trade-off between the ability to decay quickly in the short term and the ability to potentially retain older data, and that the choice should depend on the perceived importance of older data and on the time scales of correlations between values; e.g., the latter might correspond to the amount of time it takes a prior pattern to reassert itself. The authors therefore argue for supporting a rich class of decay functions. Similarly, Xie et al.~\cite{XieTSBH15} show how a decay function can be chosen to meet application-specific criteria. For example, by using an exponential decay function $f(\alpha)=\exp(-\lambda \alpha)$ with $\lambda=0.058$, a data item from 40 batches ago is $1/10$ as likely to appear in the current analysis as a newly arrived item. If training data is available, $\lambda$ can also be chosen to maximize accuracy of a specified ML model via cross validation combined with grid search---in our experiments, where ground truth data was available, we found empirically that accuracy tended to be a quasiconvex function of $\lambda$, which bodes well for automatic optimization methods such as stochastic gradient descent. We find exponential and sub-exponential decay functions such as polynomial decay to be of the greatest interest. As will become apparent, exponential decay functions, though of limited flexibility, are the easiest to work with, and most prior work has centered around exponential decay.
Super-exponential decay functions are of less practical interest: older items decay too fast and the sampling scheme behaves essentially like a sliding window.}


\textbf{Sample-size control:} For the case in which the item-arrival rate is high, the main issue is to keep the sample size from becoming too large. On the other hand, when the incoming batches become very small or widely spaced, the sample sizes for all of the time-biased algorithms that we discuss (as well as for sliding-window schemes based on wall-clock time) can become small. This is a natural consequence of treating recent items as more important, and is characteristic of any sampling scheme that satisfies \eqref{eq:expratio}. We emphasize that---as shown in our experiments---a smaller, but carefully time-biased sample typically yields better prediction accuracy than a sample that is larger due to overloading with too much recent data or too much old data. I.e., more sample data is not always better. Indeed, with respect to model management, this decay property can be viewed as a feature in that, if the data stream dries up and the sample decays to a very small size, then this is a signal that there is not enough new data to reliably retrain the model, and that the current version should be kept for now. In any case, we show that, within the class of sampling algorithms that support exponential decay, our new R-TBS algorithm with an exponential decay function maximizes the expected sample size whenever the sample is not saturated.

\textbf{Prior work:} It is surprisingly hard to both enforce \eqref{eq:expratio} and to bound the sample size. As discussed in detail in Section~\ref{sec:relwork}, prior algorithms cannot handle arbitrary decay functions, and can only support ``forward decay'' schemes, which are less intuitive for users than ``backward decay'' schemes that enforce \eqref{eq:expratio}, and can lead to both numerical issues and poor adaptation behavior for ML algorithms. The only decay functions that support \eqref{eq:expratio} and are handled by prior algorithms are the exponential decay functions, because backward and forward decay coincide in this case. Even in this restricted setting, prior algorithms that bound the sample size either cannot consistently enforce \eqref{eq:expratio} or cannot handle wall-clock time. Examples of the former include algorithms based on the A-Res scheme of Efraimidis and Spirakis~\cite{efraimidisS06} and Chao's algorithm~\cite{Chao82}. A-Res enforces conditions on the \emph{acceptance} probabilities of items; this leads to appearance probabilities that, unlike \eqref{eq:expratio}, are both hard to compute and not intuitive. In Appendix~\ref{sec:chao}
we demonstrate how Chao's algorithm can be specialized to the case of exponential decay and modified to handle batch arrivals. We then observe that the resulting algorithm fails to enforce \eqref{eq:expratio} either when initially filling up an empty sample or in the presence of data that arrives slowly relative to the decay rate, and hence fails if the data rate fluctuates too much. The second type of algorithm, due to Aggarwal~\cite{Aggarwal06}, can only control appearance probabilities based on the indices of the data items. For example, after $n$ items arrive, one could require that, for some specified $k<n$, the $(n-k)$th item is 1/10 as likely to be in the sample as the current item. If the data arrival rate is constant, then this might correspond to a constraint of the form ``a data item that arrived 10 hours ago is 1/10 as likely to be in the sample as the current item". For varying arrival rates, however, it is impossible to enforce the latter type of constraint, and a large batch of arriving data can prematurely flush out older data. Thus our new sampling schemes are interesting in their own right, significantly expanding the set of unequal-probability sampling techniques.

\textbf{T-TBS:} We first provide and analyze Target\-ed-Size Time-Biased Sampling (T-TBS), a relatively simple algorithm that generalizes the Bernoulli sampling scheme in \cite{XieTSBH15}. T-TBS allows complete control over the decay rate (expressed in wall-clock time) and probabilistically maintains a target sample size. That is, the expected and average sample sizes converge to the target and the probability of large deviations from the target decreases exponentially or faster in both the target size and the deviation size. T-TBS is easy to implement and highly scalable when applicable, but only works under the strong restriction that the mean sizes of the arriving batches are constant over time and known a priori. T-TBS is a good choice in some scenarios (see Section~\ref{sec:ttbs}), but many applications have non-constant, unknown mean batch sizes, thus cannot tolerate sample overflows.

\textbf{R-TBS:} We then provide a novel algorithm, Reservoir-Based Time-Biased Sampling (R-TBS), that is the first to simultaneously enforce \eqref{eq:expratio} at all times, provide a guaranteed upper bound on the sample size, and allow unknown, varying data arrival rates. Guaranteed bounds are desirable because they avoid memory management issues associated with sample overflows, especially when large numbers of samples are being maintained---so that the probability of \emph{some} sample overflowing is high---or when sampling is being performed in a limited memory setting such as at the ``edge'' of the IoT. Also, bounded samples reduce variability in retraining times and do not impose upper limits on the incoming data flow. For an exponential decay function, the appearance probability for an item of age $\alpha$ is always exactly proportional to $f(\alpha)$; in general, for a ``cutoff age'' $\alpha^*$, the appearance probability for an item of age $\alpha$ is proportional to $\fa(\alpha)$, where $\fa(\alpha)=f(\alpha)$ for $\alpha\le \alpha^*$ but $\fa(\alpha)\not=f(\alpha)$  for $\alpha> \alpha^*$. At the cost of additional storage, the user can make $\alpha^*$ arbitrarily large---so that only a small set of very old items are affected---and the discrepancy $|f-\fa|$ for these old items arbitrarily small. We emphasize that, even though R-TBS involves some approximations in the case of general decay functions, the magnitude of these approximations is completely controllable a priori by the user; in contrast, prior schemes such as A-Res and Chao's Algorithm offer no control over departures from \eqref{eq:expratio} and indeed it can be difficult even to quantify the extent of these departures.

The idea behind R-TBS is to adapt the classic reservoir sampling algorithm, which bounds the sample size but does not allow time biasing. Our approach rests on the notion of a ``fractional'' sample whose nonnegative size is real-valued in an appropriate sense. For exponential decay, we show that R-TBS maximizes the expected sample size whenever the data arrival rate is low and also minimizes the sample-size variability; in general, there again is a user-controllable tradeoff between storage requirements and sample size stability.

\textbf{Distributed implementation:} Both T-TBS and R-TBS can be parallelized. Whereas T-TBS is relatively straightforward to implement, an efficient distributed implementation of R-TBS is nontrivial. We exploit various implementation strategies to reduce I/O relative to other approaches, avoid unnecessary concurrency control, and make decentralized decisions about which items to insert into, or delete from, the reservoir.

\textbf{Extensions of our prior work:} A preliminary version of this work appeared in \cite{HentschelHT18}; that paper focused entirely on the case of exponential decay and was missing many of the proofs for the given theoretical results. In the current paper, we extend our results to the setting of general decay functions. Handling such functions requires significant extensions to the algorithms, theory, and experimental study given in \cite{HentschelHT18}. Interestingly, viewing the original R-TBS algorithm in \cite{HentschelHT18} as a special case of the general algorithm has led to streamlining of the original algorithm as well as its theoretical analysis. The current paper contains all relevant proofs. 


\textbf{Organization:} The rest of the paper is organized as follows. In Section~\ref{sec:background} we describe our batch-arrival model and, to provide context for the current work, discuss two prior simple sampling schemes: a simple Bernoulli scheme as in~\cite{XieTSBH15} and the classical reservoir sampling scheme, modified for batch arrivals.  These methods either bound the sample size but do not control the decay rate, or control the decay rate but not the sample size. We next present and analyze the T-TBS and R-TBS algorithms in Section~\ref{sec:ttbs} and Section~\ref{sec:tbsamp}. We describe the distributed implementation in Section~\ref{sec:imp}, and Section~\ref{sec:exp} contains experimental results. We review the related literature in Section~\ref{sec:relwork} and conclude in Section~\ref{sec:concl}.

\section{Background}\label{sec:background}

For the remainder of the paper, we focus on settings in which batches arrive at regular time intervals, so that $t_i=i\Delta$ for some $\Delta>0$. This simple integer batch sequence often arises from the discretization of time~\cite{QianHSWZZZYZ13,ZahariaDLSS13}. Specifically, the continuous time domain is partitioned into intervals of length $\Delta$, and the items are observed only at times $\{k\Delta:k=1,2,\ldots\}$. All items that arrive in an interval $\bigl((k-1)\Delta,k\Delta\bigr]$ are treated as if they arrived at time $k\Delta$, i.e., at the end of the interval, so that all items in batch $\xB_i$ have time stamp~$i\Delta$. It follows that the age at time~$t_k$ of an item that arrived at time $t_i\le t_k$ is simply $\alpha_{i,k}=(k-i)\Delta$.

In this section, we briefly review two classical sampling schemes whose properties we will combine in the R-TBS algorithm.

\begin{algorithm}[h]
\caption{Bernoulli time-biased sampling (B-TBS)}\label{alg:bernsamp}
$\lambda$: decay factor ($\ge 0$)
\BlankLine
Initialize: $S\gets \emptyset$; $p\gets e^{-\lambda\Delta}$\Comment*[r]{$p=$ retention prob.}
\For{$i\gets1,2,\ldots$}{
$M \gets \textsc{Binomial}(|S|,p)$\Comment*[r]{simulate $|S|$ trials}\label{ln:binom}
$S \gets \textsc{Sample}(S,M)$\Comment*[r]{retain $M$ random elements}\label{ln:retainbinom}
$S\gets S\cup \xB_i$\;\label{ln:accept}
output $S$}
\end{algorithm}

\revision{\textbf{Bernoulli Time-Biased Sampling (B-TBS):} A well known, simple Bernoulli time-biased sampling scheme processes each incoming item, one at a time, by first downsampling the current sample and then accepting the incoming item into the sample with probability~1. Downsampling is accomplished by flipping a coin independently for each item in the sample: an item is retained in the sample with probability~$p$ and removed with probability~$1-p$. To adapt this sampling scheme to our batch-arrival setting, we process incoming items a batch at a time, and implicitly assume an exponential decay function $f(\alpha)=e^{-\lambda\alpha}$, setting $p=e^{-\lambda\Delta}$. Moreover, we take advantage of the fact that the foregoing downsampling operation is probabilistically equivalent to pre-selecting the number $M$ of items to retain according to a binomial distribution and then choosing the actual set of $M$ retained items uniformly from the elements in the current sample; see Appendix~\ref{sec:batchproofs}
for a proof of this fact. Generating a sample of $M$ can be done efficiently using standard algorithms \cite{StadloberZ99}, and obviates the need for executing multiple coin flips.

The resulting sampling scheme is given as Algorithm~\ref{alg:bernsamp}. At each time~$t_i$ we accept each incoming item $x\in \xB_i$ into the sample with probability~1 (line~\ref{ln:accept}). Downsampling is accomplished in lines \ref{ln:binom} and \ref{ln:retainbinom}: the function $\textsc{Binomial}(j,r)$ returns a random sample from the binomial distribution with $j$ independent trials and success probability $r$ per trial, and the function $\textsc{Sample}(A,m)$ returns a uniform random sample, without replacement, containing $\min(m,|A|)$ elements of the set $A$; note that the function call $\textsc{Sample}(A,0)$ returns an empty sample for any empty or nonempty $A$.

To see that Algorithm~\ref{alg:bernsamp} enforces the relation in \eqref{eq:expratio} as required, observe that the sequence of samples is a set-valued Markov process, so that
\[
\prob{x\in S_{i+j}\mid x\in S_{i+j-1}, x\in S_{i+j-2},\ldots,x\in S_i}= \prob{x\in S_{i+j}\mid x\in S_{i+j-1}} 
\]
for $i,j\ge 1$. We then have, for $x\in\xB_i$,
\begin{equation}\label{eq:bern1}
\prob{x\in S_k}
=\prob{x\in S_i}\times\prod_{j=1}^{k-i} \prob{x\in S_{i+j}\mid x\in S_{i+j-1}}
= 1\times p^{k-i}=e^{-\lambda (k-i)\Delta} = e^{-\lambda(t_k-t_i)},
\end{equation}
and \eqref{eq:expratio} follows immediately from \eqref{eq:bern1}. Thus Algorithm~\ref{alg:bernsamp} precisely controls the relative inclusion probabilities according to the exponential decay function $f$ given above. This is the algorithm used, e.g., in \cite{XieTSBH15} to implement time-biased edge sampling in dynamic graphs.}

Unfortunately, the user cannot independently control the expected sample size, which is completely determined by $\lambda$ and the sizes of the incoming batches. In particular, if the batch sizes systematically grow over time, then sample size will grow without bound. Arguments in \cite{XieTSBH15} show that if $\sup_i |\xB_i|<\infty$, then the sample size can be bounded, but only probabilistically. Because B-TBS is a special case of T-TBS with an exponential decay function and a unitary acceptance probability for arriving items, the results in Section~\ref{sec:ttbs} represent a significant extension and refinement of the analysis in \cite{XieTSBH15}.

\begin{algorithm}[tbh]
\caption{Batched reservoir sampling (B-RS)}\label{alg:rs}
$n$: maximum sample size
\BlankLine
Initialize: $S\gets\emptyset$; $W\gets 0$;
\For{$i=1,2,\ldots$}{
$C=\min(n,W+|\xB_i|)$\Comment*[r]{new sample size}\label{ln:bddsize}
$M \gets \textsc{HyperGeo}\bigl(C,|\xB_i|,W\bigr)$\;\label{ln:hyper}
\Comment{add $M$ elements to $S$, overwriting $\max\bigl(M-(n-|S|),0\bigr)$ items}
$S\gets \textsc{Sample}\bigl(S,\min(n-M,|S|)\bigr)\cup \textsc{Sample}(\xB_i,M)$\;\label{ln:haccept}
$W\gets W+|\xB_i|$\;
output $S$}
\end{algorithm}

\revision{\textbf{Batched Reservoir Sampling (B-RS):} The classical reservoir sampling algorithm~\cite{kn:sna,mb:ca} maintains a bounded uniform sample of items in a data stream. The idea is to fill up the reservoir with the first $n$ items, where $n$ is the reservoir size. For $k>n$, the $k$th incoming item is accepted into the sample with probability $q_k=n/k$, and an accepted item overwrites a randomly chosen victim. Our choice of $q_k$ is intuitively motivated by the observation that, in general, a given item from a population of size~$k$ appears in a uniform sample of size~$n\le k$ with probability precisely equal to $n/k$.

We can extend the classical algorithm to our batch setting; to our knowledge, a batch-oriented variant has not appeared previously in the literature. To informally motivate the algorithm, we generalize the foregoing intuition. For $k\ge 1$, let $\xU_k=\bigcup_{j=1}^k \xB_j$ be the set of items arriving up through time $t_k$ and set $W_k=|\xU_k|$. Suppose that the sample is full (i.e., $|S_{k-1}|=n$) just before batch $\xB_k$ arrives. After processing $\xB_k$, we ought to have a uniform sample of $n$ items from the set $\xU_k=\xU_{k-1}\cup\xB_k$. We would thus expect the number $M$ of $\xB_k$-items in the sample to follow a hypergeometric$(n,|\xB_k|,W_{k-1})$ distribution; here the hypergeometric$(k,a,b)$ probability mass function is given by $p(n)=\binom{a}{n}\binom{b}{k-n}/\binom{a+b}{k}$ if $\max(0,k-b)\le n\le\min(a,k)$ and $p(n)=0$ otherwise. This motivates us to accept new items from $\xB_k$ into the sample by first generating the number of items to accept as a hypergeometric variate $M$ and then selecting $M$ specific items for acceptance in a random and uniform manner. As with classical reservoir sampling, incoming items arriving before the sample fills up are accepted into the sample with probability~1 and do not overwrite random victims, whereas subsequent incoming items do overwrite random victims. In the corner case where $W_{k-1}<n$ and $W_{k-1}+|\xB_k|\ge n$, so that an incoming batch would cause the sample to overflow if all items were accepted, we generate $M$ as before, but $\min(n-W_{k-1},M)$ of these elements are accepted into the sample without overwriting a random victim, and the remainder overwrite a random victim from $\xU_{k-1}$.

The resulting sampling scheme is given as Algorithm~\ref{alg:rs}. In the algorithm, \textsc{Sample} is defined as before and $\textsc{HyperGeo}(k,a,b)$ returns a sample from the hypergeometric$(k,a,b)$ distribution; see \cite{StadloberZ99} for a discussion of efficient implementations of \textsc{HyperGeo}. Appendix~\ref{sec:batchproofs}
contains a formal proof of correctness. Although B-RS guarantees an upper bound on the sample size, it does not support time biasing in that all items seen so far are equally likely to be in the sample. The R-TBS algorithm (Section~\ref{sec:tbsamp}) maintains a bounded reservoir as in B-RS while simultaneously allowing time-biased sampling as in B-TBS.}

\section{Targeted-Size TBS}\label{sec:ttbs}

As a first step towards time-biased sampling with a controlled sample size, we provide the T-TBS scheme, which improves upon B-TBS by ensuring the inclusion property in \eqref{eq:expratio} while providing probabilistic guarantees on the sample size. Throughout, we focus on the case where the batch sizes $\{|\xB_k|\}_{k\ge 1}$ are independent and identically distributed (i.i.d.) with common mean $b<\infty$, and assume  that the decay function $f$ satisfies $\lim_{\alpha\to\infty}f(\alpha)=0$.

\subsection{The Algorithm}\label{sec:ttbsAlg}

The key idea is to not just downsample to remove older items as in B-TBS, but to also downsample incoming batches at a rate $q$ such that $n$ becomes (asymptotically) the ``equilibrium'' sample size. \revision{Unlike B-TBS, we now want to have the retention probability of an item depend on its age. In particular, if for $x\in \xB_i$, we set the retention probability at time $t_k\ge t_i$ equal to 
\begin{equation}\label{eq:defpik}
p_{i,k}=f(\alpha_{i,k})/f(\alpha_{i,k-1}),
\end{equation}
then we have
\begin{equation}\label{eq:gttbs1}
\prob{x\in S_k}
=\prob{x\in S_i}\times\prod_{j=1}^{k-i} \prob{x\in S_{i+j}\mid x\in S_{i+j-1}}
= q\times\prod_{j=1}^{k-i} \frac{f(\alpha_{i,i+j})}{f(\alpha_{i,i+j-1})}=qf(\alpha_{i,k}),
\end{equation}
and \eqref{eq:expratio} follows immediately from \eqref{eq:gttbs1}.

To choose $q$, we reason as follows. Suppose that the sample size equals the target value $n$ and we are about to process batch $\xB_k$. Prior to incrementing the ages and processing the arriving batch, the ages in the sample range from $\alpha_{0,k-1}=t_{k-1}$ down to $\alpha_{k-1,k-1}=0$, with an expected number $n\phi_{i,k-1}$ of sample items belonging to batch $\xB_i$, where $\phi_{i,k-1}=f(\alpha_{i,k-1})/\sum_{j=0}^{k-1}f(\alpha_{j,k-1})$. Thus the expected number of  $\xB_i$ items removed prior to processing $\xB_k$ is $n\phi_{i,k-1}(1-p_{i,k})$, where $p_{i,k}$ is defined in \eqref{eq:defpik}. Summing over all batches $\xB_i$, we find, after some algebra, that the expected total number of removed items is $n\gamma_k$, where 
\begin{equation}\label{eq:defgammak}
\gamma_k=\sum_{i=0}^{k-1}\phi_{i,k-1}(1-p_{i,k})=1-\frac{\sum_{i=0}^{k-1}f(\alpha_{i,k})}{\sum_{i=0}^{k-1}f(\alpha_{i,k-1})}
\end{equation}
for $k\ge 1$. On the other hand, the expected number of items entering the sample is $q_kb$ (where we initially allow $q$ to depend on $k$). For $n$ to be an equilibrium point, we equate the expected inflow and outflow and solve for $q_k$ to obtain $q_k=n\gamma_k/b$. If \begin{equation}\label{eq:gammaconverge}
\lim_{k\to\infty}\gamma_k=\gamma
\end{equation}
for some $\gamma\in (0,1]$, then $\lim_{k\to\infty}q_k= n\gamma/b$. In light of \eqref{eq:gttbs1}, we see that by setting $q=n\gamma/b$, we ensure that \eqref{eq:expratio} holds at all times, and that the sample size $n$ is asymptotically an equilibrium point as $k$, the number of batches processed, becomes large. (See Section~\ref{sec:properties} for a formal statement and proof.) Note that, even if we always accept all items in an arriving batch (i.e., $q=1$) but the resulting expected inflow $b$ is less than the expected outflow $n\gamma$, the sample will consistently fall below $n$, and so we require that $b\ge n\gamma$.

Because we are assuming that $t_i=i\Delta$ for $i\ge 1$, we can easily derive a necessary and sufficient condition for \eqref{eq:gammaconverge} to hold. Writing $f_i=f(i\Delta)$, we have $\gamma_k=(f_0-f_k)/F_{k-1}$, where $F_j=\sum_{i=0}^j f_i$ for $j\ge 0$. Thus, if 
$F_\infty=\sum_{i=0}^\infty f_i<\infty$,
then \eqref{eq:gammaconverge} holds with $\gamma= f_0/F_\infty=1/F_\infty$, so that $q=n/(bF_\infty)$. Necessity follows from the fact that $\lim_{k\to\infty}f_k=0$ by assumption. For polynomial decay with $f(\alpha)=1/(1+\alpha)^s$, we have $\gamma=\zeta(s,1/\Delta)/\Delta^s$, where $\zeta$ is the Hurwitz zeta function. If the decay function is exponential, i.e., $f(\alpha)=\exp(-\lambda \alpha)$, and we choose a time scale so that $\Delta=1$, then a simple calculation shows that $\gamma_k\equiv 1-\exp(-\lambda)$ and $p_{i,k}\equiv\exp(-\lambda)$ for $k\ge i\ge 1$, and we obtain the T-TBS algorithm for exponential decay as described in \cite{HentschelHT18}. Here $n$ is an equilibrium point for every value of $k$, and not merely in an asymptotic sense as $k\to\infty$. For this special case, we do not need to maintain the arrival timestamp for each item, and therefore do not need to partition the sample items based on arrival time. If we further assume that $q=1$, then we obtain the B-TBS algorithm as a special case in which the equilibrium sample size is $b/\gamma=b/(1-e^{-\lambda})$, which is completely determined by $b$ and $\lambda$. For complex functions $f$, we can compute $\gamma$ numerically.}

\begin{algorithm}[ht]
\caption{Targeted-size TBS (T-TBS)}\label{alg:targsampB}
{\footnotesize
$f$: decay function\;
$n$: target sample size\;
$b$: assumed mean batch size such that $b\ge n\gamma$\;
\BlankLine
Initialize: $S\gets \emptyset$; $\gamma=\textsc{Gamma}(f)$; $q\gets n\gamma/b$\;\label{ln:init}
\For{$k\gets 1,2,\ldots$}{
  \For(\Comment*[f]{update current sample items} ){$H_i\in S$}{
    $p=f(\alpha_{i,k})/f(\alpha_{i,k-1})$\Comment*[r]{retention probability}\label{ln:defp}
    $m \gets \textsc{Binomial}(|H_i|,p)$\Comment*[r]{simulate $|H_i|$ trials}\label{ln:binoma}
    \If{$m>0$}{
      $H_i\gets \textsc{Sample}(H_i,m)$\Comment*[r]{retain $m$ random elements}
    }
    \Else(\Comment*[f]{$H_i$ is now empty}){$S\gets S\setminus \{H_i\}$}
    }
 $l \gets \textsc{Binomial}(|\xB_k|,q)$\;\label{ln:binomb}
 $H_k \gets \textsc{Sample}(\xB_k,l)$\Comment*[r]{downsample new batch}
 \If(\Comment*[f]{insert new items}){$|H_k|>0$}{$S\gets S\cup \{H_k\}$}\label{ln:taccept}
 output $S$}
 }
\end{algorithm}

\revision{The resulting sampling scheme is given as Algorithm~\ref{alg:targsampB}; it precisely controls inclusion probabilities in accordance with \eqref{eq:expratio} while constantly pushing the sample size toward the target value $n$.} We represent a sample $S$ as a collection of sets $H_i$, where $H_i$ is the set of sample items that arrived at time $t_i$; thus $H_i\subseteq\xB_i$. The function $\textsc{Gamma}$ in line~\ref{ln:init} computes the constant $\gamma$ defined above. \revision{Conceptually,} at each time $t_k$, T-TBS first downsamples the current sample by independently flipping a coin for each item. The retention probability for an item depends on its age; specifically, an item $x\in\xB_i$ is retained with probability $p_{i,k}$.
T-TBS then downsamples the arriving batch $\xB_k$ via independent coin flips; an item in $\xB_k$ is inserted into the sample with probability $q$. As with B-TBS, the algorithm efficiently simulates multiple coin flips by directly generating the binomially distributed number of successes; thus the functions $\textsc{Binomial}(j,r)$ and $\textsc{Sample}(A,m)$ are defined as before.


\begin{remark}
The constraint that $b\ge n\gamma$ may lead to an inconveniently large required mean batch size. Intuitively, the problem is that an item's weight can become too small too quickly, even for polynomial decay. For instance, with $f_i=1/(1+i)^2$, all items lose three fourths of their weight going from age~0 to age~1. For subexponential decay functions $f$, we can deal with this issue by using a shifted decay function $f^{(d)}_i=f_{d+i}/f_d$, where $d$ is a positive integer. By choosing $d$ sufficiently large, the corresponding value of $\gamma^{(d)}=f^{(d)}_0/\sum_{i=0}^\infty f^{(d)}_i=f_d/\sum_{i=d}^\infty f(i)$ can be made as small as desired. For instance, with $\Delta=1$ and $f_i=1/(1+i)^2$, we have $\gamma\approx 0.61$, whereas $\gamma^{(3)}\approx 0.22$. Of course, the original constraint requiring that $\prob{x\in S_k}/\prob{y\in S_k}=f_i/f_j$ for $x\in\xB_i$ and $y\in\xB_j$ is now modified to require that $\prob{x\in S_k}/\prob{y\in S_k}=f_{d+i}/f_{d+j}$, so that the relative inclusion probabilities have essentially the same ``tail behavior'' as $f$ for large $i$ and $j$, but the initial decay rate will be slower. This trick will not work for exponential decay, because here $f^{(d)}_i=e^{-\lambda d}f_i$, which implies that $\gamma^{(d)}\equiv\gamma$ for $d\ge 1$. In this case we must select $\lambda$ small enough to accommodate the mean batch size. Thus non-exponential decay functions allow an additional degree of freedom when parameterizing the sampling algorithm. For superexponential decay, shifting will actually increase $\gamma$ but, as discussed previously, such decay functions are of less practical interest. Over a broad range of experiments, quadratic decay with a shift of $d=10$ yielded superior ML robustness results for both T-TBS and R-TBS, and we often use this variant in our experiments (Section~\ref{sec:exp}).
\end{remark}

\subsection{Sample-Size Properties}\label{sec:properties}

\revision{We now analyze the sample size behavior of T-TBS, which directly impacts memory requirements, efficiency of memory usage, and ML model retraining time. We continue to assume that the batch sizes $\{|\xB_k|\}_{k\ge 1}$ are i.i.d.\ with common mean $b\in [n\gamma,\infty)$. Our first result (Theorem~\ref{th:recurr}) describes the probabilistic behavior of the sample size $C_k=\sum_{H_i\in S_k} |H_i|$ for a fixed time $t_k$. Specifically, we give approximate expressions for the mean and variance of $C_k$ when $k$ is large. We also use Hoeffding's inequality to give exponential bounds valid for any $k$, showing that the probability of a very large deviation above or below the target value~$n$ at any given time $t_k$ is very low.  The proof of the theorem (and of most other results in the paper) is given in Appendix~\ref{sec:proofs}. Denote by $\bbar\ge 1$ the maximum possible batch size, so that $\prob{B\le\bbar}=1$. Recall that $F_k\triangleq\sum_{i=0}^k f_i$, and set $F^{(2)}_k=\sum_{i=0}^kf^2_i$.}

\begin{theorem}\label{th:recurr}
For any decay function $f$ such that $F_\infty<\infty$,
\begin{enumerate}\parskip=0pt
\item[(i)] $\mean[C_k]=nF_{k-1}/F_\infty\uparrow n$ as $k\to\infty$;
\item[(ii)] $\var[C_k]\to bqF_\infty - bq^2F^{(2)}_\infty$;
\item[(iii)] if $\bbar<\infty$, then
\begin{enumerate}
\item[(a)] $\prob{C_k\ge (1+\eps)n}\le e^{-O(kn^2\eps^2)}$ for $\eps,k>0$ 
and
\item[(b)] $\prob{C_k \le (1-\eps)n}\le e^{-O(kn^2)}$ for $\eps\in(0,1)$ and sufficiently large $k$.
\end{enumerate}
\end{enumerate}
\end{theorem}
Thus, from (i), the expected sample size converges to the target size $n$  as $t$ becomes large and, from~(ii), the variance also converges to a constant that depends on $b$ and $f$. By (iii), the probability that the sample size deviates from $n$ by more than $100\eps\%$ is exponentially small when $k$ or $n$ is large.

\begin{remark}
If $f_i$ decays very slowly as $i\to\infty$, then the convergence of the expected sample size to $n$ will also  be very slow. For example, if $f_i=1/(1+i\Delta)^s$ for some $s>1$, then, using (i) above and a standard bound,  it is easy to show that $n-\mean[C_k]=\Theta(1/k^{s-1})$. Thus choosing a value of, say, $s=1.0001$ will result in a long sequence of undersized samples. Similarly, if the sample size becomes overly large at some point, recovery will be slow.
\end{remark}


\revision{Theorem~\ref{th:recurr} does not tell the entire story. Although it follows from this theorem that, over many different sampling runs, the average sample size at a given  (large) time~$t_k$ is close to $n$ and the probability of being far away from $n$ is small, the successive sample sizes during an individual sampling run need not be well behaved. This issue is addressed by Theorem~\ref{th:recurGam} below. Assertion~(i) shows that any sample size can be attained with positive probability, so one potential type of bad behavior might occur if, with positive probability, the sample size is unstable in that it drifts off to $+\infty$ over time. Assertion~(ii) shows that such unstable behavior is ruled out if the maximum batch size is bounded and $f_i$ decays rapidly enough so that equation~\eqref{eq:sumifi} below holds. If $f_i$ decays even faster, so that equation~\eqref{eq:supratio} below holds, then the stability assertion can be strengthened to guarantee that the times between successive attainments of a given sample size are not only all finite, but all have the same finite mean; moreover, the average sample size---averaged over times $t_1,t_2,\ldots,t_k$---converges to $n$ with probability~1 as $k$ becomes large. On the negative side, it follows that, for a given sampling run, the sample size will repeatedly---though infrequently, since the expected sample size at any time point is finite---become arbitrarily large, even if the average behavior is good. This result shows that the sample-size control provided by T-TBS is incomplete, and thus motivates the more complex R-TBS algorithm.}

In the following, write ``i.o.'' to denote that an event occurs ``infinitely often'', i.e., for infinitely many values of $k$, and write ``w.p.1'' for ``with probability~1''. 
\begin{theorem}\label{th:recurGam}
Let the T-TBS decay function $f$ satisfy $F_\infty<\infty$ and let $\bbar$ be the maximum possible batch size. Then
\begin{enumerate}
\item[(i)] for all $m\ge 0$, there exists $k\ge 0$ such that $\prob{C_k\ge m}>0$;
\item[(ii)] if $\bbar<\infty$ and 
\begin{equation}\label{eq:sumifi}
\sum_{i=0}^\infty if_i<\infty,
\end{equation}
then $\prob{C_k=m\text{ i.o.}}=1$ for all $m\ge 0$; 
\item[(iii)] if $\bbar<\infty$ and, for $k\ge 1$,
\begin{equation}\label{eq:supratio}
\sup_{i\ge 0}(f_{i+k}/f_i)\le g_k
\end{equation}
for some sequence $\{g_k\}_{k\ge 0}$ with $\lim_{k\to\infty}g_k=0$, then (a) the expected times between successive visits to state~$m$ are uniformly bounded for any $m\ge 0$, and (b)
$\lim_{k\to\infty}(1/k)\sum_{i=0}^k C_i=n\text{ w.p.1}$.
\end{enumerate}
\end{theorem}

\revision{The proof of Theorem~\ref{th:recurGam} rests on a reduced representation $\xS_k$ of the state $S_k$ of the sample at a time $t_k$, comprising a collection of pairs of the form $(n,i)$, where $n$ is the number of sample items of age~$\Delta i$. In Appendix~\ref{sec:proofs} we argue that the process $\{\xS_k\}_{k\ge 0}$ is a time-homogeneous Markov chain, and hence we can use tools from the theory of Markov chains to establish the ``recurrence'' properties that correspond to our stability results. The state space of this Markov chain is quite complex, as are the transition probabilities between states, so the application of these tools is decidedly nontrivial.}

\begin{remark}
Note that the assumptions on $f$ indeed become increasingly strong when going from Assertions~(i) to (iii). The condition in \eqref{eq:sumifi} trivially implies that $F_\infty<\infty$. Also, \eqref{eq:supratio} implies \eqref{eq:sumifi}. To see this, fix $k$ large enough so that $g_k<1$, and observe that, since $f_{i+k}\le g_kf_i$ for all $i$,
\[
\sum_{i=0}^\infty if_i\le\sum_{m=0}^{k-1}f_m\biggl(\sum_{i=0}^\infty (ki+m)g_k^i\biggr)<\infty.
\]
This increase in strength is strict: the decay function $f_i=1/(i\Delta+1)^2$ satisfies $F_\infty<\infty$ but not \eqref{eq:sumifi}, and the decay function $f_i=1/(i\Delta+1)^3$ satisfies \eqref{eq:sumifi} but not \eqref{eq:supratio}. The condition in \eqref{eq:sumifi}  holds, e.g., for exponential decay and for polynomial decay $f_i=1/(1+i\Delta)^s$ with $s>2$.  The condition in \eqref{eq:supratio}  holds, e.g., for functions that decay exponentially or faster.
\end{remark}

Even in the most stable case, however, we do not have complete control over the sample size. Indeed, any sample size~$m$, no matter how large, is exceeded infinitely often w.p.1 and the expected time between such incidents is uniformly bounded. Although the expected times are often very large, so that the incidents are infrequent, and the faster the decay, the faster the recovery from an incident, T-TBS is ultimately fragile with respect to sample size. This fragility is amplified when batch sizes fluctuate in a non-predicable way, as often happens in practice, and T-TBS can break down; see the experiments in Section~\ref{sec:samplesize}.

Despite the fluctuations in sample size, T-TBS is of interest because, when the mean batch size is known and constant over time, and when some sample overflows are tolerable, T-TBS is relatively simple to implement and parallelize, and is very fast (see Section~\ref{sec:exp}). For example, if the data comes from periodic polling of a set of robust sensors, the data arrival rate will be known a priori and will be relatively constant, except for the occasional sensor failure, and hence T-TBS might be appropriate. On the other hand, if data is coming from, e.g., a social network, then batch sizes may be hard to predict.


\section{Reservoir-Based TBS}\label{sec:tbsamp}

Targeted time-biased sampling (T-TBS) controls the decay rate but only partially controls the sample size, whereas batched reservoir sampling (B-RS) bounds the sample size but does not allow time biasing. Our new reservoir-based time-biased sampling algorithm (R-TBS) combines the best features of both, controlling the decay rate while ensuring that the sample never overflows. For exponential decay, R-TBS has optimal sample size and stability properties, and in the general case the user can trade off storage for both sample-size stability and accuracy. Importantly, unlike T-TBS, the R-TBS algorithm can handle any sequence of batch sizes.

\subsection{Item Weights and Latent Samples}\label{sec:latent}

R-TBS combines the use of a reservoir with the notion of ``latent samples" to enforce \eqref{eq:expratio} and bound the sample size.  Latent samples, in turn, rest upon the notion of ``item weights".

\textbf{Item weights:} In R-TBS, the \emph{weight} of an item of age $\alpha$ is given by $f(\alpha)$, where $f$ is the decay function; note that a newly arrived item has a weight of $f(0)=1$. As discussed later, R-TBS ensures that the probability that an item appears in the sample is proportional (or approximately proportional) to its weight. All items arriving at the same time have the same weight, so that the \emph{total weight} of all items seen up through time~$t_k$ is $W_k=\sum_{i=1}^k |\xB_i|f(\alpha_{i,k})$. \revision{For traditional (sequential or batch-oriented) reservoir sampling, an item does not decay, and so has a weight equal to 1 at all times; thus the notions of items and item weights coincide. Moreover, in the traditional setting the weight of a sample coincides with the number of items in the sample. In our generalized setting, the ``size'' (weight) of a sample and the number of items in the sample differ, with samples having fractional sizes. We handle this complication via the notion of a \emph{latent fractional sample}.}

\begin{algorithm}[tbh]
\caption{Generating a sample from a latent sample}\label{alg:getSample}
{\footnotesize
$L=(A,\pi,C)$: latent sample\;
\BlankLine
$U\gets\textsc{Uniform}()$\;
\leIf{$U\le\frc(C)$}{$S\gets A\cup\pi$}{$S\gets A$}
\Return $S$
}
\end{algorithm}

\textbf{Latent samples:} A latent fractional sample  formalizes the idea of a sample of fractional size. Formally, given a set $U$ of items, a \emph{latent sample} of $U$ with \emph{sample weight} $C$ is a triple $L=(A,\pi,C)$, where $A\subseteq U$ is a set of $\floor{C}$ \emph{full} items and $\pi\subseteq U$ is a (possibly empty) set containing at most one \emph{partial} item; $\pi$ is nonempty if and only if $C>\floor{C}$.

\begin{figure}[tbh]
       \centering
       \includegraphics[width=0.3\linewidth]{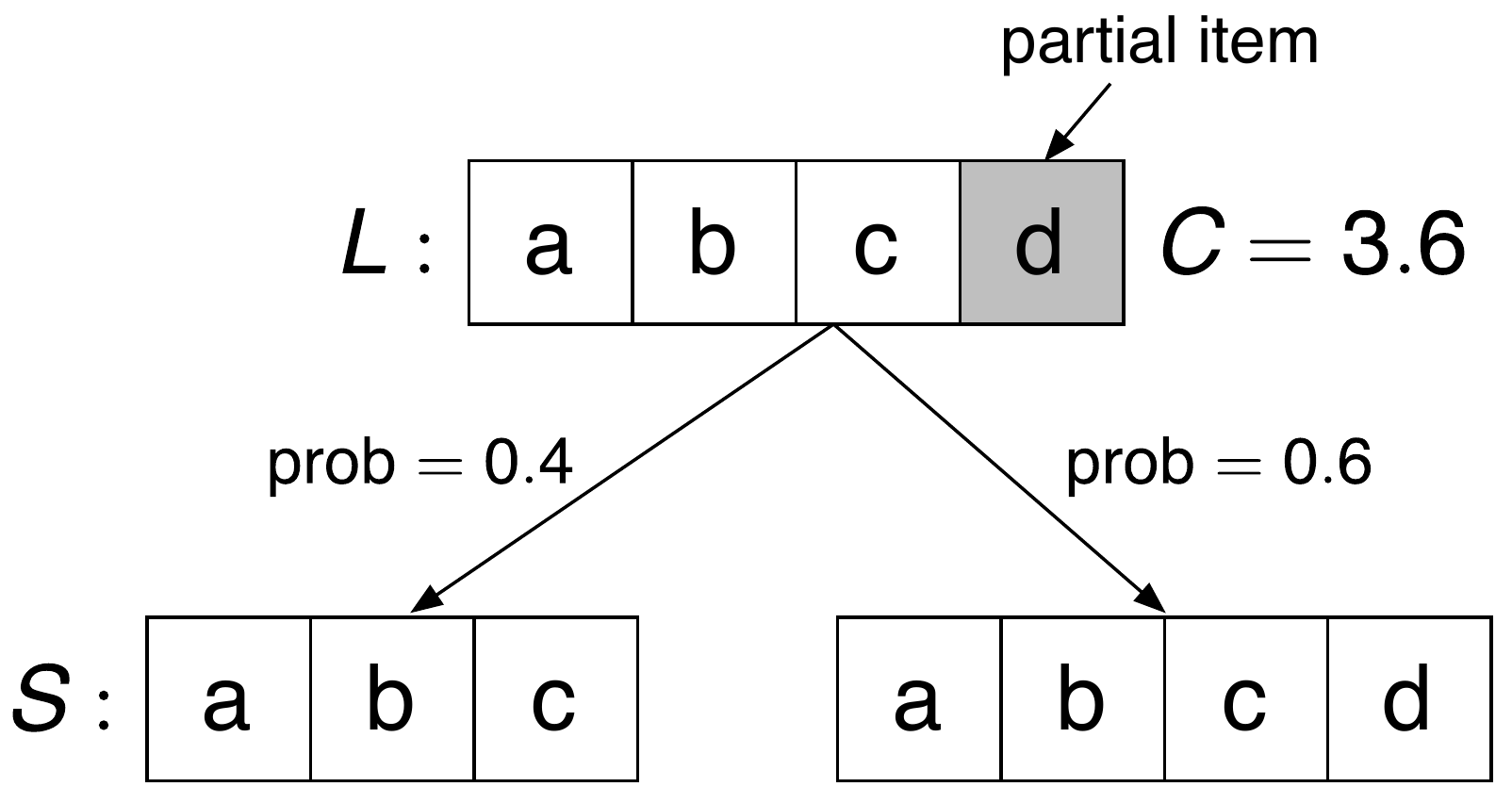}
       \caption{Latent sample $L$ (sample weight $C=3.6$) and possible realized samples}
       \label{fig:latentsamp}
\end{figure}

\revision{R-TBS maintains a latent sample $L$ over time and produces an actual sample $S$ from $L$ on demand by sampling as described in Algorithm~\ref{alg:getSample}}; see Figure~\ref{fig:latentsamp} for an example. In the pseudocode, $\frc(x)=x-\floor{x}$ and the function $\textsc{Uniform}()$ generates a random number uniformly distributed on $[0,1]$. Because each full item is included with probability~1 and the partial item is included with probability $\frc(C)$, we have
\begin{equation}\label{eq:meansize}
\mean[|S|]=\ceil{C}\frc(C)+\floor{C}\bigl(1-\frc(C)\bigr)
=(\ceil{C}-\floor{C})\frc(C)+\floor{C}
=\frc(C)+\floor{C}=C,
\end{equation}
so that the size of $S$ equals $C$ in expectation. By allowing at most one partial item, we minimize the latent sample's footprint: $|A\cup\pi|\le\floor{C}+1$. Importantly, if the weight $C$ of a latent sample $L$ is an integer, then $L$ contains no partial item, and the sample $S$ generated from $L$ via Algorithm~\ref{alg:getSample} is unique and contains exactly $C$ items; \revision{thus, the sample weight and the number of sample items coincide in this case. We now describe two key operations on latent samples that are used by R-TBS.}

\textbf{Downsampling:} Besides extracting an actual sample from a latent sample, another key operation on latent samples is \emph{downsampling}. For $\theta\in(0,1)$, the goal of downsampling $L=(A,\pi,C)$ is to obtain an new latent sample $L'=(A',\pi',\theta C)$ such that, if we generate $S$ and $S'$ from $C$ and $C'$ via Algorithm~\ref{alg:getSample}, we have
\begin{equation}\label{eq:downsample}
\prob{x\in S'}=\theta\prob{x\in S}
\end{equation}
for all $x\in S$. \revision{Thus the appearance probability for each item in $S$, as well as the expected size of the sample $\mean[S]$, is scaled down by a factor of $\theta$.} Theorem~\ref{th:downsamp} (later in this section) asserts that Algorithm~\ref{alg:downsamp} satisfies this property. 

\begin{algorithm}[t]
\caption{Downsampling}\label{alg:downsamp}
{\footnotesize
$L=(A,\pi,C)$: input latent sample\;
$\theta$: scaling factor with $\theta\in(0,1)$\;
\BlankLine
$U\gets \textsc{Uniform}()$; $C'=\theta C$\;
\uIf(\Comment*[f]{no full items retained}){$\floor{C'}=0$\label{ln:no_full}}{
    \If{$U>\frc(C)/C$}{$(A',\pi')\gets\textsc{Swap1}(A,\pi)$}\label{ln:swap0}
    $A'\gets\emptyset$\;\label{ln:killA}
}
\uElseIf(\Comment*[f]{no items deleted}){$0<\floor{C'}=\floor{C}$\label{ln:no_delete}}{
    \If{$U>\bigl(1-\theta\frc(C)\bigr)/\bigl(1-\frc(C')\bigr)$\label{ln:no_swap}}{
        $(A',\pi') \gets\textsc{Swap1}(A,\pi)$\;\label{ln:convert}
        }
}
\Else(\Comment*[f]{items deleted: $0<\floor{C'}<\floor{C}$}){
    \eIf{$U\le \theta\frc(C)$\label{ln:normal}}{
        $A' \gets \textsc{Sample}(A,\floor{C'})$\;\label{ln:swapA}
        $(A',\pi')\gets\textsc{Swap1}(A',\pi)$\;\label{ln:swapB}
        }{
        $A' \gets \textsc{Sample}(A,\floor{C'}+1)$\;\label{ln:smpl}
        $(A',\pi')\gets\textsc{Move1}(A',\pi)$\;\label{ln:move}
        }      
}
\If(\Comment*[f]{no fractional item}){$C'=\floor{C'}$}{
    $\pi'\gets\emptyset$\;
    }
}
\Return $L'=(A',\pi',C')$
\end{algorithm}

In the pseudocode for Algorithm~\ref{alg:downsamp}, the subroutine $\textsc{Swap1}(A,\pi)$ moves a randomly selected item from $A$ to $\pi$ and moves the current item in $\pi$ (if any) to $A$. Similarly, $\textsc{Move1}(A,\pi)$ moves a randomly selected item from $A$ to $\pi$, replacing the current item in $\pi$ (if any). More precisely, $\textsc{Swap1}(A,\pi)$ executes the operations $I\gets \textsc{Sample}(A,1)$, $A\gets (A\setminus I)\cup\pi$, and $\pi\gets I$,
and $\textsc{Move1}(A,\pi)$ executes the operations $I\gets \textsc{Sample}(A,1)$, $A\gets A\setminus I$, and $\pi\gets I$.

\begin{figure*}[tbh]
        \centering
	\subfigure[$C=3\rightarrow C'=1.5$]{
	   \label{fig:downsamp1}\includegraphics[width=0.45\linewidth]{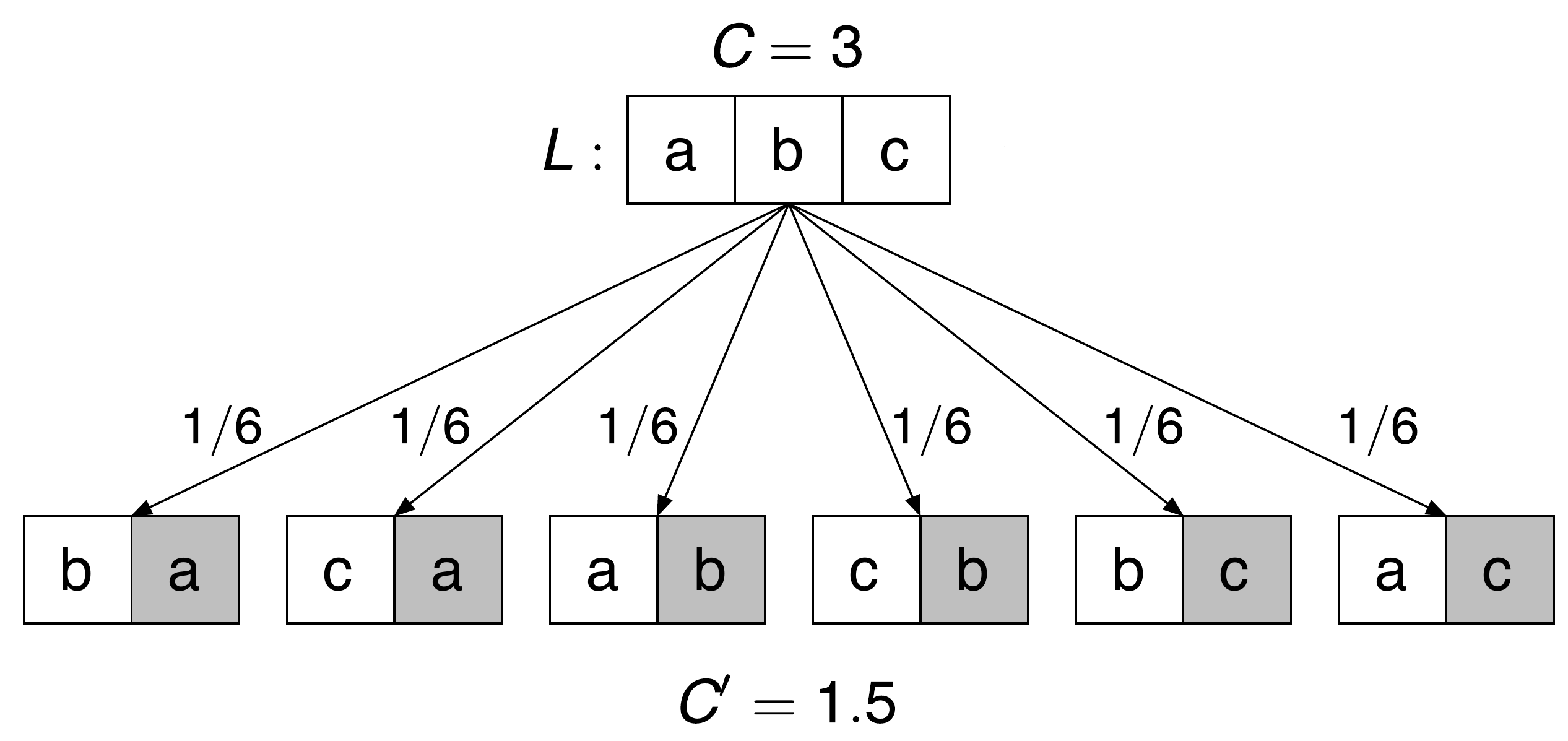}} \hspace{0.02\textwidth}
	\subfigure[$C=3.2\rightarrow C'=1.6$]{
	   \label{fig:downsamp2}\includegraphics[width=0.45\linewidth]{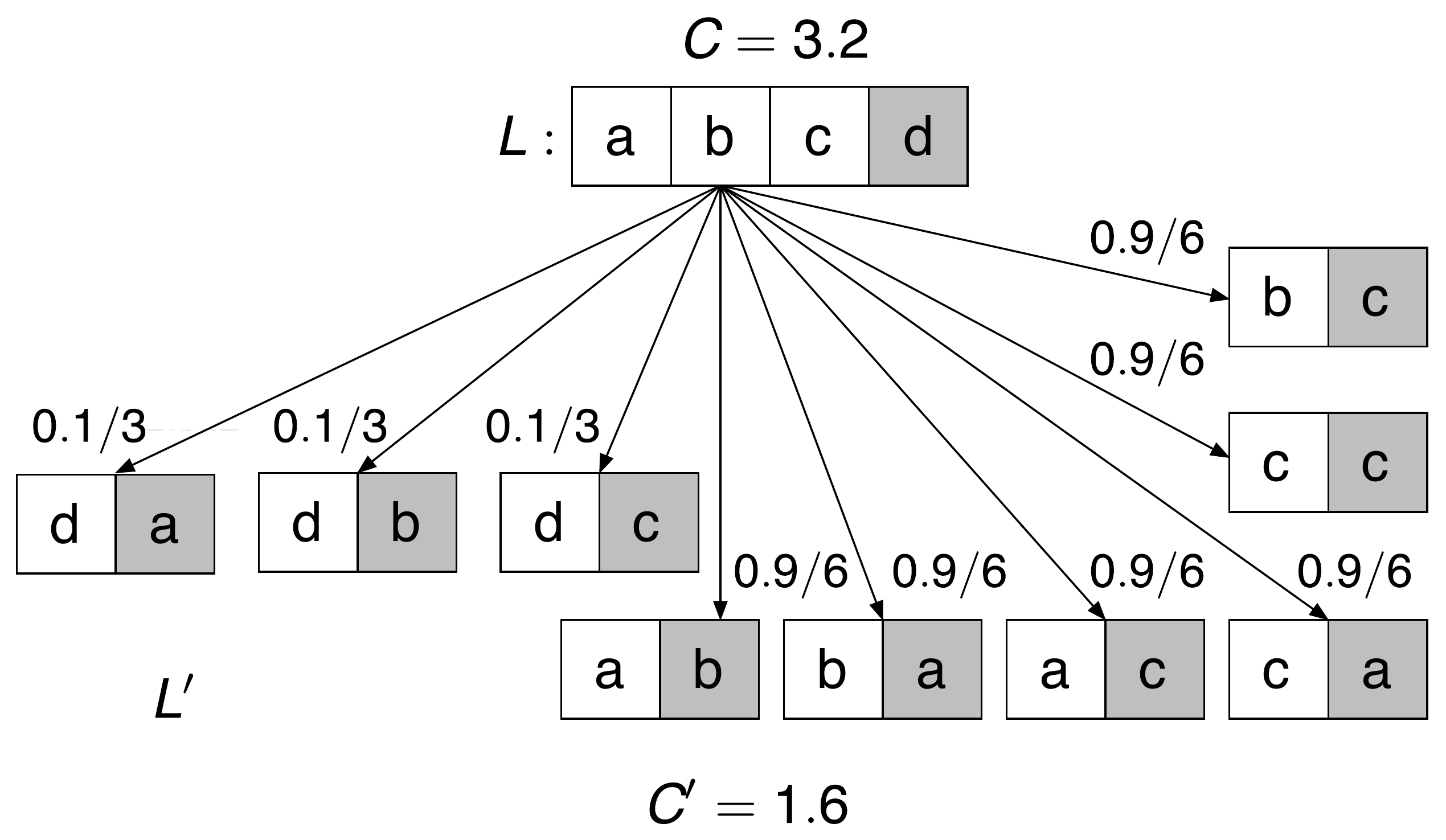}} \hspace{0.02\textwidth}
	\subfigure[$C=2.4\rightarrow C'=0.4$]{
	   \label{fig:downsamp3}\includegraphics[width=0.2\linewidth]{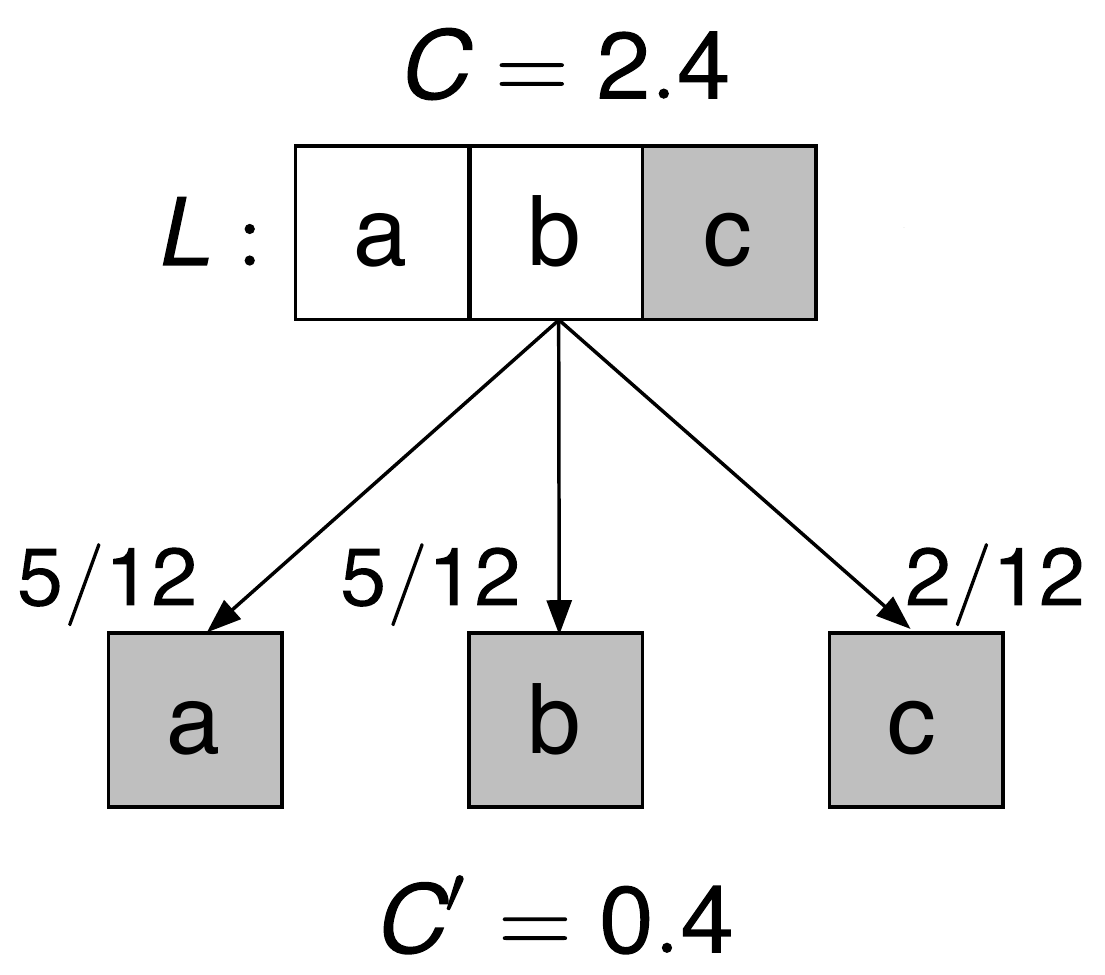}} \hspace{0.02\textwidth}
	\subfigure[$C=2.4\rightarrow C'=2.1$]{
	   \label{fig:downsamp4}\includegraphics[width=0.35\linewidth]{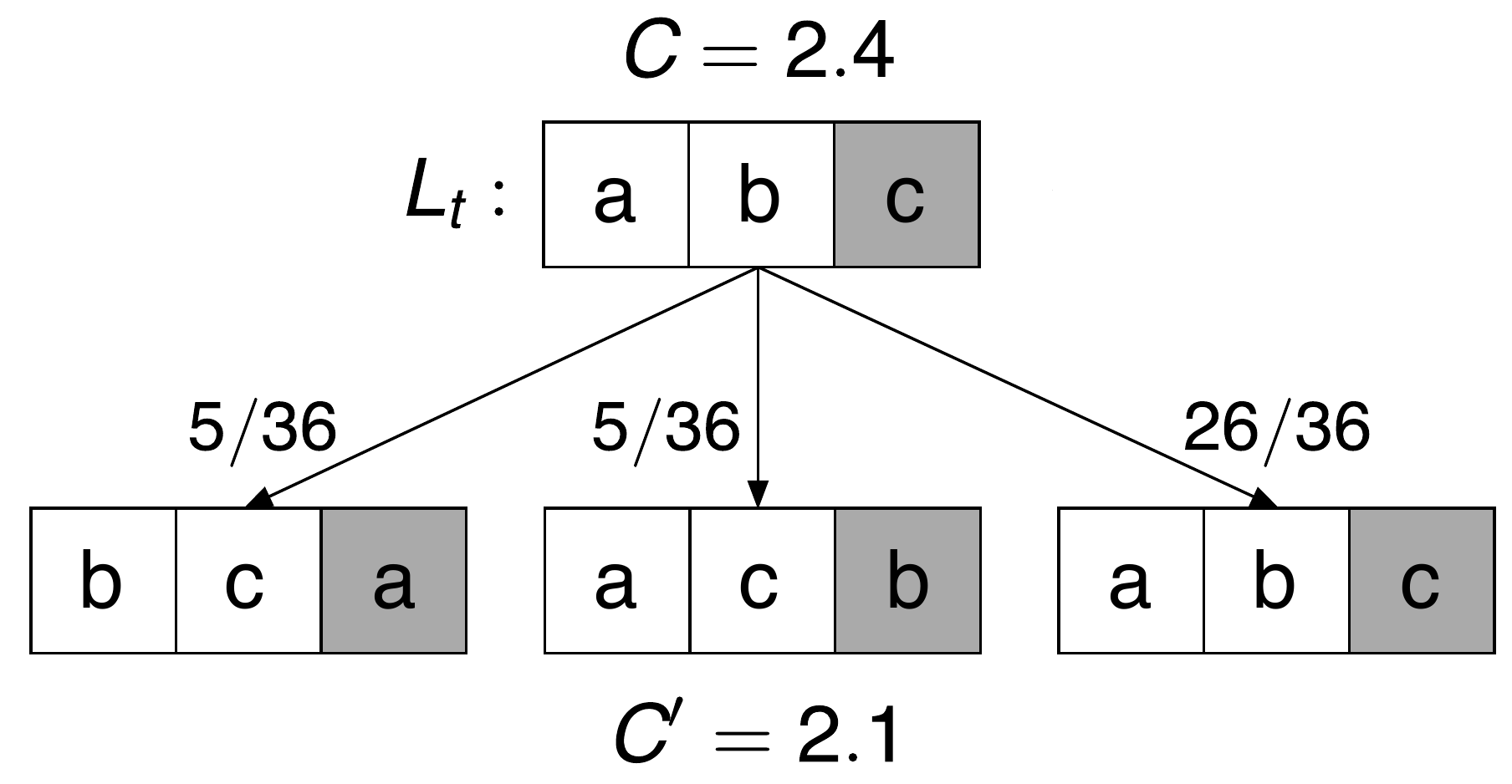}} 
\caption{\label{fig:downsampling}Downsampling examples}
\end{figure*}

To gain some intuition for why the algorithm works, consider a simple special case, where the goal is to form a latent sample $L' = (A',\pi', \theta C)$ from a latent sample $L=(A,\pi,C)$ of integral size $C$; that is, $L$ comprises exactly $C$ full items. Assume that $C'=\theta C$ is non-integral, so that $L'$ contains a partial item, and that $\floor{C}>\floor{C'}$; e.g., $C=3$ and $C'=1.5$, so that $\theta=0.5$. In this case, we simply select an item at random (from $A$) to be the partial item in $L'$ and then select $\floor{C'}$ of the remaining $C -1$ items at random to be the full items in $L'$; see Figure~\ref{fig:downsamp1}. Denote by $S$ and $S'$ the samples obtained from $L$ and $L'$ via Algorithm~\ref{alg:getSample}. By symmetry, each item $i\in L$ is equally likely to be included in $S'$, so that the inclusion probabilities for the items in $L$ are all scaled down by the same fraction, as required for \eqref{eq:downsample}. In Figure~\ref{fig:downsamp1}, for example, item $a$ appears in $S$ with probability~1 since it is a full item. In $S'$, where the weights have been reduced by 50\%, item $a$ (either as a full or partial item, depending on the random outcome) appears with probability $2\cdot(1/6)+2\cdot(1/6)\cdot 0.5=0.5$, as expected. This scenario corresponds to lines~\ref{ln:smpl} and \ref{ln:move} in the algorithm, where we carry out the above selections by randomly sampling $\floor{C'}+1$ items from $A$ to form $A'$ and then choosing a random item in $A'$ as the partial item by moving it to $\pi$.

In the case where $L$ contains a partial item~$x^*$ that appears in $S$ with probability $\frc(C)$, it follows from \eqref{eq:downsample} that $x^*$ should appear in $S'$ with probability $p=\theta P[x^*\in S]=\theta \frc(C)$. Thus, with probability~$p$, lines~\ref{ln:swapA}--\ref{ln:swapB} retain $x^*$ and convert it to a full item so that it appears in $S'$. Otherwise, in lines~\ref{ln:smpl}--\ref{ln:move}, $x^*$ is removed from the sample when it is overwritten by a random item from $A'$; see Figure~\ref{fig:downsamp2}. Again, a new partial item is chosen from $A$ in a random manner to uniformly scale down the inclusion probabilities. For instance, in Figure~\ref{fig:downsamp2}, item $d$ appears in $S$ with probability 0.2 (because it is a partial item) and in $S'$, appears with probability $3\cdot (0.1/3)=0.1$. Similarly, item $a$ appears in $S$ with probability~1 and in $S'$ with probability $(1.8)/6+0.6\cdot (1.8/6)+0.6\cdot(0.1/3)=0.5$.

The if-statement in line~\ref{ln:no_full} corresponds to the corner case in which $L'$ does not contain a full item. The partial item $x^*\in L$ either becomes full or is swapped into $A'$ and then immediately ejected; see Figure~\ref{fig:downsamp3}.

The if-statement in line~\ref{ln:no_delete} corresponds to the case in which no items are deleted from the latent sample, e.g., when $C=4.7$ and $C'=4.2$. In this case, $x^*$ either becomes full by being swapped into $A'$ or remains as the partial item for $L'$. Denoting by $\rho$ the probability of \emph{not} swapping, we have $P[x^*\in S'] = \rho\cdot\frc(C') + (1-\rho)\cdot 1$. On the other hand, \eqref{eq:downsample} implies that $P[x^*\in S']=\theta\frc(C)$. Equating these expression shows that $\rho$ must equal the expression on the right side of the inequality on line~\ref{ln:no_swap}; see Figure~\ref{fig:downsamp4}.

\begin{theorem}\label{th:downsamp}
For $\theta\in(0,1)$, let $L'=(A',\pi',\theta C)$ be the latent sample produced from a latent sample~$L=(A,\pi,C)$ via Algorithm~\ref{alg:downsamp}, and let $S'$ and $S$ be samples produced from $L'$ and $L$ via Algorithm~\ref{alg:getSample}. Then $\prob{x\in S'}=\theta\prob{x\in S}$ for all $x\in A\cup\pi$.
\end{theorem}

\textbf{The union operator:} We also need to take the union of disjoint latent samples while preserving the inclusion probabilities for each. Two latent samples $L_1=(A_1,\pi_1,C_1)$ and $L_2=(A_2,\pi_2,C_2)$ are \emph{disjoint} if $(A_1\cup\pi_1)\cap(A_2\cup\pi_2)=\emptyset$. The pseudocode for the union operation is given as Algorithm~\ref{alg:union}. The idea is to add all full items to the combined latent sample. If there are partials items in $L_1$ and $L_2$, then we transform them to either a single partial item, a full item, or a full plus partial item, depending on the values of $\frc(C_1)$ and $\frc(C_2)$. Such transformations are done in a manner that preserves the appearance probabilities. Of course, we can obtain the union of an arbitrary number of latent samples by iterating Algorithm~\ref{alg:union}; for latent samples $L_1,\ldots,L_k$, we denote by $\bigcup_{j=1}^k L_j$ the latent sample produced by this procedure. 

\begin{algorithm}[ht]
\caption{Union}\label{alg:union}
{\footnotesize
$L_1=(A_{1}, \pi_{1}, C_{1})$: fractional sample of size $C_{1}$\\
$L_2= (A_{2}, \pi_{2}, C_{2})$: fractional sample of size $C_{2}$\\
\BlankLine
$C \gets C_1 + C_2$\;\label{ln:plusC}
$U\gets \textsc{Uniform()}$\;
\uIf{$\frc(C_1) + \frc(C_2) < 1$}{
	$A \gets A_{1} \cup A_{2}$ \;\label{ln:unionA1}
	\leIf{$U\le\frc(C_1)/\bigl(\frc(C_1) + \frc(C_2)\bigr)$}{
		$\pi \gets \pi_1$
	}{
		$\pi\gets \pi_2$
	}
}
\uElseIf{$\frc(C_1) + \frc(C_2) = 1$}{
     $\pi\gets\emptyset$\;
     \leIf{$U\le\frc(C_1)$}{
          $A\gets A_1\cup A_2\cup\pi_1$
     }{
          $A\gets A_1\cup A_2\cup\pi_2$
     }
}
\Else(\Comment*[h]{$\frc(C_1) + \frc(C_2) >1$}){
	\eIf{$U\le\bigl(1 - \frc(C_1)\bigr)\bigm/\bigl[\bigl(1 - \frc(C_1)\bigr) + \bigl(1 - \frc(C_2)\bigr)\bigr]$}{
		$\pi = \pi_1$\;
		$A \gets A_{1} \cup A_{2} \cup \pi_2$ \;\label{ln:unionA2}
	}{
		$\pi = \pi_2$\;
		$A \gets A_1 \cup A_2 \cup \pi_1$ \;\label{ln:unionA3}
	}	
}
\Return L=(A,$\pi$, C)
}
\end{algorithm}

\begin{theorem}\label{th:union}
Let  $L_{1} = (A_1, \pi_1, C_1)$ and $L_2 = (A_2, \pi_2, C_2)$, be disjoint latent samples, and let $L=(A,\pi,C)$ be the latent sample produced from $L_1$ and $L_2$ by Algorithm~\ref{alg:union}.  Let  $S_1$, $S_2$, and $S$ be random samples generated from  $L_1$, and $L_2$, and $L$ via Algorithm~\ref{alg:getSample}. Then
\begin{romlist}
\item $C=C_1+C_2=\mean[S]$;
\item $\forall x \in L_1$, $\prob{x\in S}=\prob{x\in S_1}$; and
\item $\forall x \in L_2$, $\prob{x \in S} = \prob{x \in S_2}$.
\end{romlist}
\end{theorem}

\subsection{The R-TBS Algorithm with Exponential Decay}\label{sec:rtbsAlgE}

Our general goal is to provide a sampling algorithm that bounds the sample size at $n$ while enforcing \eqref{eq:expratio}. For the special case of exponential decay, this task is greatly facilitated by the fact that, at each time step, all items in the sample decay by the same multiplicative factor. We exploit this fact to provide a relatively simple version of R-TBS for the case of exponential decay. In Section~\ref{sec:rtbsAlg}, we show how to generalize our approach to the case of arbitrary decay functions.

\textbf{The algorithm:} R-TBS for exponential decay is given as Algorithm~\ref{alg:rtbsE}. The algorithm generates a sequence of latent samples $\{L_k\}_{k\ge 1}$ and from these generates a sequence of actual samples $\{S_k\}_{k\ge 1}$ that are returned to the user. In the algorithm, the functions \textsc{Getsample}, \textsc{Downsample}, and \textsc{Union} execute the operations described in Algorithms~\ref{alg:getSample}, \ref{alg:downsamp}, and \ref{alg:union}.

The goal of the algorithm is to ensure that 
\begin{equation}\label{eq:inclRTBS}
\prob{x\in S_k}=\rho_k f(\alpha_{i,k})
\end{equation}
for all $k\ge 1$, $i\le k$, and $x\in\xB_i$, where $f(\alpha)=e^{-\lambda \alpha}$ and $\{\rho_k\}_{k\ge 1}$ are the successive values of the variable $\rho$ during a run of the algorithm. Clearly, \eqref{eq:inclRTBS} immediately implies \eqref{eq:expratio}. We choose $\rho_k$ to make the sample size as large as possible without exceeding $n$. In more detail, we show in Theorem~\ref{th:rtbsIncl} below that $C_k=\rho_kW_k$ for all $k$. We therefore set $\rho_k=\min(1,n/W_k)$---see line~\ref{ln:RTBSdefRho}---so that $C_k=\min(W_k,n)$. Thus if $W_k<n$, then the sample weight is at its maximum possible value $W_k$, leading to the maximum possible sample size of either $\floor{W_k}$ or $\ceil{W_k}$. If $W_k\ge n$, then the sample weight, and hence the sample size, is capped at $n$. The algorithm functions analogously to classic reservoir sampling: if the (weighted) items seen so far can fit into the reservoir of size $n$, then they are simply accepted, if the total item weight exceeds $n$, then, when a new batch arrives, a random subset of old items is removed from the sample via downsampling (line~\ref{ln:dsampleE}) and a random subset of the arriving items, also filtered via downsampling (line~\ref{ln:dsampleBE}), take their place (line~\ref{ln:unionE}). \revision{Note that if $|\xB_j|\equiv 1$ for all $j$, so that we process items one at a time, and if there is no decay, so that $f(\alpha)\equiv 1$, then $W_k=k$ and the inclusion probability in \eqref{eq:inclRTBS} reduces to  $\prob{x\in S_k}=\rho_k\cdot 1=q_k$, where $q_k=\min(1,n/k)$, exactly as in traditional reservoir sampling.}

\begin{algorithm}[t]
\caption{Reservoir-based TBS (R-TBS) for exponential decay}\label{alg:rtbsE}
{\footnotesize
$\lambda$: decay rate ($\ge 0$)\;
$n$: maximum sample size\;
\BlankLine
Initialize: $W\gets 0$; $A\gets \emptyset$; $\pi\gets\emptyset$; $C\gets 0$; $t_0\gets0$; $\rho\gets 1$\;
\For{$k\gets1,2,\ldots$}{
$\theta\gets e^{-\lambda(t_k-t_{k-1})}$\Comment*[r]{decay factor}
$W\gets \theta W+|\xB_k|$\Comment*[r]{update total weight}
$\rho'\gets \rho$\;
$\rho\gets \min(1,n/W)$\;\label{ln:RTBSdefRho}
\lIf(\Comment*[f]{decay old items}){C>0}{$(A,\pi,C)\gets\textsc{Downsample}\bigl((A,\pi,C),(\rho/\rho')\theta\bigr)$}\label{ln:dsampleE}
$L_0\gets\textsc{Downsample}\bigl((\xB_k,\emptyset,|\xB_k|),\rho \bigr)$\Comment*[r]{take in new items}\label{ln:dsampleBE}
$L\gets\textsc{Union}\bigl(L_0,(A,\pi,C)\bigr)$\Comment*[r]{combine old and new items}\label{ln:unionE}
$S\gets \textsc{Getsample}(L)$\;
output $S$
}
}
\end{algorithm}

\textbf{Algorithm properties:} Theorem~\ref{th:rtbsIncl}(i) below asserts that R-TBS satisfies \eqref{eq:inclRTBS} and hence \eqref{eq:expratio}, thereby maintaining the correct inclusion probabilities. \revision{Indeed, suppose that the inclusion probability for $x\in\xB_i$ (with $i<k$) at time $t_{k-1}$ is $\rho_{k-1} f(\alpha_{i,k-1})$.  Write $\theta_k=e^{-\lambda(t_k-t_{k-1})}$ and observe that $\theta_k=f(\alpha_{i,k})/f(\alpha_{i,k-1})$ for any $i\in[0..k-1]$. It follows that  
\[
\rho_{k-1} f(\alpha_{i,k-1}) \cdot \Bigl(\frac{\rho_k}{\rho_{k-1}}\Bigr)\theta_k 
= \rho_{k-1} f(\alpha_{i,k-1})\cdot \Bigl(\frac{\rho_k f(\alpha_{i,k})}{\rho_{k-1}f(\alpha_{i,k-1})}\Bigr)=\rho_k f(\alpha_{i,k}),
\]
preserving the desired inclusion probability \eqref{eq:inclRTBS}; the downsampling operation in line~\ref{ln:dsampleE} executes this adjustment. (See Appendix~\ref{sec:proofs} for the detailed inductive proof.) Similarly,  an incoming item $x\in\xB_k$ is accepted into the sample with probability $\rho_k\cdot 1=\rho_kf(\alpha_{k,k})$---see line~\ref{ln:dsampleBE}---so that \eqref{eq:inclRTBS} holds in this case as well. Note that we can combine the above results over all of the batches by virtue of Theorem~\ref{th:union}.}
Theorem~\ref{th:rtbsIncl}(ii) implies that the sample size and stability are maximized, as formalized in Theorems~\ref{th:maxMean} and \ref{th:minVar} below. Finally, the assertion in Theorem~\ref{th:rtbsIncl}(iii) ensures that the inclusion probabilities for a given item are nonincreasing over time. This is crucial, since otherwise we might have to recover an item that was previously deleted from the sample, which is impossible. \revision{This monotonicity property trivially holds for traditional one-item-at-a-time reservoir sampling, where $\rho_kf(\alpha_{i,k})= \min(1,n/k)$ as discussed previously. For general decay functions, the monotonicity property hinges on the interplay of item decay and new-item arrival, and in fact does not generally hold if we try to define $\rho_k$ simply as $\min(1,n/W_k)$. To deal with this issue, we need a more complex scheme for defining $\rho_k$; see Section~\ref{sec:mono}.}

\begin{theorem}\label{th:rtbsIncl}
Let $\{L_k=(A_k,\pi_k,C_k)\}_{k\ge 1}$ and $\{S_k\}_{k\ge 1}$ be a sequence of latent samples and samples, respectively, produced by Algorithm~\ref{alg:rtbsE} and define $\rho_k=\min(1,n/W_k)$. Then
\begin{romlist}
\item $\prob{x\in S_k}=\rho_k f(\alpha_{i,k})$ for all $k\ge 1$, $i\le k$, and $x\in\xB_i$;
\item  $C_k=\rho_k W_k$ for all $k$; and
\item $\rho_kf(\alpha_{i,k})\le \rho_{k-1}f(\alpha_{i,k-1})$ for all $k>1$ and $i<k$.
\end{romlist}
\end{theorem}

We call a sample $S_k$ \emph{unsaturated} if $C_k<n$ and \emph{saturated} if $C_k=|S_k|=n$; note that we also have $W_k<n$ if and only if the sample is unsaturated. Theorems~\ref{th:maxMean} and \ref{th:minVar} below assert that, among all sampling algorithms with exponential time biasing, R-TBS both maximizes the expected sample size in unsaturated scenarios and minimizes sample-size variability. Thus R-TBS tends to yield more accurate ML results (via more training data) and greater stability in both result quality and retraining costs. 

\begin{theorem}\label{th:maxMean}
Let $H$ be any sampling algorithm for exponential decay that satisfies~\eqref{eq:expratio} and denote by $S_k$ and $S^{H}_k$ the samples produced at time~$t_k$ by R-TBS and H. If the total weight at some time $t_k\ge 1$ satisfies $W_k<n$, then $\mean[|S^H_k|]\le\mean[|S_k|]$.
\end{theorem}

\begin{proof}
Since $H$ satisfies \eqref{eq:expratio}, it follows that, for each time $t_j\le t_k$ and $x\in\xB_j$, the inclusion probability $\prob{x\in S^H_k}$ must be of the form $r_ke^{-\lambda(t_k-t_j)}$ for some function $r_k$ independent of $j$. Taking $j=k$, we see that $r_k\le 1$. For R-TBS in an unsaturated state, \eqref{eq:inclRTBS} implies that $r_k=\rho_k=C_k/W_k=1$, so that $\prob{x\in S^H_k}\le\prob{x\in S_k}$ , and the desired result follows directly.
\end{proof}

\begin{theorem}\label{th:minVar}
Let $H$ be any sampling algorithm for exponential decay that satisfies~\eqref{eq:expratio} and has maximal expected sample size $C_k$, and denote by $S_k$ and $S^{H}_k$ the samples produced at time~$t_k$ by R-TBS and H. Then $\var[|S^H_k|]\ge\var[|S_k|]$ for any $k\ge 1$.
\end{theorem}

\begin{proof}
Considering all possible distributions over the sample size having a mean value equal to $C_k$, it is straightforward to show that variance is minimized by concentrating all of the probability mass onto $\floor{C_k}$ and $\ceil{C_k}$. There is precisely one such distribution, which results from application of Algorithm~\ref{alg:getSample}, and this is precisely the sample-size distribution attained by R-TBS.
\end{proof}

\begin{algorithm}[t]
\caption{Reservoir-based TBS (R-TBS)}\label{alg:rtbs}
{\footnotesize
$f$: decay function\quad
$n$: maximum sample size\quad
$n'$: maximum sample weight\quad
$\lambda$: decay rate for consolidated sample\;
$\delta_1\in(0,1)$ and $\delta_2>0$: approximation parameters for consolidated latent sample\;
\revision{$m$: dynamic parameter such that items arriving at time $t_i< t_m$ are in consolidated latent sample}
\BlankLine
Initialize: $(\yA,\ypi,\yC)\gets(\emptyset,\emptyset,0)$; $\bL\gets\{(\yA,\ypi,\yC)\}$; $W\gets 0$;  $F_\infty=1/\textsc{Gamma}(f)$; $\rho\gets1$; $B^*\gets 0$; $m\gets 1$\;
\For{$k\gets1,2,\ldots$}{
  $B^*\gets \max(B^*,|\xB_k|)$\;
  \Comment{update total weight}
  \If{$k>1$}{
    \For(\Comment*[f]{decay weight of recent items}){$i\gets k-1,k-2,\ldots,m$}{
      $W\gets W- \Bigl(1-\frac{f(\alpha_{i,k})}{f(\alpha_{i,k-1})}\Bigr)(C_i/\rho)$\;\label{ln:downsampleOrd}
    }
    $W\gets W- (1-e^{-\lambda\Delta})(\yC/\rho)$\Comment*[r]{decay weight of consolidated items}\label{ln:downsampleExp}
  }
  $W\gets W+|\xB_k|$\Comment*[r]{add weight of new items}\label{ln:addNewW}
  \Comment{update $\rho$}
  \If{$k>1$}{\label{ln:startRho}
    $\rho'\gets\rho$\;
    $\rho^*\gets\min_{m\le i\le k}\rho' f(\alpha_{i,k-1})/f(\alpha_{i,k})$\;
    \lIf{$m>1$}{$\rho^*\gets\min(\rho^*,\rho' e^{\lambda\Delta})$}\label{ln:expComp}
    $\rho\gets \min(1,n'/W,\rho^*)$\;\label{ln:endRho}
  }
  \Comment{update samples}
  \If{$k>1$}{
     \For(\Comment*[f]{update recent samples}){$i\gets k-1,k-2,\ldots,m$}{
     $(A_i,\pi_i,C_i)\gets\textsc{Downsample}\bigl((A_i,\pi_i,C_i),\frac{\rho f(\alpha_{i,k})}{\rho' f(\alpha_{i,k-1})}\bigr)$\;\label{ln:updateS}
    } 
    $(\yA,\ypi,\yC)\gets\textsc{Downsample}\bigl((\yA,\ypi,\yC),(\rho/\rho')e^{-\lambda\Delta}\bigr)$\Comment*[r]{update consolidated latent sample}\label{ln:UpdateConS}
  }
  $(A_k,\pi_k,C_k)\gets \textsc{Downsample}\bigl((\xB_k,\emptyset,|\xB_k|),\rho\bigr)$\; \label{ln:downB}
  $\bL\gets \bL\cup\{(A_k,\pi_k,C_k)\}$\Comment*[r]{add sample of new items}
  \Comment{try to consolidate samples}
  \While{$\bigl(f(\alpha_{m,k})<\delta_1\bigr)\land\bigl( F_\infty-\sum_{i=m}^kf(\alpha_{i,k})<\delta_2/B^*\bigr)$}{\label{ln:entryCheck}
    $(\yA,\ypi,\yC)\gets \textsc{Union}\bigl((\yA,\ypi,\yC),(A_m,\pi_m,C_m)\bigr)$\;
    $\bL\gets \bL\setminus \{(A_m,\pi_m,C_m)\}$\;
    $m\gets m+1$\label{ln:endConsolidate}
  }
  \Comment{output sample to user}
  $(A,\pi,C)\gets\textsc{Union}\bigl((A_k,\pi_k,C_k),(A_{k-1},\pi_{k-1},C_{k-1}),\ldots,(A_m,\pi_m,C_m),(\yA,\ypi,\yC)\bigr)$\;\label{ln:bigUnion}
  \lIf{$C>n$}{$(A,\pi,C)\gets \textsc{Downsample}\bigr((A,\pi,C),n\bigr)$}\label{ln:finalDownsize}
  $S\gets \textsc{Getsample}\bigr((A,\pi,C)\bigr)$\;\label{ln:getsample}
  output $S$
}
}
\end{algorithm}

\subsection{The General R-TBS Algorithm}\label{sec:rtbsAlg}

\revision{For a general decay function, the decay factor for items in different batches is no longer the same, as in the exponential case, which adds substantial complexity to the R-TBS algorithm. In particular, we need to track the timestamps of individual items} and so, analogously to T-TBS, we represent the state $\bL$ of the sample as a set of triples of the form $L_i=(A_i,\pi_i,C_i)$, where $L_i$ is a latent sample of items from $\xB_i$. For $k\ge 1$ we denote by $\bL_k$ the state of the sample at time~$t_k$ and by $\bC_k=\sum_{(A_i,\pi_i,C_i)\in\bL_k}|C_i|$ the total sample weight at time~$t_k$. The pseudocode for the general R-TBS procedure is given as Algorithm~\ref{alg:rtbs}; the ideas behind the algorithm are developed below.

\subsubsection{A naive algorithm}

\revision{The core idea of Algorithm~\ref{alg:rtbs} is to try and mimic the exponential case. The naive version of doing this is as follows. When batch $\xB_k$ arrives, we first update the total weight $W_{k-1}$ by updating the weight $W_{k-1,i}$ of each batch $\xB_i$ (line~\ref{ln:downsampleOrd}) for $i<k$, and then adding the weight of batch $\xB_k$ (line~\ref{ln:addNewW}), which is just $|\xB_k|$ since each incoming item has weight $f(0)=1$. (By the ``weight of batch $\xB_j$'', we mean the total weight of all items in $\xB_j$.) Note that arguments essentially identical to those in Theorem~\ref{th:rtbsIncl} show that $C_i=\rho_j W_{j,i}$ for all $j\ge 1$ and $i\le j$, where $\rho_j$ is defined appropriately (see Section~\ref{sec:mono} below). Thus the term $C_i/\rho$ in line~\ref{ln:downsampleOrd} is precisely $W_{k-1,i}$, which loses a fraction $\rho_k f(\alpha_{i,k})/\bigl(\rho_{k-1}f(\alpha_{i,k-1})\bigr)$ of its weight. We next downsample each latent sample $L_i$ by a factor of $\rho_k f(\alpha_{i,k})/\bigl(\rho_{k-1}f(\alpha_{i,k-1})\bigr)$ for $i<k$---see line~\ref{ln:updateS}---and then downsample $(\xB_k,\emptyset,|\xB_k|)$ by a factor of $\rho_k$ (line~\ref{ln:downB}) to create $L_k$. To output a sample $\bS_k$ to the user, we first union $L_1,\ldots,L_k$ using Algorithm~\ref{alg:union} (line~\ref{ln:bigUnion}) and then create $\bS_k$ via Algorithm~\ref{alg:getSample} (line~\ref{ln:getsample}).

The algorithm as described indeed satisifies \eqref{eq:expratio}. Specifically, denoting by $S_i$ a sample created from $L_i$ via Algorithm~\ref{alg:getSample}, arguments as before show that $\prob{x\in S_i}=\rho_kf(\alpha_{i,k})$ for $x\in\xB_i$ with $i\le k$. Then Theorem~\ref{th:union} implies  that $\prob{x\in \bS_k}=\rho_kf(\alpha_{i,k})$ for all $i\le k$ and $x\in\xB_i$, so that \eqref{eq:expratio} holds.  Moreover, $\bC_k=\rho_k W_k$ as before. This naive algorithm, however, has two issues that take some effort to address.}

\subsubsection{Sample footprint}\label{sec:footprint}

Perhaps the most important problem with the naive algorithm is that the sample footprint grows without bound. To see this, observe that the latent sample for items in a given batch $\xB_i$ never empties out completely. At time goes on, the latent sample will eventually contain one partial item, whose appearance probability is always positive (though decreasing to 0). Thus the sample footprint at time $t_k$ is $\Omega(k)$.

\revision{Our solution to this problem is to approximate the exact time biasing scheme in the naive algorithm by maintaining, at each time $t_k$, distinct latent samples only for items that have arrived at time $t_{m(k)}$ or later, where $m(k)$ is a carefully chosen index that increases with $k$. We denote by $\alpha^*_k=\alpha_{m(k),k}$ the age at time $t_k$ of items that arrived at $t_{m(k)}$. Sample items that arrived earlier than time $t_{m(k)}$, i.e., whose age is greater than $\alpha^*_k$, are maintained in a single \emph{consolidated} latent sample, which decays at an exponential rate $\lambda$. Thus $\bL_k$ comprises $k-m(k)+2$ latent samples in total. The values of $m(k)$ and $\lambda$ are determined by parameters $\delta_1$ and $\delta_2$ that control the accuracy of the approximate time biasing scheme, as described below.}

Recall our running assumption that $t_i=i\Delta$ for $i\ge 0$ and some $\Delta>0$. We also assume that $F_\infty<\infty$---where $f_i=f(i\Delta)$  and $F_k=\sum_{j=0}^k f_i$ as before---and that $\lambda$ is chosen to ensure that
\begin{equation}\label{eq:bdLambda}
e^{-\lambda\Delta}\le f(\alpha+\Delta)/f(\alpha)
\end{equation}
for all $\alpha$ large enough so that $f(\alpha)<\delta_1$. This is always possible for subexponential decay functions. For example, if $f(\alpha)=1/(1+\alpha)^s$, then \eqref{eq:bdLambda} holds for any $\lambda\ge s\ln\bigl((1+\hat\alpha+\Delta)/(1+\hat\alpha)\bigr)/\Delta$, where $\hat\alpha$ is the smallest value of $\alpha$ such that $f(\alpha)<\delta_1$. We also assume that $\bbar\triangleq\sup_i|\xB_i|<\infty$.

The use of a consolidated sample is equivalent to using, at each time $t_k$, a modified decay function $\fa_k$ given by
\[
\fa_k(\alpha)= 
\begin{cases}
f(\alpha)&\text{if $\alpha\le \alpha^*_k$};\\
f(\alpha^*_k)e^{-\lambda(\alpha-\alpha^*_k)}&\text{if $\alpha > \alpha^*_k$}.
\end{cases}
\]
Thus the appearance probability for recent items is governed exactly by the desired function $f$, whereas the appearance probability for older items is perturbed. \revision{Set $N=\min\{\,n\ge 1:\sum_{i=n}^\infty f_i\le \delta_2/\bbar\,\}$ and observe that $N<\infty$ by finiteness of $F_\infty$. Theorem~\ref{th:consol} below shows that, at each time $t_k$, (i) the parameter $\delta_1$ in Algorithm~\ref{alg:rtbs} is a bound on the absolute difference between $f$ and $\fa_k$, i.e., the amount by which an older item's appearance probability is perturbed, (ii) the parameter $\delta_2$ is a bound on the expected number of older items whose appearance probability is perturbed, and (iii) at most $N+2$ latent samples need to be stored at any time point. This bound on the latent samples, coupled with the bound on the sample weights enforced by the reservoir capacity, ensures that the sample footprint is bounded.}  
\begin{theorem}\label{th:consol}
Algorithm~\ref{alg:rtbs} ensures the following properties for each $t_k$:
\begin{romlist}
\item $|f(\alpha)-\fa_k(\alpha)|\le\delta_1$ for all $\alpha\ge\alpha^*_k$;
\item $\sum_{i=1}^{m(k)-1} |\xB_i| f(\alpha_{i,k})<\delta_2$; and
\item $k-m(k)\le N$.
\end{romlist}
\end{theorem}
\revision{Examination of the proof of this theorem shows that the key invariants that must be maintained are (i) $f(\alpha)<\delta_1$ for $\alpha>\alpha^*_k$ and (ii) $F_\infty-\sum_{i=m(k)}^kf(\alpha_{i,k})<\delta_2/\max_{1\le i\le k}|\xB_i|$. Lines~\ref{ln:entryCheck}--\ref{ln:endConsolidate} maintain these invariants while trying to add as many latent samples as possible to the consolidated sample. Note that the $\textsc{Gamma}(f)$ function is the same function that is used to compute $1/F_\infty$ in the T-TBS algorithm. The code in lines~\ref{ln:downsampleExp}, \ref{ln:expComp}, and \ref{ln:UpdateConS} performs the same operations on the consolidated latent sample as the code in the adjacent lines performs on the other latent samples.}

Note that $N$ is typically a conservative upper bound, both because of the use of the conservative constant $\bbar$ and the fact that we have implicitly used the upper bound of $\rho_k=1$ in our analysis for $k\ge 1$. Suppose, for example, that $\Delta=1$, $\bbar=10{,}000$, $n=100{,}000$, and $f(\alpha)=1/(1+\alpha)^2$. If $\delta_1=10^{-4}$ and $\delta_2=100$ items---so that only 0.1\% of sample items have perturbed appearance probabilities---then the number of latent samples stored is bounded above by $N+2\approx100$. This corresponds to storage of up to 100 partial items, and represents about a 0.1\% storage overhead relative to exponential decay (which stores at most one partial item). When $f$ decays slowly, however, the overhead can become substantial. Thus, as in T-TBS, very slowly decaying functions can be expensive and cumbersome from a practical standpoint.

\subsubsection{Monotonicity of inclusion probabilities} \label{sec:mono}

\revision{The remaining issue is that, as hinted above, we cannot simply choose $\rho_k=\min(1,n/W_k)$, even though this would produce the largest possible sample sizes. The problem is that, with this choice, the resulting appearance probability $p_{i,k}=\rho_k f(\alpha_{i,k})=\min(1,n/W_k)f(\alpha_{i,k})$ is no longer guaranteed to be nonincreasing in $k$. This would cause the algorithm to break. In particular, the downsizing operation for latent sample $L_i$ would actually try to \emph{upsize} the sample, which is impossible since any $\xB_i$ items not in the sample have been discarded.} For example, suppose that $n=1000$ and $W_{k-1}=2000$, and that $|\xB_i|=100$ and $f(\alpha_{i,k-1})=0.1$ for some batch $\xB_i$ with $i<k$. Then $p_{i,k-1}=\rho_{k-1}f(\alpha_{i,k-1})=(n/W_{k-1})f(\alpha_{i,k-1})=0.5\cdot 0.1=0.05$, so that the expected number of sample items from $\xB_i$ at time $t_{k-1}$ is $|\xB_i|\cdot 0.05=5$. Now suppose that $f(\alpha_{i,k})/f(\alpha_{i,k-1})=0.9$ so that $f(\alpha_{i,k})=0.09$, but that $W_k$ drops to 1000 through rapid decay in other batches. Then $\rho_k=(n/W_k)=1$ and hence $p_{i,k}=1\cdot 0.09=0.09>p_{i,k-1}$. 
The largest inclusion probability we can support is $\rho^*f(\alpha_{i,k})=0.05$, where $\rho^*=\rho_{k-1}f(\alpha_{i,k-1})/f(\alpha_{i,k})=5/9$. Our problem arises because the large 50\% decrease in the total weight  causes a factor of 2 increase in $\rho$, which overwhelms the factor of 0.9 decrease in $f$ and causes a net increase in the appearance probability. We note that this type of situation does not arise when the decay function is exponential because changes in weight are linked to changes in $f$ uniformly across all batches.

A general solution to this problem is to choose $\rho_k$ as large as possible to maximize the sample size, while ensuring monotonicity in the appearance probabilities. Specifically, for a given $i$ we want to choose $\rho_k$ such that $\rho_kf(\alpha_{i,k})\le \rho_{k-1}f(\alpha_{i,k-1})$. The maximum feasible value is $\rho^*_{k,i}=\rho_{k-1}f(\alpha_{i,k-1})/f(\alpha_{i,k})$ as in our example. Because the inequality must hold for every $i\le k$, the overall maximum feasible value is $\rho^*_1=1$ and $\rho^*_k=\min_{i \le k}\rho^*_{i,k}$ for $k>1$. \revision{(The definition of $\rho^*_k$ must actually be adjusted slightly to deal wth the consolidated sample, as in line~\ref{ln:expComp} of Algorithm~~\ref{alg:rtbs}.)}
As indicated by the previous example, $\rho^*_k$ imposes an upper bound on the sample weight, so that direct use of $\rho^*_k$ will produce samples whose expected size is generally less than $\min(W_k,n)$ at each $t_k$. In particular, the algorithm may return a sample of size less than $n$ at a time $t_k$ even when $W_k>n$. We deal with this problem by increasing the maximum sample weight, thereby trading off storage and sample size stability. In detail, we set $n'>n$ as the maximum sample weight, so that maximum sample weight and maximum sample size no longer coincide as in the exponential case. We then set
\begin{equation}\label{eq:setRho}
\rho_k=\min(1,n'/W_k,\rho^*_k)
\end{equation}
for $k\ge 1$. If the total sample weight $\bC_k$ exceeds $n$ at  time~$t_k$ then we downsample to reduce the weight to $n$ before outputting $S_k$ \revision{(line~\ref{ln:finalDownsize})}. Lines~\ref{ln:startRho}--\ref{ln:endRho} in Algorithm~\ref{alg:rtbs} implement the computation of $\rho_k$.  Arguments essentially identical to those in Theorem~\ref{th:rtbsIncl} show that $\bC_k=\rho_k W_k$ for all $k\ge 1$, where $\rho_k$ is now defined as above. Since $\rho_k=\min(1,n'/W_k,\rho^*_k)$, we have that the total sample weight is bounded by $\min(W_k,n')$. As discussed in Section~\ref{sec:footprint}, we can bound the number of latent samples in $\bL_k$, so that the overall sample footprint is bounded.
 
The following proposition helps in understanding the sample size dynamics and the precise way in which the choice of $n'$ trades off storage for sample-size stability. In the proposition, we set $W^*_k=\max_{1\le i\le k}W_i$.
\begin{proposition}\label{prop:ssBDx}
Algorithm~\ref{alg:rtbs} has the following properties:
\begin{romlist}
\item If $W^*_k\le n'$, then $\rho_i=1$ for $1\le i\le k$;
\item If $W^*_k>n'$, then $\rho_k\ge n'/W^*_k$; and
\item If the decay function $f$ is strictly decreasing and $\rho_k=\rho^*_k<\min(1,n'/W_k)$, then $\rho_k>\rho_{k-1}$.
\end{romlist}
\end{proposition}

Observe that, by \eqref{eq:setRho}, $\rho_k\le 1$ for all $k$; when $\rho_k=1$ (so that $W_k\le n'$), R-TBS achieves the largest possible sample weight of $\bC_k=W_k$. Assertion~(i) of the proposition implies that, as in the the case of exponential decay, R-TBS returns the maximum possible sample size until the total weight first exceeds $n'$. As indicated by the previous example, one problematic scenario occurs when $W_k$ exceeds $n'$ and then drops rapidly but stays above $n$. After such an ``adversarial'' drop, the necessity of ejecting existing items and filtering out newly arrived items to enforce \eqref{eq:expratio} means that R-TBS returns a sample of size smaller than $n$. In the case where $W^*_k>n'$, it follows from assertion~(ii) that the \textsc{Union}ed latent sample has a sample weight of $\bC_k=\rho_kW_k\ge n'W_k/W^*_k$. This weight will exceed $n$---so that we can return a final sample of size $n$ to the user after downsampling---as long as $W_k\ge (n/n')W^*_k$. Thus the larger the value of $n'$, the larger the adversarial drop from $W^*_k$ that can be tolerated without reducing the returned sample size below its optimal value of $n$. As indicated by assertion~(iii), we see that, when $\rho_k<\min(1,n'/W_k)$, so that $\rho_k=\rho^*_k$ and suboptimal sample sizes are being returned, the $\rho_k$ values will increase toward the optimal values of $\min(1,n'/W_k)$ as time progresses, so that the algorithm will steadily ``recover'' from the adversarial drop until the next adversarial drop occurs. An analogous situation occurs when the total weight decreases to a value  less than $n$ at an adversarial drop. In this case the sample weight $\bC_k$ falls below the desired value of $W_k$ but then recovers as time goes on.

\begin{figure*}[tbh]
 \centering
 	\subfigure[$\min_k W_k>n$; $n'=n$]{
	   \label{fig:WgtNb}\includegraphics[width=0.3\linewidth]{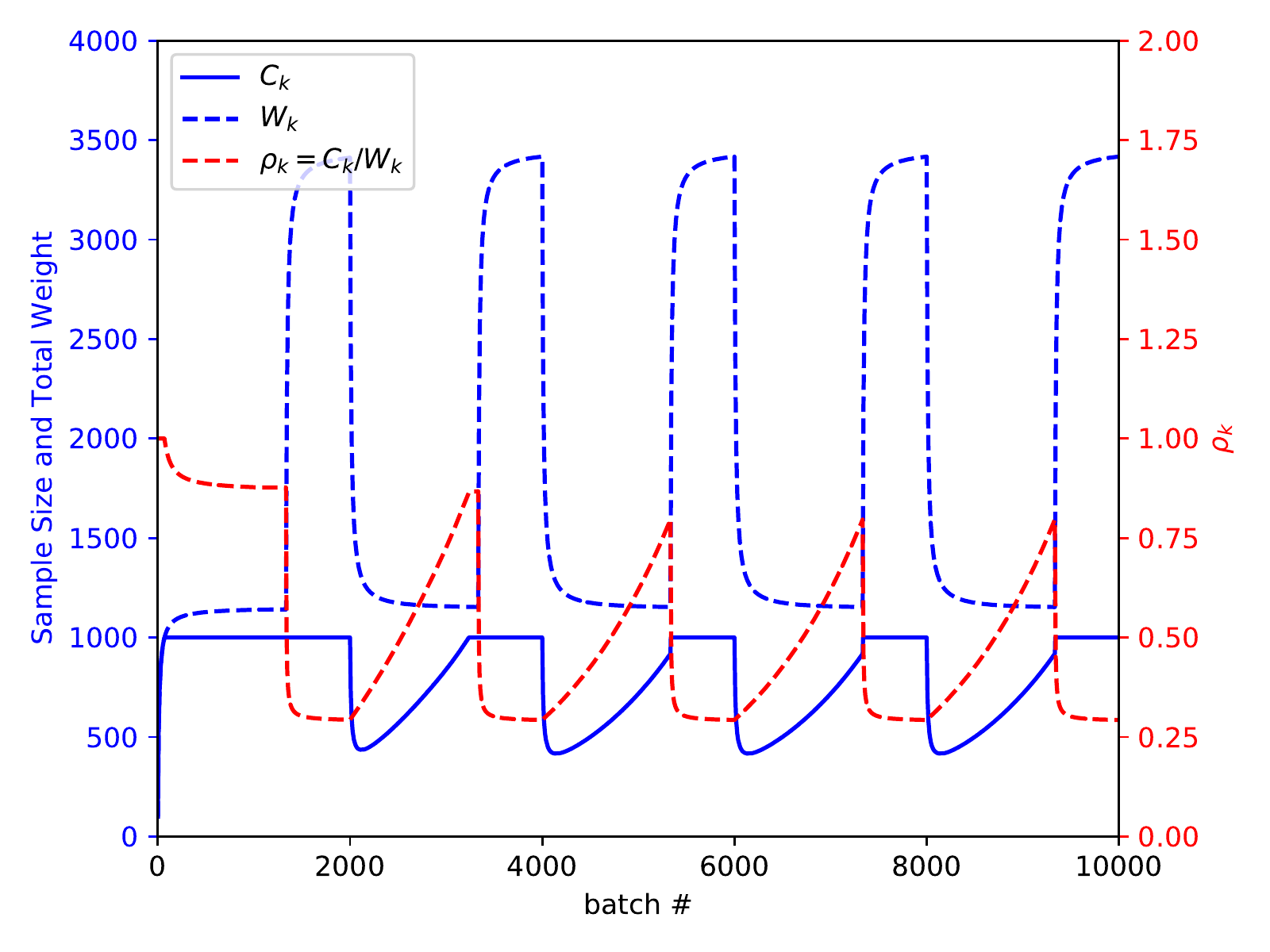}} 
	\subfigure[$\min_k W_k>n$; $n'=2n$]{
	   \label{fig:WgtN}\includegraphics[width=0.3\linewidth]{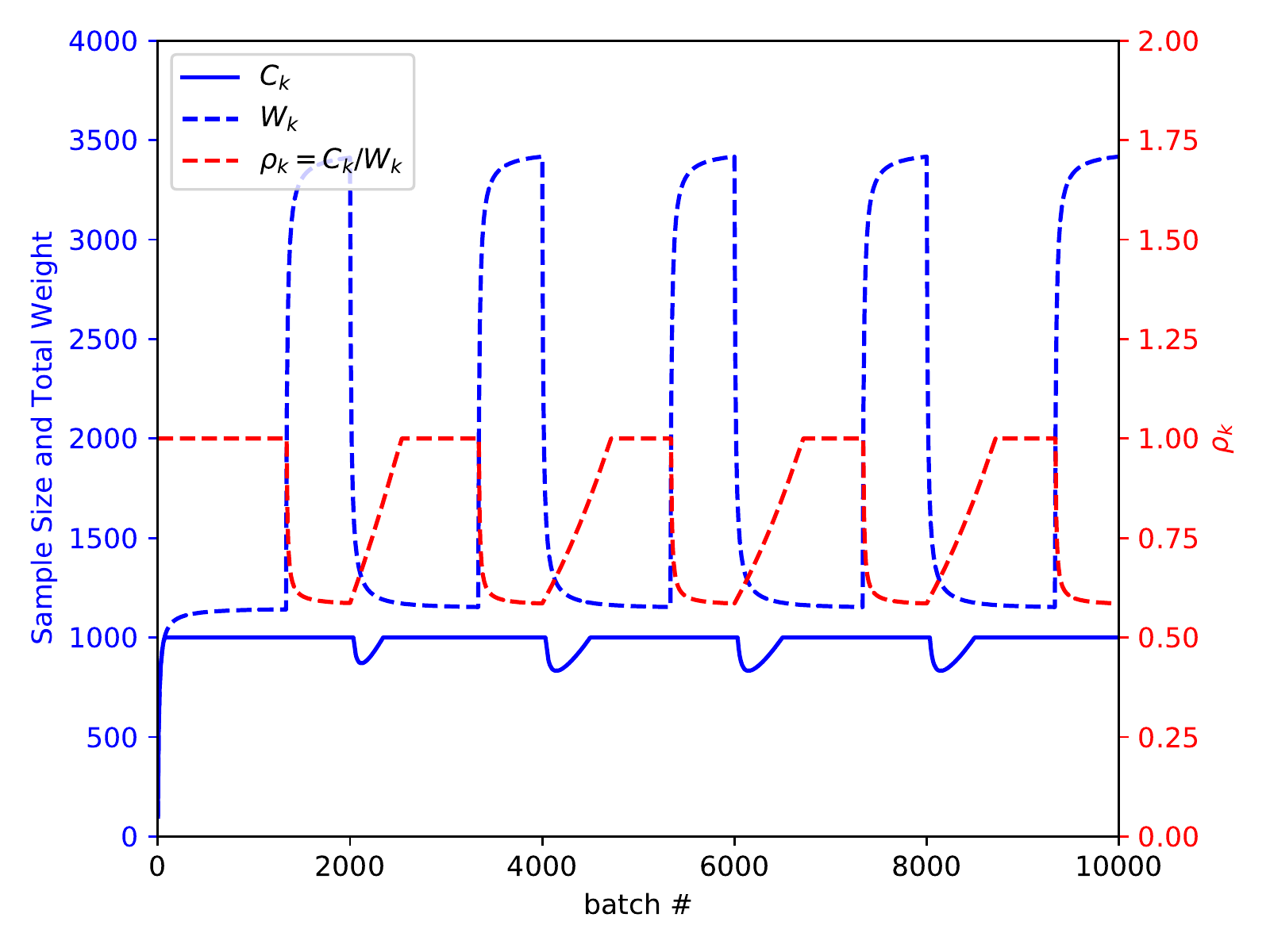}} 
	\subfigure[$\min_k W_k < n$; $n'=2n$]{
	   \label{fig:WlN}\includegraphics[width=0.3\linewidth]{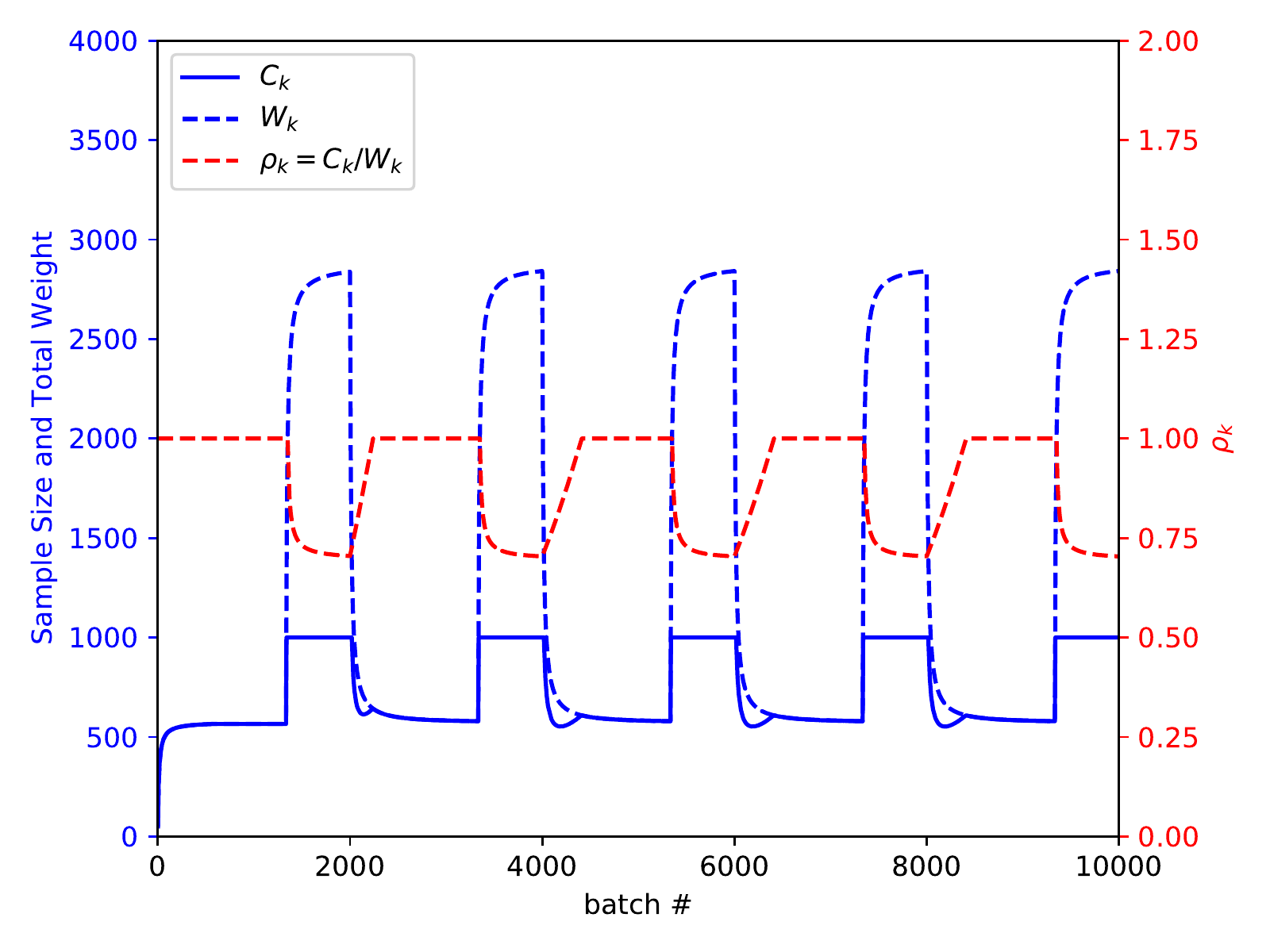}} 
\caption{\label{fig:rhoEffect}R-TBS sample size fluctuations over time with quadratic decay; $n=1000$}
\end{figure*}

This behavior is illustrated in Figure~\ref{fig:rhoEffect}. In the figure, arriving batch sizes rise and fall periodically, over periods of length $m=2000$ batches. In each period, the first $(2/3)m$ batches each contain $b_1$ items and the last $(1/3)m$ batches each contain $b_2$ items, where $b_1<b_2$, leading to a sequence of adversarial drops. In Figures~\ref{fig:WgtNb} and \ref{fig:WgtN}, we have $b_1=100$, $b_2=300$ and the total weight $W_k$ exceeds $n=1000$ at all times. Each sharp spike upwards in $W_k$ (due to the sudden jump in batch size) causes $\rho_k$ to fall, in order to maintain the proper inclusion probabilities. At each subsequent sharp decline in $W_k$, the sample size $\bC_k$ dips below $n$, even though the total weight exceeds $n$. The $\rho_k$ values then slowly recover, approaching 1, until the next spike causes the cycle to repeat. Comparing Figures~\ref{fig:WgtNb} and \ref{fig:WgtN}, we see that increasing the maximum sample weight from $n'=n$ to $n'=2n$ largely ameliorates the dips in the sample size. The sample-size dynamics in Figure~\ref{fig:WlN} are similar, but here $b_1=50$ and $b_2=250$, so that the total weight $W_k$ occasionally falls below $n$. Just after each sharp decline in $W_k$, there is a small dip where $\bC_k<W_k<n$ for a small number of sequential values of $k$; in the case of exponential decay we would have $\bC_k=W_k<n$ at these time points. Because we take $n'=2n$, the magnitude of the dips is small.


\section{Distributed TBS Algorithms}\label{sec:imp}

In this section, we describe how to implement distributed versions of T-TBS and R-TBS to handle large volumes of data. 




\subsection{Overview of Distributed Algorithms}

The distributed T-TBS and R-TBS algorithms, denoted as D-T-TBS and D-R-TBS respectively, 
need to distribute large data sets across the cluster and parallelize the computation on them.


For exponential decay, at any point in time, all items have exactly the same decay rate for their appearance probabilities, regardless of age, so there is no need to keep track of an item's age in the sample. This nice ``memoryless'' property of the exponential function makes the distributed implementation of both algorithms easier. We first describe D-T-TBS and D-R-TBS for exponential decay, and then discuss the extensions for general decay functions in Section~\ref{sec:general-implementation}.

\textbf{Overview of D-T-TBS:} The implementation of the D-T-TBS algorithm is very similar to the simple distributed Bernoulli time-biased sampling algorithm in~\cite{XieTSBH15}. It is embarrassingly parallel, requiring no coordination. At each time point~$t_k$, each worker in the cluster downsamples its partition of the sample with probability $p$, downsamples its partition of $\xB_k$ with probability $q$, and then takes a union of the resulting data sets. 

\textbf{Overview of D-R-TBS:} This algorithm, unlike D-T-TBS, maintains a bounded sample, and hence is not embarrassingly parallel. D-R-TBS first needs to aggregate the local partition sizes for the incoming batch $\xB_k$ to compute the total batch size $|\xB_k|$ and calculate the new total weight $W_k$. Then, based on $|\xB_k|$, $W_k$, and the current sample weight $C_k$, D-R-TBS computes the downsample rate for the items in the reservoir, as well as the downsample rate for the items in $\xB_k$. After that, D-R-TBS chooses the items in the reservoir to delete through a \textsc{Downsample} operation, selects items in $\xB_k$ (also via \textsc{Downsample}), inserts the selected items into the reservoir (via \textsc{Union}), and finally generates the sample (via \textsc{Getsample}). Each of the expensive operations \textsc{Downsample}, \textsc{Union}, and \textsc{Getsample} is performed in a distributed manner. They each require the master to coordinate among the workers. \textsc{Getsample} and \textsc{Union} operations are relatively straightforward. The most challenging part of D-R-TBS lies in choosing items to delete from the reservoir and selecting new items to insert into the reservoir. In Section~\ref{sec:updates}, we introduce two alternative approaches to determine the deleted and inserted items. 

Both D-T-TBS and D-R-TBS periodically checkpoint the reservoir as well as other system state variables to ensure fault tolerance---\revision{the fault tolerance of the distributed algorithms in Spark is discussed in Appendix~\ref{sec:spark-impl}}. The implementation details for D-T-TBS are mostly subsumed by those for D-R-TBS, so we focus on the latter.

\subsection{Distributed Data Structures}\label{sec:reservoir}

There are two important data structures in the D-R-TBS algorithm: the incoming batch and the reservoir. Conceptually, we view an incoming batch $\xB_k$ as an array of slots numbered from 1 through $|\xB_k|$, and the reservoir as an array of slots numbered from 1 through $\floor{C_k}$ containing full items plus a special slot for the partial item. For both data structures, data items need to be distributed into partitions due to the large data volumes. Therefore, the slot number of an item maps to a specific partition ID and a position inside the partition.

The incoming batch usually comes from a distributed streaming system, such as Spark Streaming; the actual data structure is specific to the streaming system , e.g., an incoming batch is stored as an RDD in Spark Streaming. As a result, the partitioning strategy of the incoming batch is opaque to the D-R-TBS algorithm. Unlike the incoming batch, which is read-only and discarded at the end of each time period, the reservoir data structure must be continually updated. An effective strategy for storing and operating on the reservoir is thus crucial for good performance. We now explore alternative approaches to implementing the reservoir.


\begin{figure}[t]
	\centering
	\includegraphics[width=3.5 in]{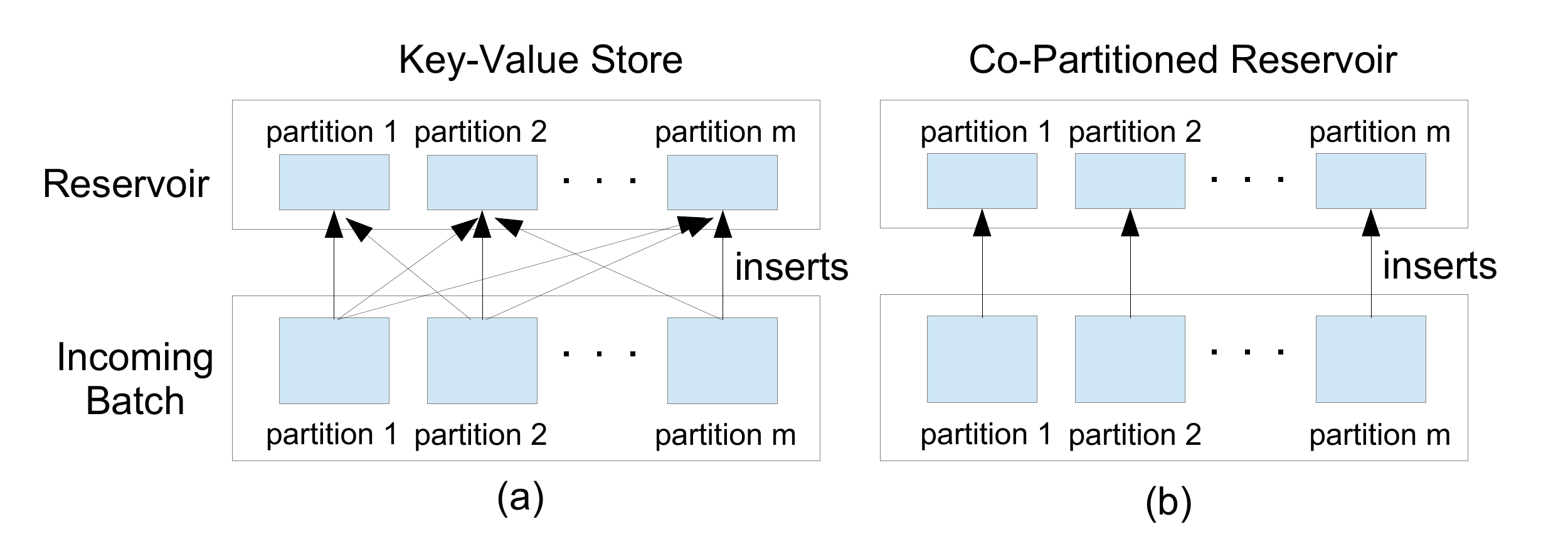}
	\caption{Design choices for implementing the reservoir}
	\label{fig:reservoir}
\end{figure}
\begin{figure}[t]
	\centering
	\includegraphics[width=3.5 in]{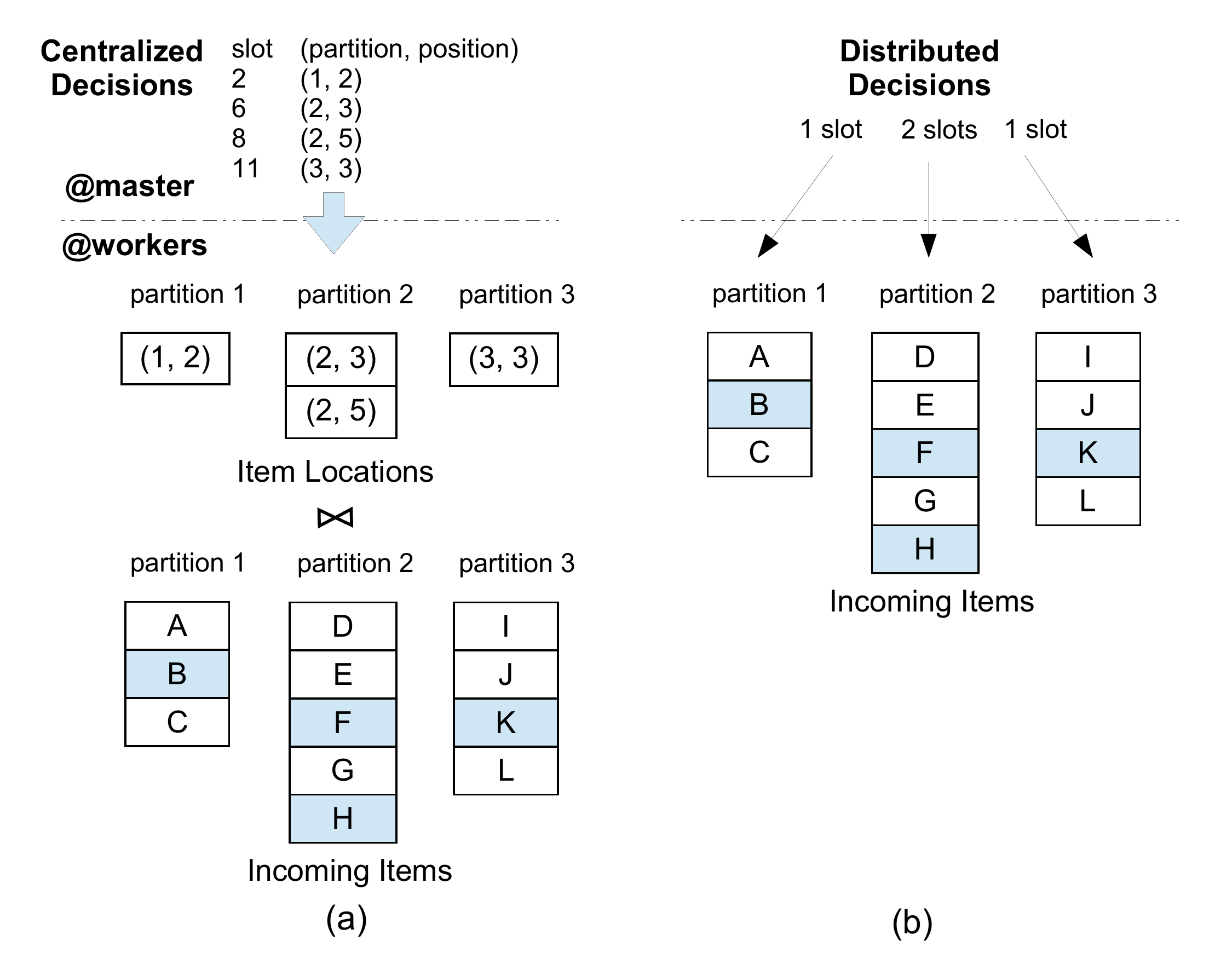}
	\caption{Retrieving insert items}
	\label{fig:insert}
\end{figure}

\textbf{Distributed in-memory key-value store:}
One quite natural approach implements the reservoir using an off-the-shelf distributed in-memory key-value store, such as Redis~\cite{redis} or Memcached~\cite{memcached}. In this scheme, each item in the reservoir is stored as a key-value pair, with the slot number as the key and the item as the value. The partial item has a special slot number such as -1. Inserts and deletes to the reservoir naturally translate into put and delete operations to the key-value store.

There are three major limitations to this approach. First, the hash-based or range-based data-partitioning scheme used by a distributed key-value store yields reservoir partitions that do not correlate with the partitions of incoming batch. As illustrated in Figure~\ref{fig:reservoir}(a), when items from a given partition of an incoming batch are inserted into the reservoir, the inserts touch many (if not all) partitions of the reservoir, incurring heavy network I/O. Second, key-value stores incur unnecessary concurrency-control overhead. For each batch, D-R-TBS already carefully coordinates the deletes and inserts so that no two delete or insert operations access the same slots in the reservoir and there is no danger of write-write or read-write conflicts. Finally, the key-value store approach requires an explicit slot number for each item. As a result, D-R-TBS needs to take extra care to make sure that after deletes and inserts of reservoir items, the slot numbers are still unique and contiguous, e.g. by recycling the slot numbers of deleted items for new inserts. The burden of keeping track of delete and insert slot numbers falls on the master node.


\textbf{Co-partitioned reservoir:}
An alternative approach implements a distributed in-memory data structure for the reservoir so as to ensure that the reservoir partitions coincide with the partitions from incoming batches, as shown in Figure~\ref{fig:reservoir}(b). This can be achieved in spite of the unknown partitioning scheme of the streaming system. Specifically, the reservoir is initially empty, and all items in the reservoir are from the incoming batches. Therefore, if an item from a given partition of an incoming batch is always inserted into the corresponding ``local'' reservoir partition and deletes are also handled locally, then the co-partitioning and co-location of the reservoir and incoming batch partitions is automatic. In addition, if the incoming batch is evenly distributed across the different partitions (which is often the case in practice), then the co-partitioned reservoir is also evenly distributed. For our experiments, we implemented the co-partitioned reservoir in Spark using the in-place updating technique for RDDs in~\cite{XieTSBH15}; see Appendix~\ref{sec:spark-impl}.

Note that, at any point in time, a given (conceptual) slot number in the reservoir maps to a specific partition ID and a position inside the partition. The mapping between a specific full item and its current slot number may change over time due to reservoir insertions and deletions. This does not cause any statistical issues, because the set-based R-TBS algorithm is oblivious to specific slot numbers. Thus the master only needs to keep track of the size of each partition and the position of the partial item. \revision{In Section~\ref{ssec:runtime}, we experimentally compare the key-value store and the co-partitioned reservoir approaches.}


\subsection{Choosing Items to Delete and Insert}\label{sec:updates}

In order to bound the reservoir size, D-R-TBS requires careful coordination among workers when choosing the set of items to delete from, and insert into, the reservoir. At the same time, D-R-TBS must ensure the statistical correctness of random number generation and random permutation operations in the distributed environment. We consider two possible approaches.

\textbf{Centralized decisions:}
In the most straightforward approach, the master makes centralized decisions about which items to delete and insert. For deletes, the master generates slot numbers of the items in the reservoir to be deleted, which are then mapped to the actual data locations in a manner that depends on the representation of the reservoir (key-value store or co-partitioned reservoir). For inserts, the master generates the slot numbers of the incoming items $\xB_k$ at time~$t_k$ that need to be inserted into the reservoir. Suppose that $\xB_k$ comprises $m\ge 1$ partitions. Each generated slot number $i\in\{1,2,\ldots,|\xB_k|\}$ is mapped to a partition $p_i$ of $\xB_k$ (where $1\le p_i\le m$) and a position $r_i$ inside partition~$p_i$. Denote by $\mathcal{Q}$ the set of ``item locations'', i.e., the set of $(p_i,r_i)$ pairs. In order to perform the inserts, D-R-TBS needs to first retrieve the actual items based on the item locations. This can be achieved with a join-like operation between $\mathcal{Q}$ and $\xB_k$, with the $(p_i,r_i)$ pair matching the actual location of an item inside $\xB_k$. To optimize this operation, we make $\mathcal{Q}$ a distributed data structure and use a customized partitioner to ensure that all pairs $(p_i,r_i)$ with $p_i=j$ are co-located with partition~$j$ of $\xB_k$ for $j=1,2,\ldots,m$. Then a co-partitioned and co-located join can be carried out between $\mathcal{Q}$ and $\xB_k$, as illustrated in Figure~\ref{fig:insert}(a) for $m=3$. The resulting set of retrieved insert items, denoted as $\mathcal{S}$, is also co-partitioned with $\xB_k$ as a by-product. After that, the actual deletes and inserts are then carried out depending on how reservoir is stored, as discussed below.

When the reservoir is implemented as a key-value store, the deletes can be directly applied based on the slot numbers. For inserts, the master takes each generated slot number of an item in $\xB_k$ and chooses a companion destination slot number in the reservoir into which the $\xB_k$ item will be inserted. This destination reservoir slot might currently be empty due to an earlier deletion, or might contain an item that will now be replaced by the newly inserted batch item. After the actual items to insert are retrieved as described previously, the destination slot numbers are used to put the items into the correct locations in the key-value store. 

When the co-partitioned reservoir is used, the delete slot numbers in the reservoir are mapped to $(p_i,r_i)$ pairs of partitions of the reservoir and positions inside the partitions. As with inserts, we again use a customized partitioner for the set of pairs $\mathcal{R}$ such that deletes are co-located with the corresponding reservoir partitions. Then a join-like operation on $\mathcal{R}$ and the reservoir performs the actual delete operations on the reservoir. For inserts, we simply use another join-like operation on the set of retrieved insert items $\mathcal{S}$ and the reservoir to add the corresponding insert items to the co-located partition of the reservoir. In this approach, we don't need the master to generate destination reservoir slot numbers for these insert items, because we view the reservoir as a set when using a co-partitioned reservoir data structure.

\textbf{Distributed decisions:}
The above approach requires the master to generate a large number of slot numbers, so we now explore an alternative approach that offloads the slot number generation to the workers while still ensuring the statistical correctness of the computation. This approach has the master choose only the number of deletes and inserts per worker according to an appropriate multivariate hypergeometric distribution. For deletes, each worker chooses random victims from its local partition of the reservoir based on the number of deletes given by the master. For inserts, the worker 
receives the number of inserts $I$ and then randomly and uniformly selects $I$ items from its local partition of the incoming batch $\xB_k$. Figure~\ref{fig:insert}(b) depicts how the insert items are retrieved under this decentralized approach. We use the technique in~\cite{Haramoto} for parallel pseudo-random number generation.

The foregoing distributed decision making approach works only when the co-partitioned reservoir data structure is used. This is because the key-value store representation of the reservoir requires a target reservoir slot number for each insert item from the incoming batch, and the target slot numbers have to be generated in such a way as to ensure that, after the deletes and inserts, all of the slot numbers are still unique and contiguous in the new reservoir. This requires a lot of coordination among the workers, which inhibits truly distributed decision making.

\subsection{Extensions to Generalized TBS Algorithms}\label{sec:general-implementation}

We now discuss the extensions to the above distributed implementations to make them work for general (non-exponential) decay functions.

\textbf{Changes to the reservoir data structure:}
Use of general decay functions requires some significant changes to the reservoir data structure. In the absence of the special memoryless property for exponential decay functions, the reservoir needs to keep track of the age for each item, since the decay rate for an item depends explicitly on its age. In addition, the generalized D-R-TBS algorithm also requires storage of multiple latent samples, and hence multiple partial items, in the reservoir---one per each age, up to the cutoff age, plus one for the consolidated latent sample of older items.

In the key-value store approach, D-R-TBS needs to record the arrival time of each item in the sample by augmenting the value component of the key-value pair that represents the item; the item's age can then be calculated on the fly. D-R-TBS also needs to maintain multiple special key-value pairs for the partial items. As discussed in Section~\ref{sec:updates}, the key-value store approach only works when centralized decisions determine which items to delete from, and insert into, the reservoir. Thus for general decay functions, the master node needs to track the arrival time corresponding to each slot number of the reservoir, so that it can apply different decay rates to differently aged items when deleting items from the reservoir.  

In the co-partitioned reservoir approach, D-R-TBS also needs to record the item arrival times and to support multiple latent samples corresponding to different ages. In the implementation, we can either simply add an arrival-time field to each item in the reservoir or organize the items for each arrival time together in each partition. We choose the latter for the ease of the downsampling process. Assuming that each incoming batch is evenly distributed across partitions, then the items corresponding to each arrival time are also evenly distributed across the reservoir partitions, and hence the overall reservoir structure is well balanced. Finally, the master needs to record the number of items for each arrival time in each partition of the reservoir and the positions of the partial items for each arrival time. This overhead is much smaller than the overhead of tracking the arrival time corresponding to each slot number, as needed in the key-value-store approach. Thus, for general decay functions, the co-partitioned reservoir approach dominates the key-value store approach even more than in the exponential decay setting.

\textbf{Changes to item insertion and deletion:}
The presence of general decay functions adds complexity to the process of choosing items to delete from, and insert into, the reservoir, since different decay rates apply to different latent samples in the reservoir. As mentioned before, the centralized-decision approach requires the master to record the arrival time corresponding to each slot number of the reservoir in order to select items to delete according to the correct probabilities. Moreover, as in the exponential case, special care needs to be taken to make sure that the slot numbers remain unique and contiguous after deletes and inserts. In contrast, for the distributed-decision approach, the master merely needs to decide the number of items to delete for each arrival time in each partition; the workers carry out the actual deletes. The remaining aspects of the distributed implementation stay the same.




\section{Experiments}\label{sec:exp}

In this section, we study the empirical performance of distributed implementations of the R-TBS and T-TBS algorithms, and demonstrate the potential benefit of using them for model retraining in online model management. 


\textbf{Experimental Setup: }
We implemented R-TBS and T-TBS on Spark (Appendix~\ref{sec:spark-impl}
contains Spark-specific implementation details). All performance experiments were conducted on a cluster of 9 ProLiant DL160 G6 servers. Each has two twelve-core Intel Xeon X5650 CPUs at 2.66GHz, with 15GB of RAM and a single 7200 RPM 500GB hard drive. Servers are interconnected using a 1 Gbit Ethernet and each server runs CentOS release 6.5, Java 1.8 and Spark 2.3. One server is dedicated to run the Spark coordinator and, each of the remaining servers runs a single Spark worker with parallelism 10, along with 10~GB of dedicated memory. All other Spark parameters are set to their default values. We used Memcached 1.4.4 as the key-value store in our experiments.

We note that the experimental setup in this paper differs from that in our previous work~\cite{HentschelHT18}. 
Besides using a different hardware configuration, we also upgraded Spark from version 1.6 to 2.3 and Memcached from version 1.4.33 to 1.4.4. All reported experimental results correspond to the current hardware and software configuration. The performance of Spark 2.3 has dramatically improved over 1.6, and all algorithms benefited from these improvements.  Although the actual runtimes of each algorithm changed, the qualitative results in the exponential case are generally consistent with those in \cite{HentschelHT18}. Interestingly, we observed a reduction in the performance gap between R-TBS and T-TBS; see below. 

For all experiments, data was streamed in from HDFS using Spark Streaming's microbatches. Decay occurs according to a time scale such that the batch-arrival interval is $\Delta=1$ in the decay formulas. We report run time per batch as the average over 100 batches, discarding the first round from this average because of Spark startup costs. We experiment with two versions of R-TBS. The first, \textit{R-TBS exp}, refers to the specialized version of R-TBS for exponential decay with $f(\alpha) = e^{-\lambda \alpha}$. The second, \textit{R-TBS poly}, refers to the generalized R-TBS with a shifted polynomial decay function $f^{(d)}(\alpha) = (1+d)^s/(1+d + \alpha)^s$; as discussed in Section~\ref{sec:ttbsAlg}, the shift is necessary so that item's weights do not become too small too quickly. We similarly experiment with T-TBS exp and T-TBS poly. Unless otherwise stated, default values of $\lambda = 0.07$ and $n = 2 \times 10^{7}$ are used for R-TBS exp.  For R-TBS poly, we use $n = 2 \times 10^{7}$ unless otherwise stated and the other parameters take on values $(s, d, n', n, \lambda, \delta_{1}, \delta_{2}) = (2, 10, 2n, n, 0.1, 0.01, 0.001n)$. Finally, for the runtime experiments, each batch contains 10 million items unless otherwise stated.  

\revision{\subsection{Summary of Results}}
\revision{Before we dive into the detailed experiments, we first highlight some of the major takeaways from our empirical studies. }

\revision{\textbf{Sample Size Behavior: } Section~\ref{sec:samplesize} first compares the sample size behavior of T-TBS and R-TBS. The results empirically validate our previous assertions about T-TBS: although it is a much simpler algorithm than R-TBS, it suffers from sample overflow and/or underflow, especially when the mean batch size changes over time or the batch size fluctuates strongly. R-TBS maintains a bounded, relatively stable sample size throughout.}

\revision{\textbf{Performance of Distributed Implementation:} Section \ref{ssec:runtime} evaluates the various distributed implementation strategies described in Section \ref{sec:imp}. The most highly optimized implementation exhibits an almost 10x performance benefit relative to the most naive implementation.
	
\textbf{ML Applications:} Sections \ref{ssec:knn} through \ref{ssec:naivebayes} compare the the accuracy and robustness of R-TBS, simple sliding windows (SW), and uniform sampling (Unif) in three representative ML applications: a kNN classifier, a naive Bayes classifier, and a linear regression predictor. These applications span both parametric and nonparametric approaches, as well as both classification and prediction tasks. We find that R-TBS tends to have better accuracy and robustness than SW and Unif in the presence of reoccurring patterns in both single-change and periodic-change regimes. In addition, for linear regression, we compared retraining of models using R-TBS with online learning approaches adapted to batched streaming inputs. These online approaches, which work only for parametric models, can adapt to drastic changes more quickly than R-TBS, yielding better accuracy, but when either the changes are not as drastic or the ML model is complex,  with a large number of parameters, R-TBS performs better.}



\begin{figure*}[tbh]
 \centering
	\subfigure[Growing Batch Size]{
	   \label{fig:grow}\includegraphics[width=0.23\linewidth]{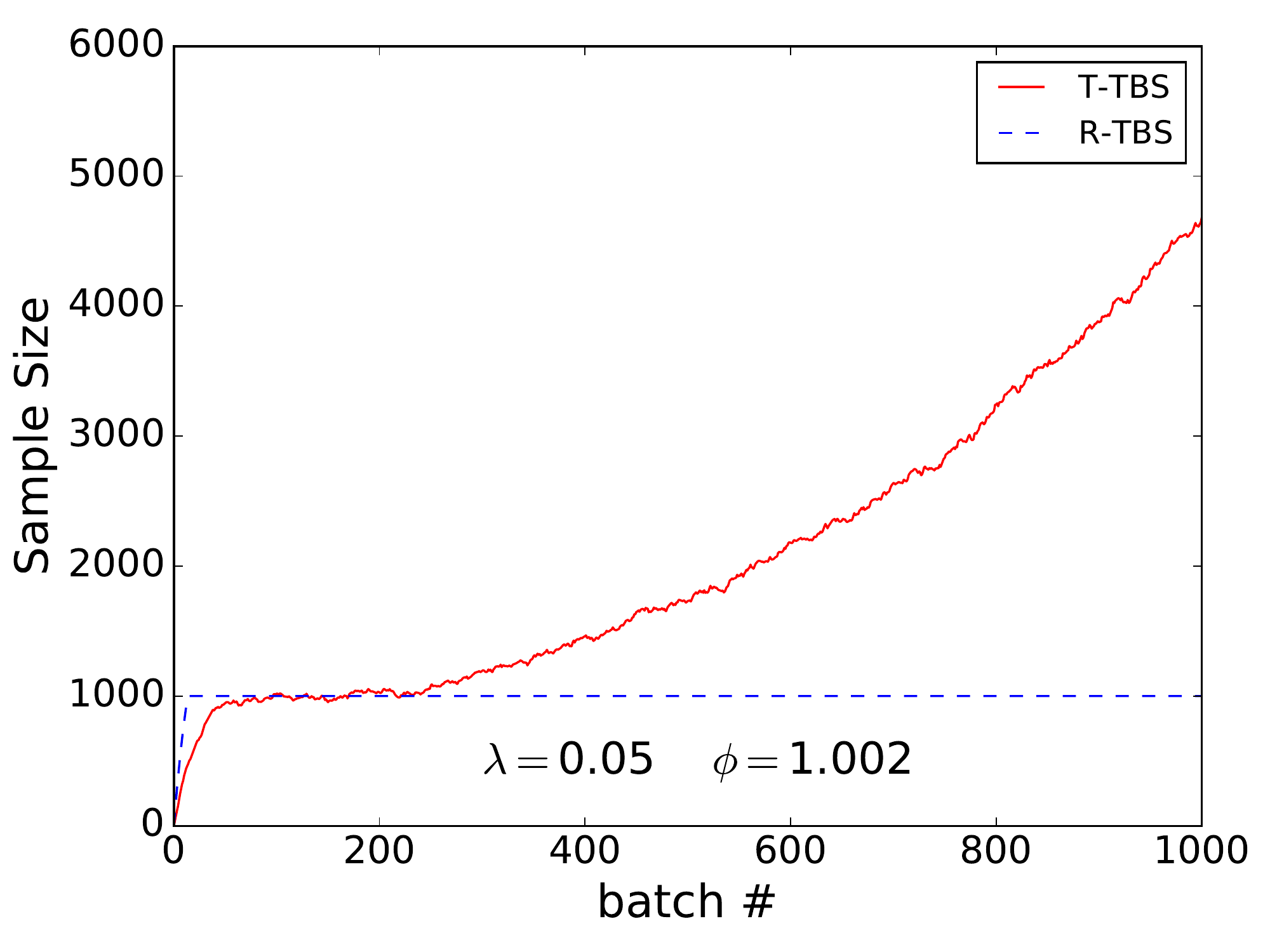}} 
	\subfigure[Stable Batch Size (Det.)]{
	   \label{fig:stableD}\includegraphics[width=0.23\linewidth]{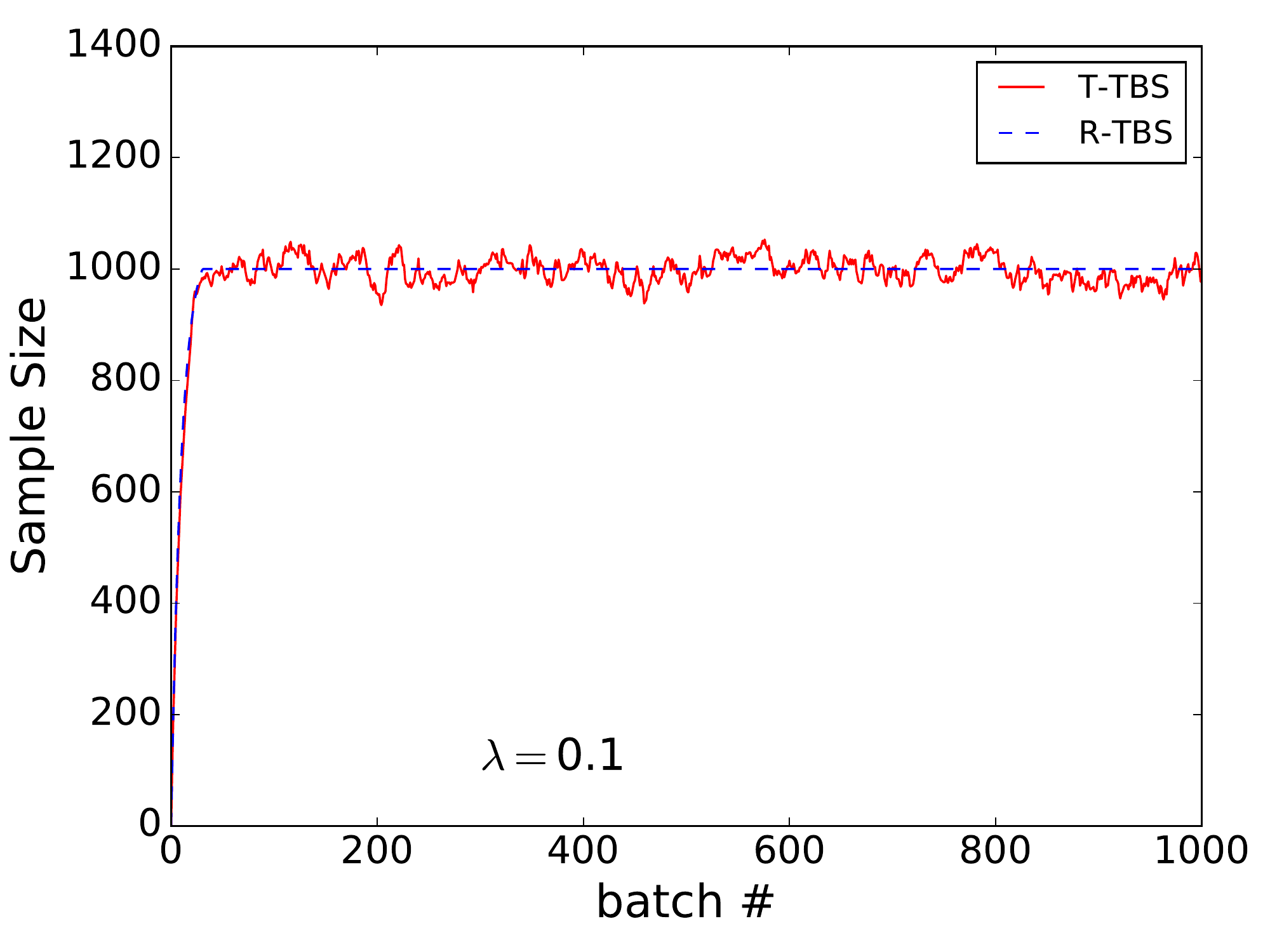}} 
	 \subfigure[Stable Batch Size (Unif.)]{
	   \label{fig:stableU}\includegraphics[width=0.23\linewidth]{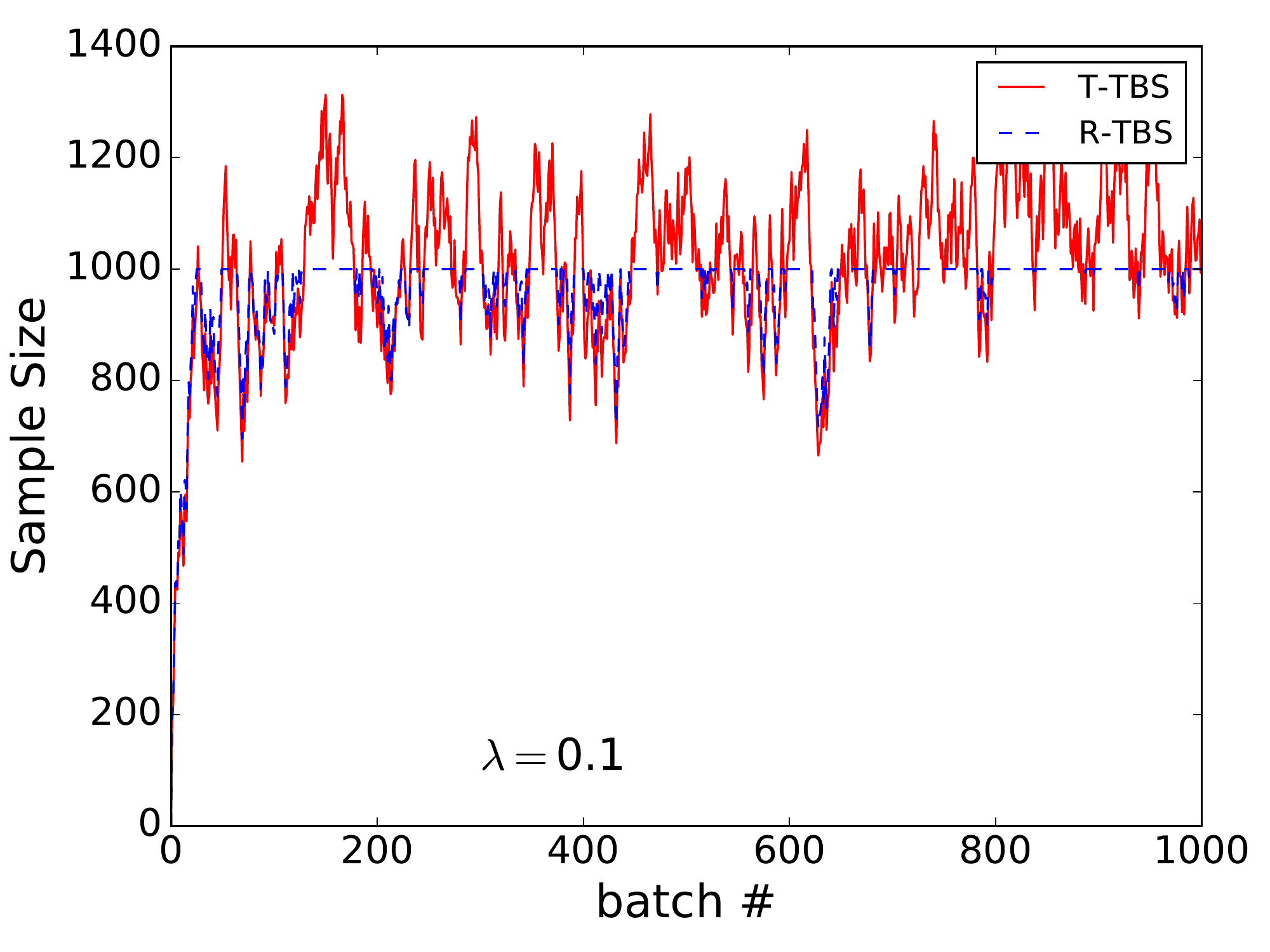}}
	 \subfigure[Decaying Batch Size]{
	   \label{fig:decay2}\includegraphics[width=0.23\linewidth]{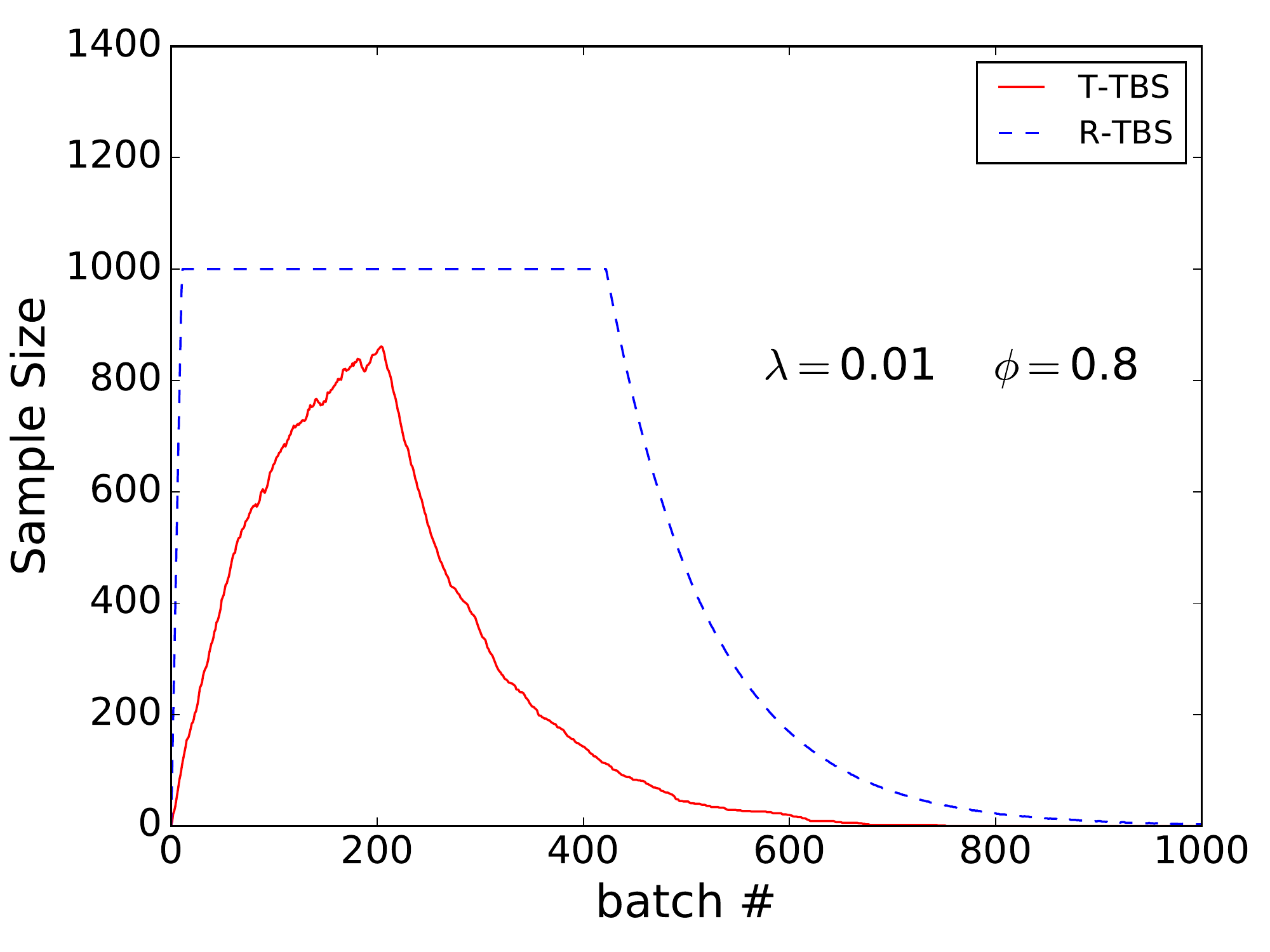}}
\caption{\label{fig:breakdown}Sample size behavior for T-TBS and R-TBS under exponential decay; $\lambda=$ decay rate and $\phi=$ batch size multiplier}
\end{figure*}

\begin{figure*}[tbh]
 \centering
	\subfigure[Growing Batch Size]{
	   \label{fig:growen}\includegraphics[width=0.23\linewidth]{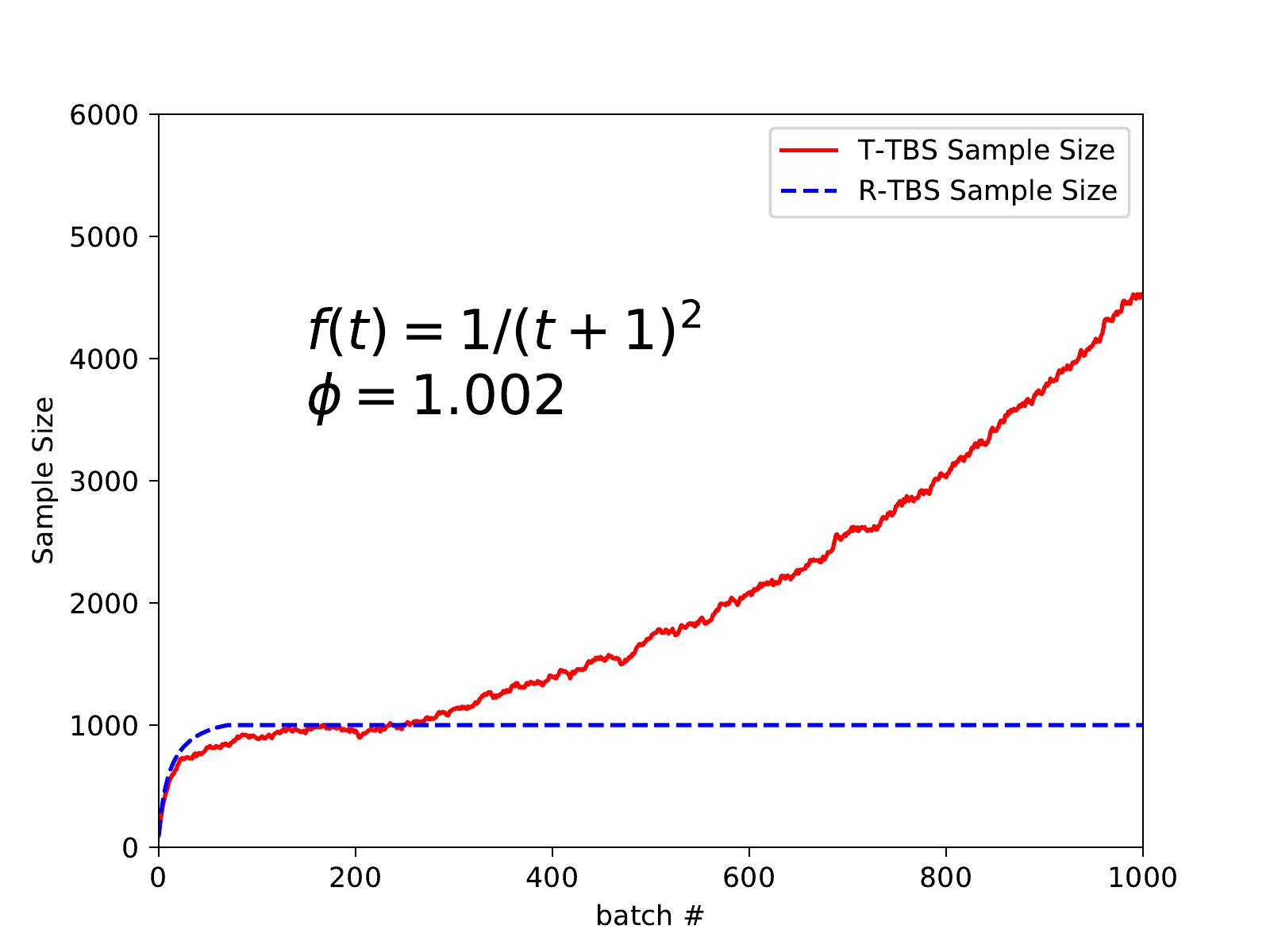}} 
	\subfigure[Stable Batch Size (Det.)]{
	   \label{fig:stableDen}\includegraphics[width=0.23\linewidth]{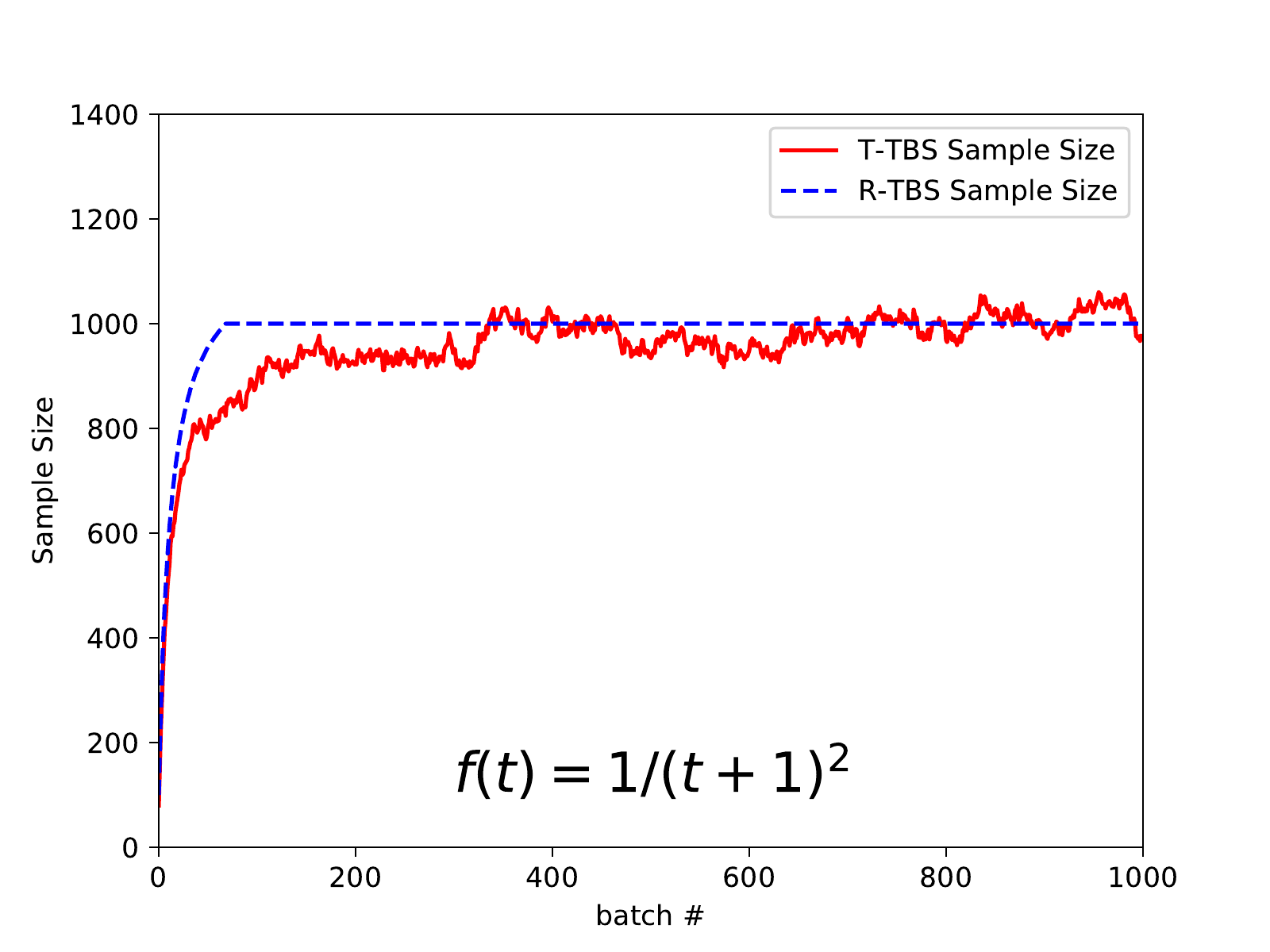}} 
	 \subfigure[Stable Batch Size (Unif.)]{
	   \label{fig:stableUen}\includegraphics[width=0.23\linewidth]{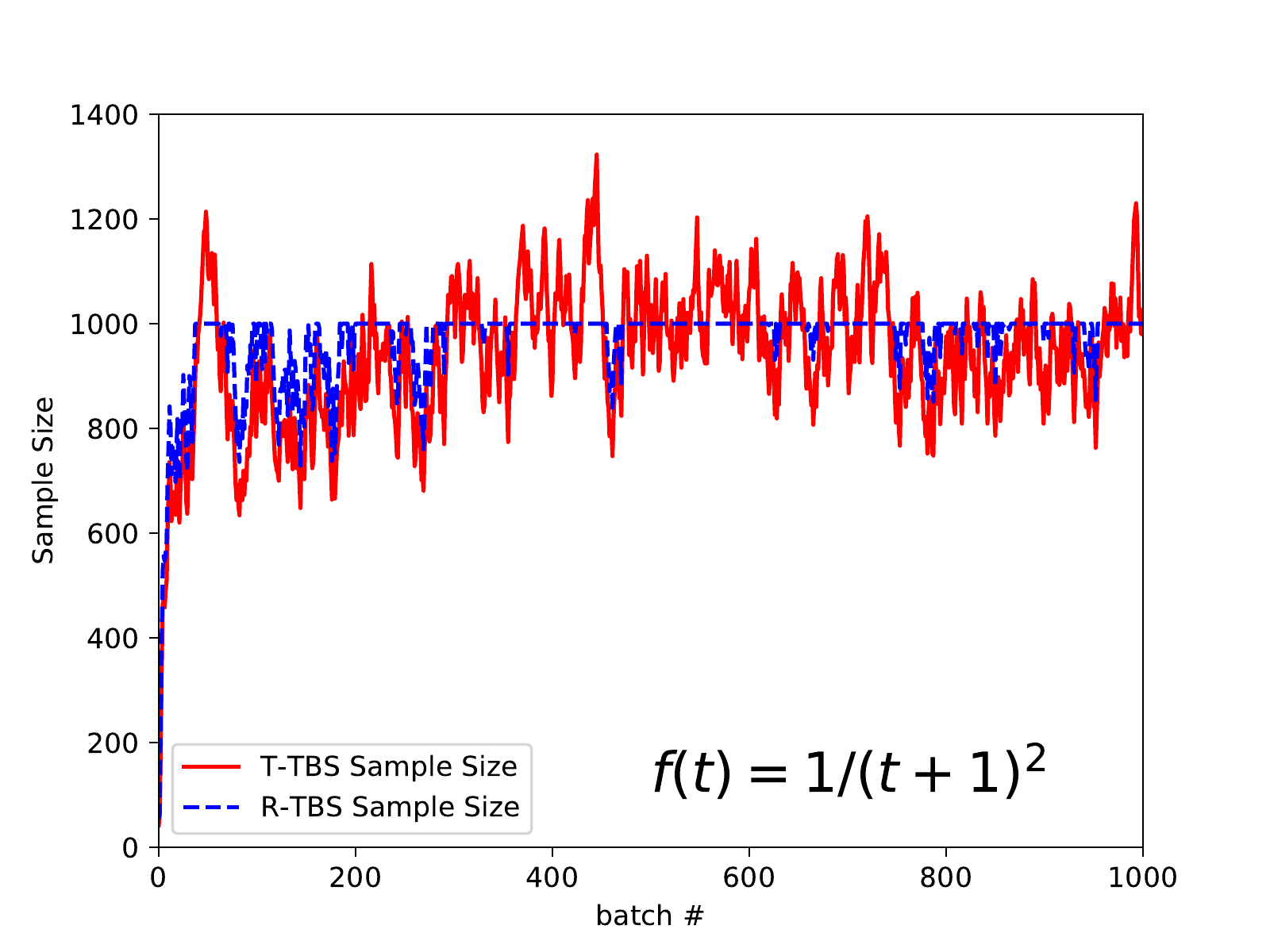}}
	 \subfigure[Decaying Batch Size]{
	   \label{fig:decay2en}\includegraphics[width=0.23\linewidth]{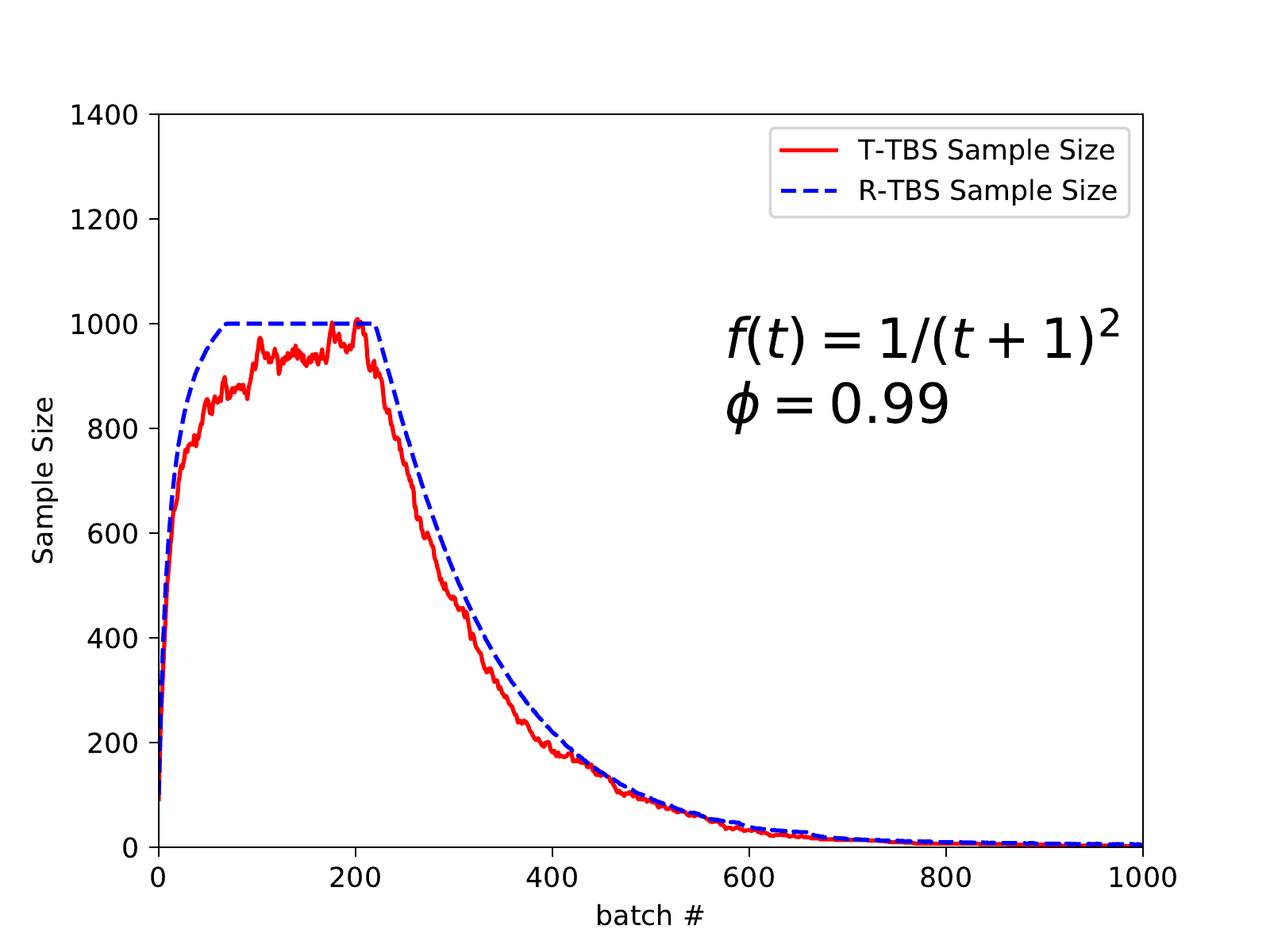}}
\caption{\label{fig:breakdowngen}Sample size behavior for T-TBS and R-TBS under quadratic decay; $\phi=$ batch size multiplier}
\end{figure*}

\revision{\subsection{Sample Size Behavior}\label{sec:samplesize}

We first compare the sample size behavior of T-TBS and R-TBS under a variety of batch size regimes. Throughout, we write $B_k=|\xB_k|$. Figures~\ref{fig:breakdown} and \ref{fig:breakdowngen} show sample size behavior over time for the two algorithms under exponential and quadratic decay. The key challenge to T-TBS is that the value of the mean batch size $b$ must be specified in advance; consequently, the algorithm cannot handle dynamic changes in $b$ without losing control of either the decay rate or the sample size.

In Figure~\ref{fig:grow}, for example, the (deterministic) batch size is initially fixed and the algorithm is tuned to a target sample size of 1000, with a decay rate of $\lambda=0.05$. At $k=200$, the batch size starts to increase (with $B_{k+1}=\phi B_k$ where $\phi=1.002$), leading to an overflowing sample, whereas R-TBS maintains a constant sample size.

Even in a stable batch-size regime with constant batch sizes (or, more generally, small variations in batch size), R-TBS can maintain a constant sample size whereas the sample size under T-TBS fluctuates in accordance with Theorems~\ref{th:recurr} and \ref{th:recurGam}; see Figure~\ref{fig:stableD} for the case of  a constant batch size $B_k\equiv100$ with $\lambda=0.1$.

Large variations in the batch size lead to large fluctuations in the sample size for T-TBS; in this case the sample size for R-TBS is bounded above by design, but large drops in the batch size can cause drops in the sample size for both algorithms; see Figure~\ref{fig:stableU} for the case of $\lambda=0.1$ and i.i.d.\ uniformly distributed batch sizes on $[0,200]$ so that $\mean[B_k]\equiv 100$. Similarly, as shown in Figure~\ref{fig:decay2}, systematically decreasing batch sizes will cause the sample size to shrink for both T-TBS and R-TBS. Here, $\lambda=0.01$ and, as with Figure~\ref{fig:grow}, the batch size is initially fixed and then starts to change at $k=200$, with $\phi=0.8$ in this case. This experiment---and others, not reported here, with varying values of $\lambda$ and $\phi$---indicate that R-TBS is more robust to sample underflows than T-TBS.

The results for polynomial decay in Figure~\ref{fig:breakdowngen} are similar to those for exponential decay. In the shrinking-batch-size scenario, note that R-TBS has a harder time maintaining a full sample under quadratic decay than under exponential decay (but still does better than T-TBS). The lower sample sizes are a direct consequence of the fact that different items decay at different rates, as discussed in detail in Section~\ref{sec:mono}.

So T-TBS has much more trouble maintaining a target sample size than R-TBS, especially when batch sizes fluctuate unpredictably. As pointed out in Section~\ref{sec:properties}, however, when T-TBS is applicable, it is much simpler and faster than R-TBS.}

\begin{figure*}[bth]
	\centerline{
		\epsfxsize=2.6in \epsffile{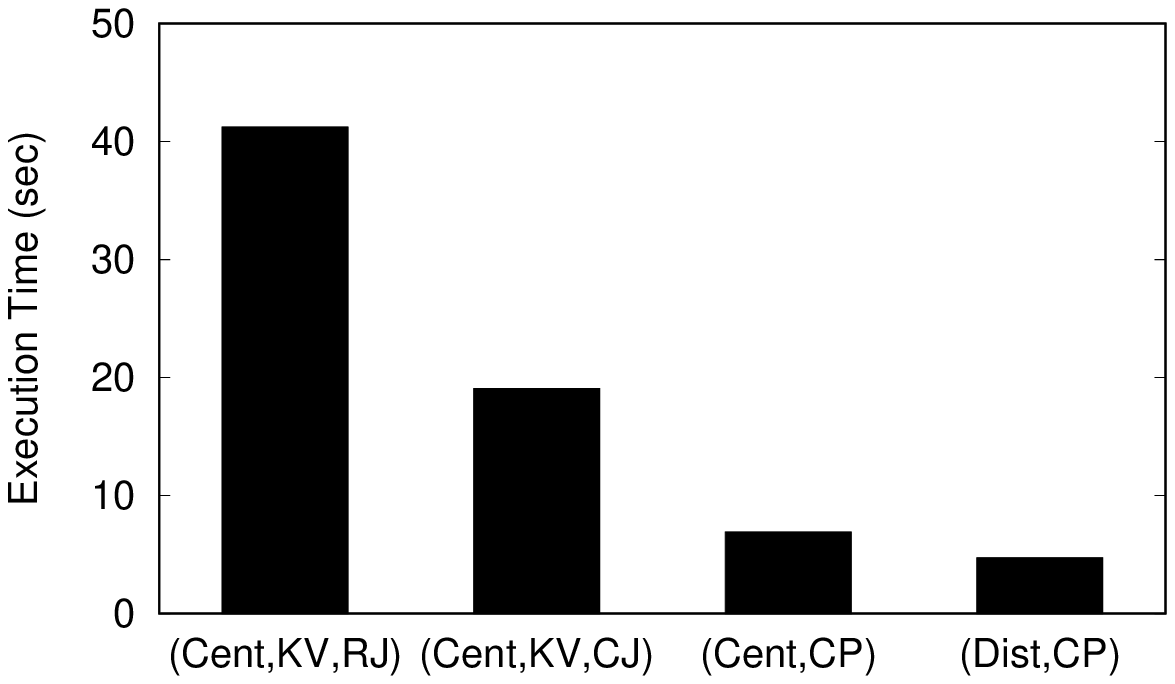}
		\hfill
		\epsfxsize=2.6in \epsffile{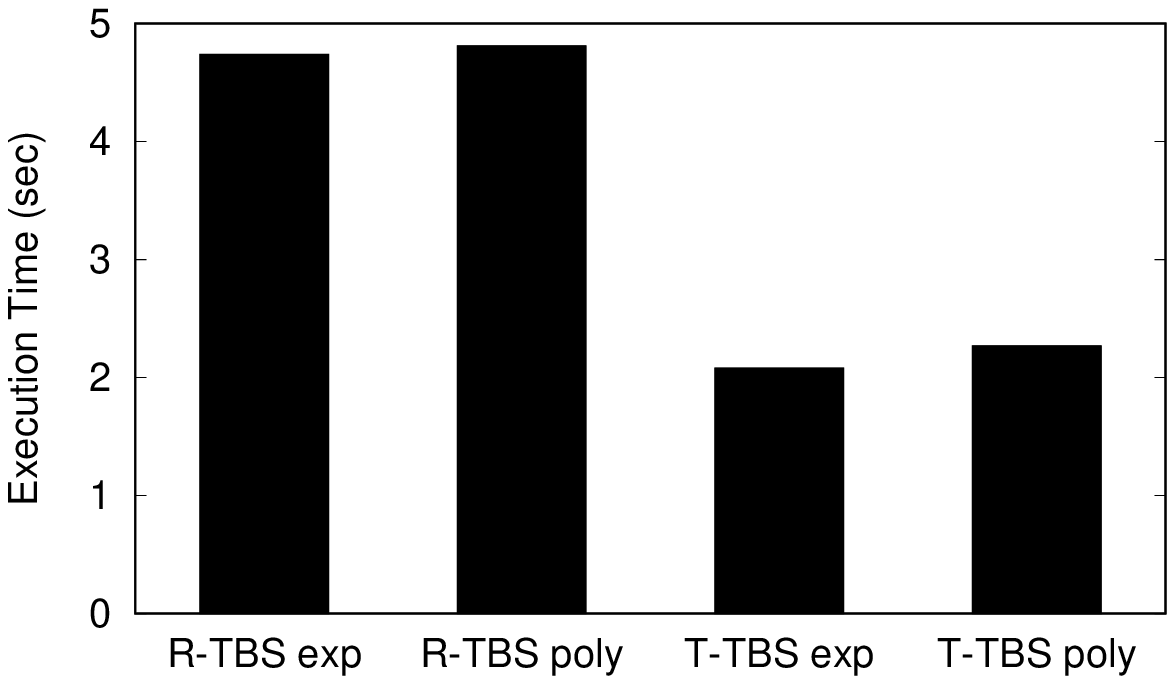}
	}
	\centerline{
		\parbox{2.6in}{\caption{Per-batch runtime comparison of different implementations of R-TBS with exponential decay}
			\label{fig:exp-runtime}}
		\hfill
		\parbox{2.6in}{\caption{Per-batch runtime of R-TBS and T-TBS with exponential decay and polynomial decay}
			\label{fig:exp-poly-runtime}}
	}
\end{figure*}

%

\begin{figure*}[bth]
	\centerline{
		\epsfxsize=1.8in \epsffile{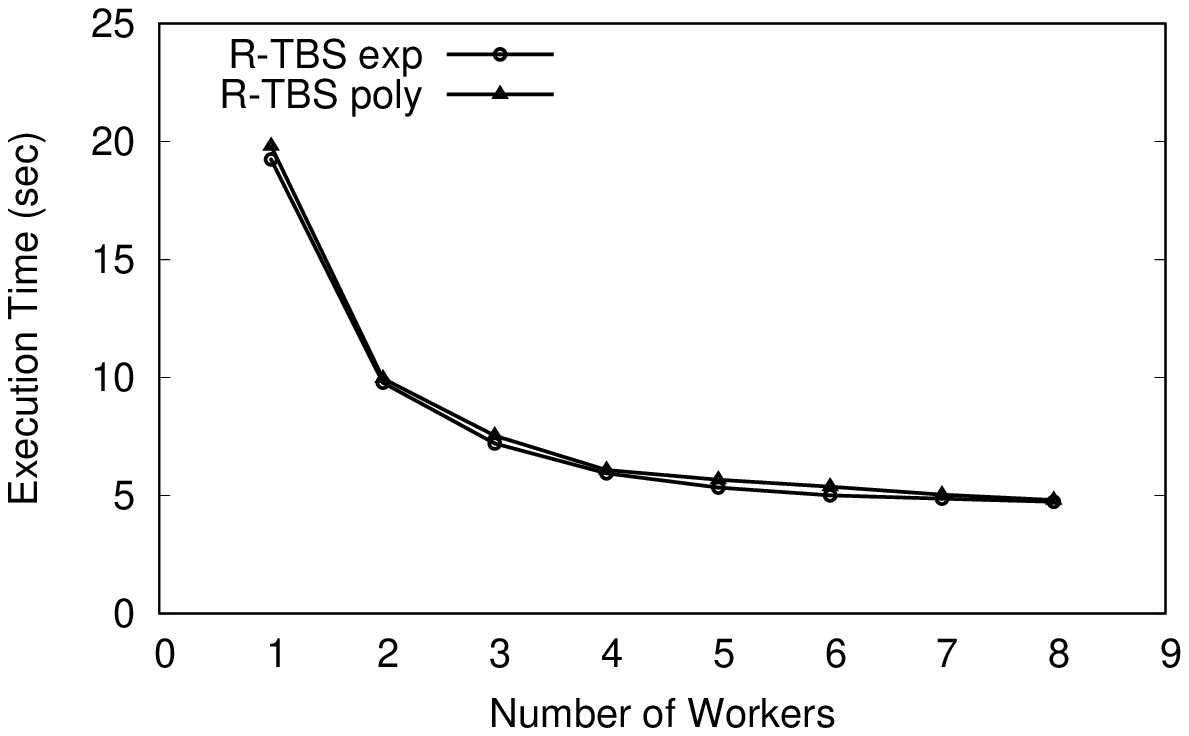}
		\hfill
		\epsfxsize=1.8in \epsffile{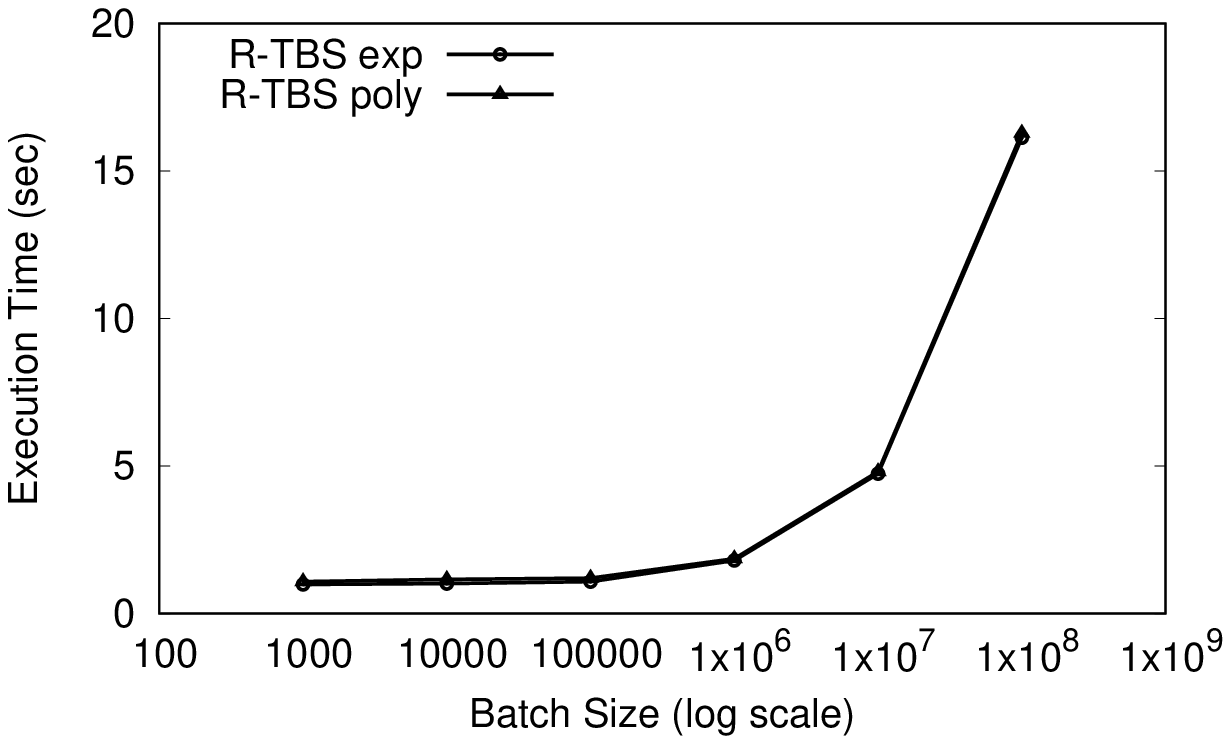}
		\hfill
		\epsfxsize=1.8in \epsffile{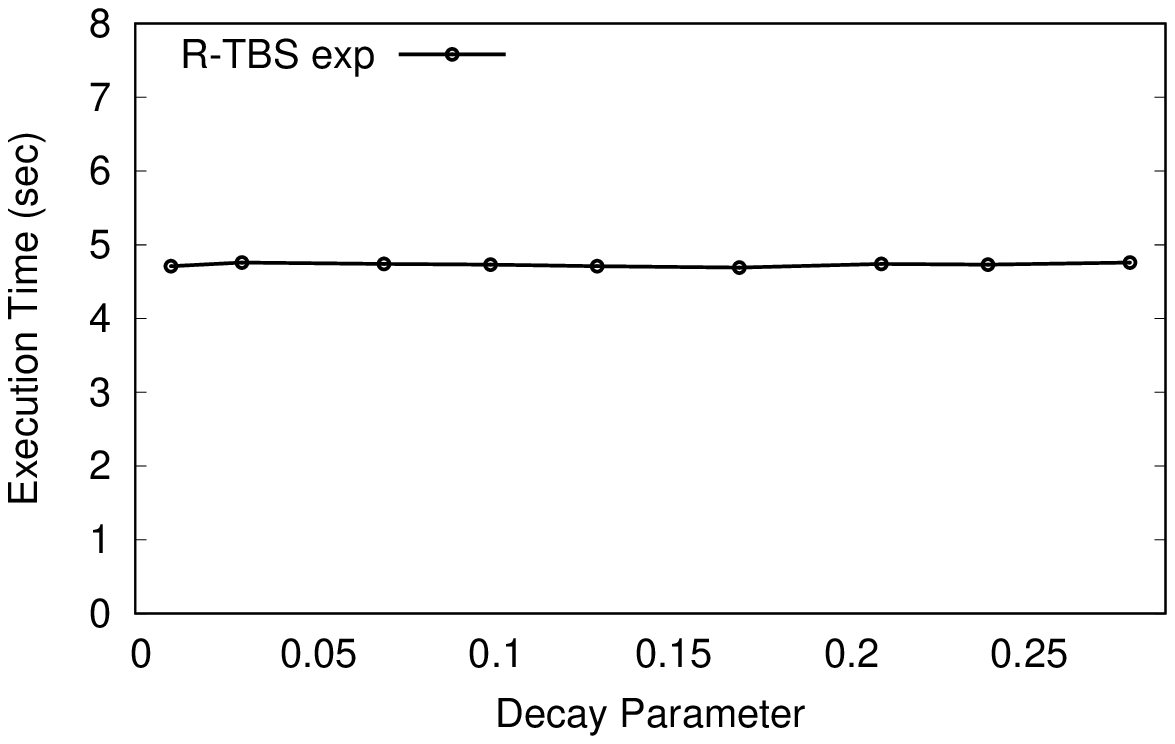}
	}
	\centerline{
		\parbox{1.8in}{\caption{Scale out of D-R-TBS}
			\label{fig:scale_out}}
		\hfill
		\parbox{1.8in}{\caption{Scale up of D-R-TBS}
			\label{fig:scale_up}}
		\hfill
		\parbox{1.8in}{\caption{Effect of Changing Decay Parameter}
			\label{fig:decay-runtime}}
	}
\end{figure*}

\subsection{Runtime Performance}\label{ssec:runtime}



\textbf{Comparison of TBS Implementations:}
\revision{Figure~\ref{fig:exp-runtime} shows the average runtime per batch for four implementations of distributed R-TBS exp with different design choices. The design decisions discussed are whether to use centralized or distributed decisions in choosing items to insert and delete (abbreviated as "Cent" and "Dist", respectively), whether to implement the reservoir using a key-value store or a co-partitioned reservoir scheme (abbreviated as "KV" and "CP"), and whether to subsample the incoming batch using the standard repartition join or using a copartitioned join (abbreviated as "RJ" and "CJ") under centralized decision scheme. These design decisions are discussed in more detail in Section \ref{sec:imp}}.

The first two implementations in Figure~\ref{fig:exp-runtime} both use the key-value store representation for the reservoir together with the centralized decision strategy for determining inserts and deletes. They only differ in how the inserted items are retrieved when subsampling the incoming batch. The first uses the repartition join, whereas the second uses the customized partitioner and co-located join. \revision{This optimization effectively cuts the runtime in half, but the KV representation of the reservoir still requires the inserted items to be written across the network to their corresponding reservoir location.} The third implementation employs the co-partitioned reservoir instead, resulting in a speedup of 2.75x. \revision{The fourth implementation further employs the distributed decision for choosing items to delete and insert. This yields a further 1.46x speedup.} The combination of co-partitioned reservoir scheme and the distributed decision making for inserting and deleting items always yields the best performance for R-TBS and T-TBS with different decay functions, so we use this combination for the remaining experiments. 

\revision{In Figure~\ref{fig:exp-poly-runtime}, we show the per batch runtimes for R-TBS and T-TBS with both exponential decay and polynomial decay.}
R-TBS exp and R-TBS poly have very similar runtime performance, with R-TBS exp being slightly faster; a similar observation holds for T-TBS. Since T-TBS is embarrassingly parallelizable, it is faster than R-TBS---though, as mentioned previously, the relative performance advantage of T-TBS is smaller than the result reported in~\cite{HentschelHT18} due to improvements in Spark. In any case, as discussed in Section~\ref{sec:ttbs}, T-TBS is faster, but only works under a very strong restriction on the data arrival rate, and can suffer from occasional memory overflows; see Figures~\ref{fig:breakdown} and~\ref{fig:breakdowngen}. 
In contrast, R-TBS has more robust sample-size behavior and works in realistic scenarios where it is hard to predict the data arrival rate.



\textbf{Scalability of R-TBS:}
Figure \ref{fig:scale_out} shows how R-TBS with exponential and non-exponential (quadratic in this case) decay functions scale with the number of workers. The two implementations have very similar performance. Initially, both versions of R-TBS scale out very nicely and see linear speedup from an increase in workers. However, beyond 4 workers, the marginal benefit from additional workers is small, because the coordination and communication overheads, as well as the inherent Spark overhead, become prominent. For the same reasons, in the scale-up experiment in Figure~\ref{fig:scale_up}, both runtimes stay roughly constant until the batch size reaches 1 million items and then increase sharply at 10 and 100 million items. This is because processing the streaming input and maintaining the sample start to dominate the coordination and communication overhead. With 8 workers, our implementation of R-TBS can handle a data flow comprising 100 million items arriving approximately every 16 seconds.

\revision{\textbf{R-TBS Runtime Robustness:} Figure~\ref{fig:decay-runtime} shows the impact of changes in the exponential decay parameter on runtime; as can be seen, the impact is negligible. Similar results hold for changes in the exponent for R-TBS poly (quadratic, cubic etc.). These results might seem counterintuitive, since changes in the decay parameter have a substantial impact on the number of items inserted into and deleted from the reservoir. However, in our optimized distributed implementation of R-TBS (with co-partitioned reservoir and distributed decision making), these are not expensive operations. All inserts and deletes happen locally and only affect local memory. In comparison, the cost of reading the incoming batch across the network or from disks, as well as the communication overhead between the Spark master and Spark workers, are much more expensive. Overall, the processing cost is dominated by the cost of reading in incoming batches, which is linear in the batch size, and thus the average runtime depends only on the expected batch size, regardless of the batch size variability. While not shown here, similar results hold for skew in the number of items per incoming batch.}


\subsection{Application: Classification using kNN}\label{ssec:knn}



Our first ML model is a kNN classifier, where a class is predicted for each item in an incoming batch by taking a majority vote of the classes of the $k$ nearest neighbors in the current sample, based on Euclidean distance; the sample is then updated using the batch. To generate training data, we first generate 100 class centroids uniformly in a $[0,80]\times [0,80]$ rectangle. Each data item is then generated from a Gaussian mixture model and falls into one of the 100 classes. Over time, the data generation process operates in one of two ``modes". In the ``normal" mode, the frequency of items from any of the first 50 classes is five times higher than that of items in any of the second 50 classes. In the ``abnormal" mode, the frequencies are five times lower. Thus the frequent and infrequent classes switch roles at a mode change. We generate each data point by randomly choosing a ground-truth class $c_i$ with centroid $(x_i,y_i)$ according to relative frequencies that depend upon the current mode, and then generating the data point's $(x,y)$ coordinates independently as samples from $N(x_i,1)$ and $N(y_i,1)$. Here $N(\mu,\sigma)$ denotes the normal distribution with mean $\mu$ and standard deviation $\sigma$.

In this experiment, the batch sizes are deterministic with $b=100$ items. The sample size for both R-TBS and uniform reservoir sampling (Unif) is 1000, and the sliding window (SW) contains the last 1000 items; thus all methods use the same amount of data for retraining. \revision{The hyperparameter $k$ is tuned individually for each of R-TBS, SW, and Unif sampling schemes; we choose the value that minimizes the misclassification percentage}. In each run, the sample is warmed up by processing $100$ normal-mode batches before the classification task begins. We test R-TBS with both exponential and shifted-polynomial decay functions as described in Section \ref{ssec:runtime}. We use $\lambda = 0.07$ for R-TBS exp, unless otherwise stated. For R-TBS poly, we use the parameter values of $(s, d, n', n, \lambda, \delta_{1}, \delta_{2}) = (2, 10, 2000, 1000, 0.1, 0.01, 1)$. We focus on two types of temporal patterns in the data, as described below. 


\textbf{Single change:} Here we model the occurrence of a singular event. The data is generated in normal mode up to $k=10$ (time is measured here in number of units after warm-up), then switches to abnormal mode, and finally at $t=20$ switches back to normal . As can be seen in Figure~\ref{fig:single}, the misclassification rate (percentage of incorrect classifications) for all sampling schemes increases from around \revision{15\%} to roughly \revision{40\%} when the distribution becomes abnormal. Both versions of R-TBS as well as SW adapt to the change, recovering to around \revision{13\%} misclassification rate after $t=16$, with R-TBS poly and SW adapting slightly better than R-TBS exp. In comparison, Unif does not adapt at all. But, when the distribution snaps back to normal, the error rate of SW rises sharply to \revision{40\%} before gradually recovering, whereas error rates of R-TBS exp and R-TBS poly stay low around \revision{13\% and 15\%} throughout. These results show that R-TBS is indeed more robust: although sometimes slightly more sluggish than SW in adapting to changes, R-TBS avoids wild fluctuations in classification error as with SW.

\begin{figure}[bht]
	\centering
	\subfigure[Single Event]{
		\label{fig:single}\includegraphics[width=0.35\linewidth]{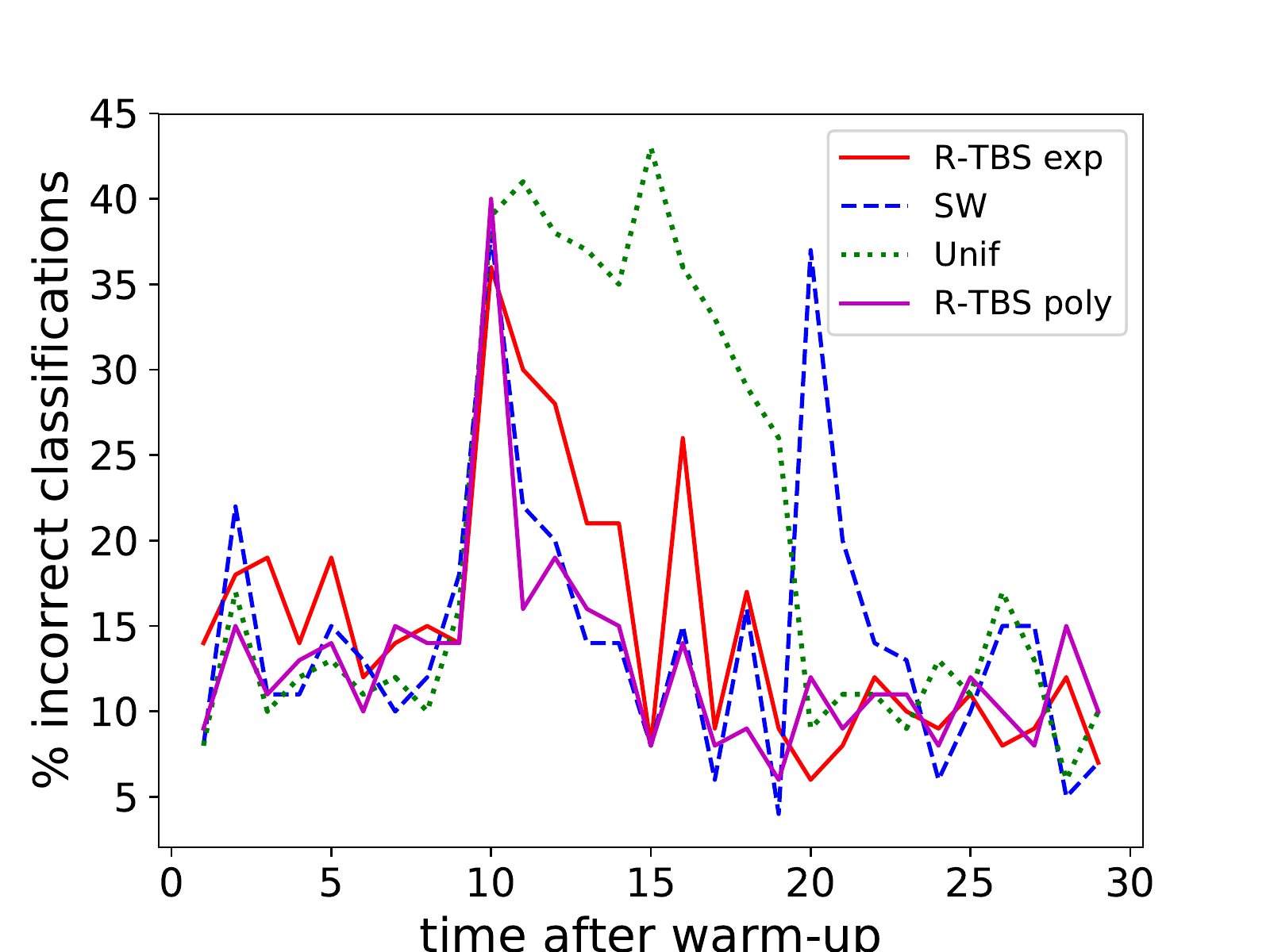}} 
	\subfigure[Periodic (10, 10)]{
		\label{fig:per1010}\includegraphics[width=0.35\linewidth]{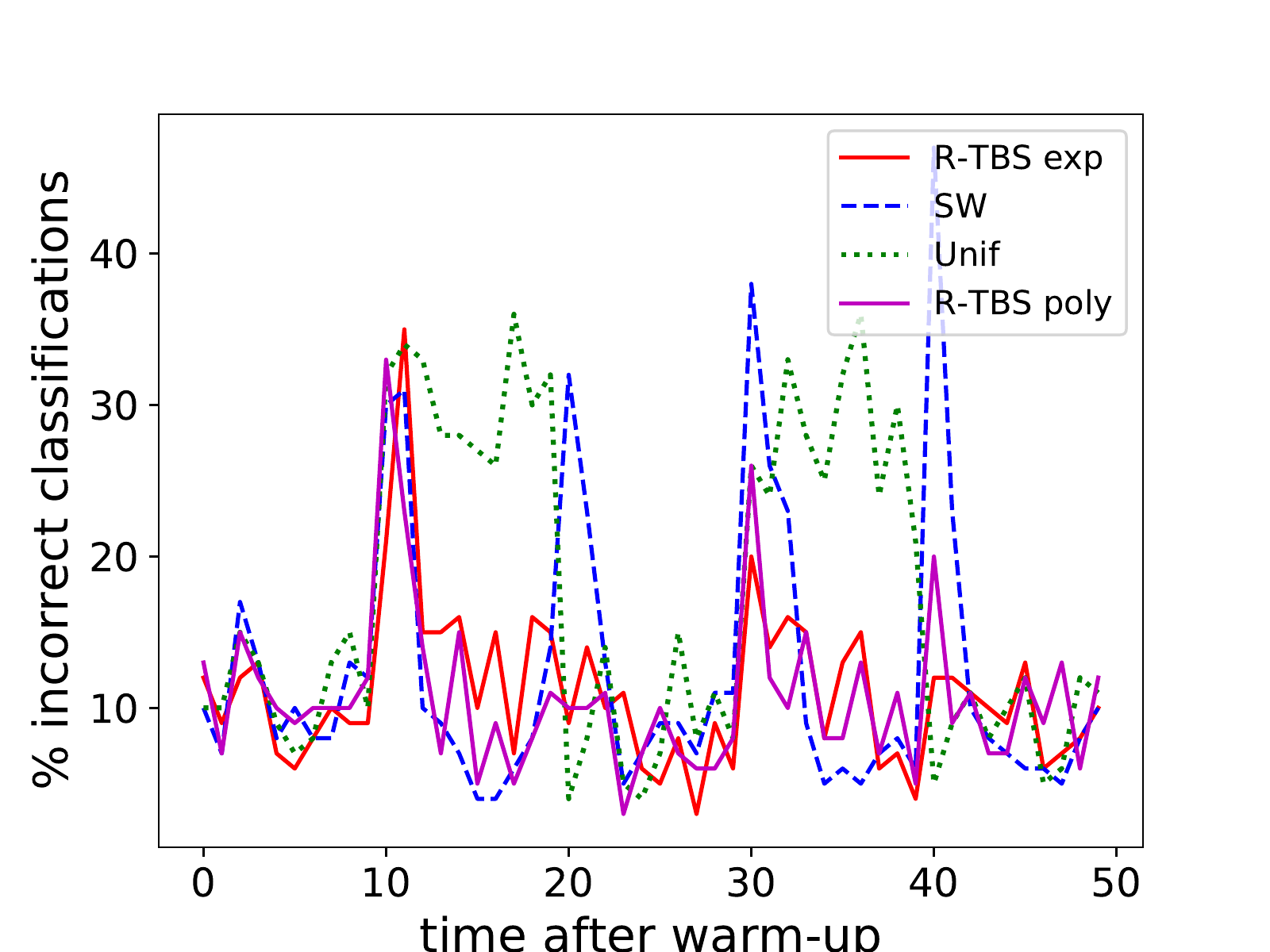}} 
	\caption{\label{fig:accuracy}Misclassification rate (percent) for kNN}
\end{figure}

\begin{figure}[bht]
 \centering
	 \subfigure[Periodic (20, 10)]{
	   \label{fig:per2010}\includegraphics[width=0.35\linewidth]{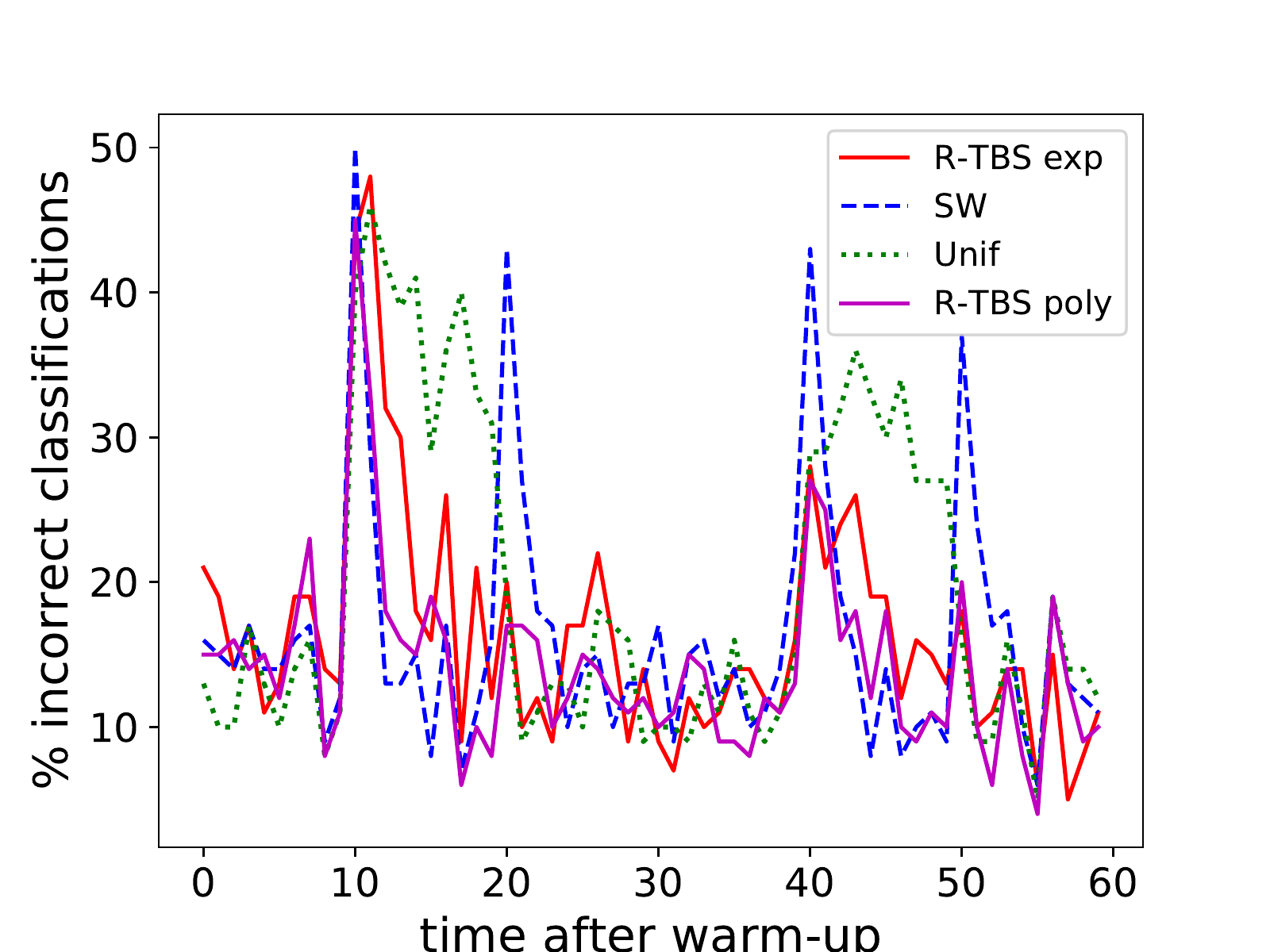}}
	 \subfigure[Periodic (30, 10)]{
	   \label{fig:per3010}\includegraphics[width=0.35\linewidth]{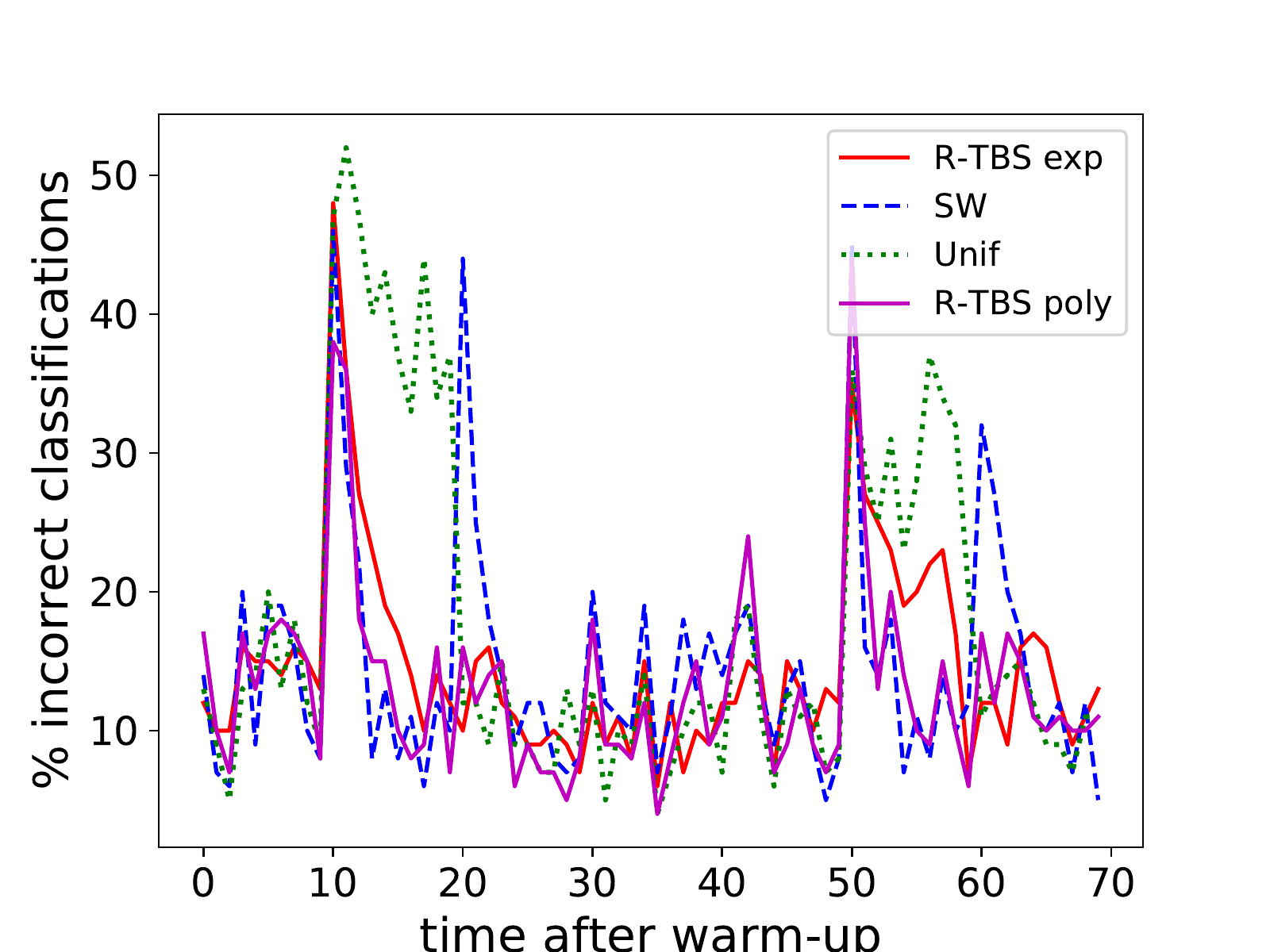}}
  \caption{Misclassification rate (percent) for kNN}
  \end{figure}

\textbf{Periodic change:} For this temporal pattern, the changes from normal to abnormal mode are periodic, with $\delta$ normal batches alternating with $\eta$ abnormal batches, denoted as $\text{Periodic}(\delta, \eta)$, or $P(\delta, \eta)$ for short. Figure~\ref{fig:per1010} shows the misclassification rate for $\text{Periodic}(10, 10)$. Experiments on other periodic patterns demonstrate similar results; see Figures~\ref{fig:per2010} and \ref{fig:per3010}. The robust behavior of R-TBS described above manifests itself even more clearly in the periodic setting. Note, for example, how both R-TBS versions react significantly better to the renewed appearances of the abnormal mode. Observe that the first 30 batches of $\text{Periodic}(10, 10)$ display the same behavior as in the single event experiment in Figure~\ref{fig:single}. We therefore focus primarily on the $\text{Periodic}(10, 10)$ temporal pattern for the remaining experiments.


\textbf{Robustness and Effect of Decay Parameter:} In the context of online model management, we need a sampling scheme that delivers high overall prediction accuracy and, perhaps even more importantly, robust prediction performance over time. Large fluctuations in the accuracy can pose significant risks in applications, e.g., in critical IoT applications in the medical domain such as monitoring glucose levels for predicting hyperglycemia events. To assess the robustness of the performance results across different sampling schemes, we use a standard risk measure called \emph{expected shortfall (ES)} \cite[p.~70]{es}. ES measures downside risk, focusing on worst-case scenarios. Specifically, the $z$\% ES is the average value of the worst $z$\% of cases. 

For each of 30 runs and for each sampling scheme, we compute the 10\% ES of the misclassification rate (expressed as a percentage) starting from $t=20$, since all three sampling schemes perform poorly (as would be expected) during the first mode change, which finishes at $t=20$. Table~\ref{tab:lambda} lists both the \emph{accuracy}, measured in terms of the average misclassification rate, and the \emph{robustness}, measured as the average 10\% ES, of the kNN classifier over 30 runs across different temporal patterns. To demonstrate the effect of the exponential decay parameter $\lambda$ and polynomial decay parameter $s$ on model performance, we also include numbers for different $\lambda$ and $s$ values in Table~\ref{tab:lambda}. Additionally, while not reported for reasons of space, we varied the shift parameter $d$ for polynomial decay and found similar results to those below for all values of $d$ between 6 and 14. 

In terms of accuracy, Unif is always the worst by a large margin. R-TBS and SW have similar accuracies, with R-TBS being slightly more accurate. On the other hand, for robustness, SW is almost always the worst, with ES ranging from \revision{1.4x to 2.5x} the ES of R-TBS. Mostly, Unif is also significantly worse than R-TBS, with ES ratios ranging from \revision{1.3x to 1.7x}. The only exception is the single-event pattern: since the data remains in normal mode after the abnormal period, time biasing becomes unimportant and Unif performs well. In general, R-TBS provides both better accuracy and robustness, with the best sampling schemes in terms of either accuracy or ES being R-TBS in all cases. Furthermore, this edge in accuracy and robustness is fairly stable across a wide range of $\lambda$ and $s$ values. Finally, when comparing R-TBS exp against R-TBS poly, we see that polynomial decay is slightly more accurate, but exponential decay has a slight edge with respect to robustness. \revision{Additionally, for both R-TBS exp and R-TBS poly, we find that the optimal decay value for ES is smaller than the optimal value for expected miss \%, as might be expected, since a lower value leads to greater retention of older items}. Overall, R-TBS exp and R-TBS poly yield comparable performance in ML models under dynamic data, but which is better depends on the experimental setup and so the choice of decay function will generally need to be driven by application requirements.

 
\textbf{Varying batch size:} We now examine model quality when the batch sizes are no longer constant. Overall, the results look similar to those for constant batch size. For example, Figure~\ref{fig:unif} shows results for a Uniform$(0,200)$ batch-size distribution, and Figure~\ref{fig:batchgrow} shows results for a deterministic batch size that grows at a rate of 2\% after warm-up. In both experiments, $\lambda=0.07$ for R-TBS exp, $s = 2$ for R-TBS poly, and the data pattern is $\text{Periodic}(10, 10)$. These figures demonstrate the robust performance of R-TBS in the presence of varying data arrival rates. Similarly, the average accuracy and robustness over 30 runs resembles the results in Table~\ref{tab:lambda}. For example, pick R-TBS exp with $\lambda=0.07$ and the $\text{Periodic}(10, 10)$ pattern. Then, the misclassification rate under uniform/growing batch sizes is \revision{1.17x/1.14x} that of R-TBS for SW, and \revision{1.39x/1.38x} for Unif. In addition, the ES is \revision{1.68x/1.95x} that of R-TBS for SW, and \revision{1.64x/1.60x} for Unif.

For the experiments in this section, we see that, in terms of model accuracy and robustness, both R-TBS exp and R-TBS poly perform similarly. We found that this observation holds generally, and so we will focus on R-TBS exp from now on.

\begin{table}[tbh]
\caption{Accuracy and robustness of kNN performance}\label{tab:lambda}
	\scriptsize{
		\begin{tabular}[0.3\linewidth]{|c|c|c|c|c|c|c|c|c|c|c|}
			\hline
			& \multicolumn{2}{c|}{
			\textbf{Single Event}} & \multicolumn{2}{c|}{\textbf{P(10,10)}} & \multicolumn{2}{c|}{\textbf{P(20,10)}} & \multicolumn{2}{c|}{\textbf{P(30,10)}} \\ \cline{2-9}
			$Alg. $  & Miss\% & ES  & Miss\% & ES &	Miss\% & ES & Miss\% & ES \\ \hline
			R-TBS exp: $\lambda = 0.05$ & 17.1 & {\color{red}\textbf{16.8}}  & 16.1 & 22.1 &	15.3 & 24.4 & 15.1 & 25.9 \\ \hline
			R-TBS exp: $\lambda = 0.07$ & 16.5 & 17.3 &	15.3 & {\color{red}\textbf{21.3}} &	14.9 & {\color{red}\textbf{24.0}} & 14.4 & 25.2 \\ \hline
			R-TBS exp: $\lambda = 0.10$ & 15.7 & 18.5 &	15.1 & 22.1 &	14.7 & 24.9 &	14.7 & 26.9 \\ \hline
			R-TBS poly: $s=1.8$ & 15.0 & 17.9 &	{\color{red}\textbf{13.9}} & 23.2 &	14.5 & 24.1 &	13.8 & {\color{red}\textbf{24.0}} \\ \hline
			R-TBS poly: $s=2.0$ & {\color{red}\textbf{14.9}} & 18.9 &	14.2 & 23.7 &	{\color{red}\textbf{14.1}} & 24.2 &	{\color{red}\textbf{13.7}} & 24.6 \\ \hline
			R-TBS poly: $s=2.2$ & 17.6 & 21.7 &	14.5 & 24.1 &	14.3 & 24.2 &	14.2 & 26.8 \\ \hline
			SW & 19.2 & 42.1 & 17.1 & 41.7 &	16.1 & 39.8 &	15.9 & 38.3 \\ \hline
			Unif & 21.3 & 18.3 &	25.4 & 34.8 &	19.6 & 35.7 &	19.0 & 35.8 \\ \hline
		\end{tabular}
	}
\end{table}

\begin{figure}[tbh]
	\centering
	\subfigure[Uniform Batch Size]{
		\label{fig:unif}\includegraphics[width=0.35\linewidth]{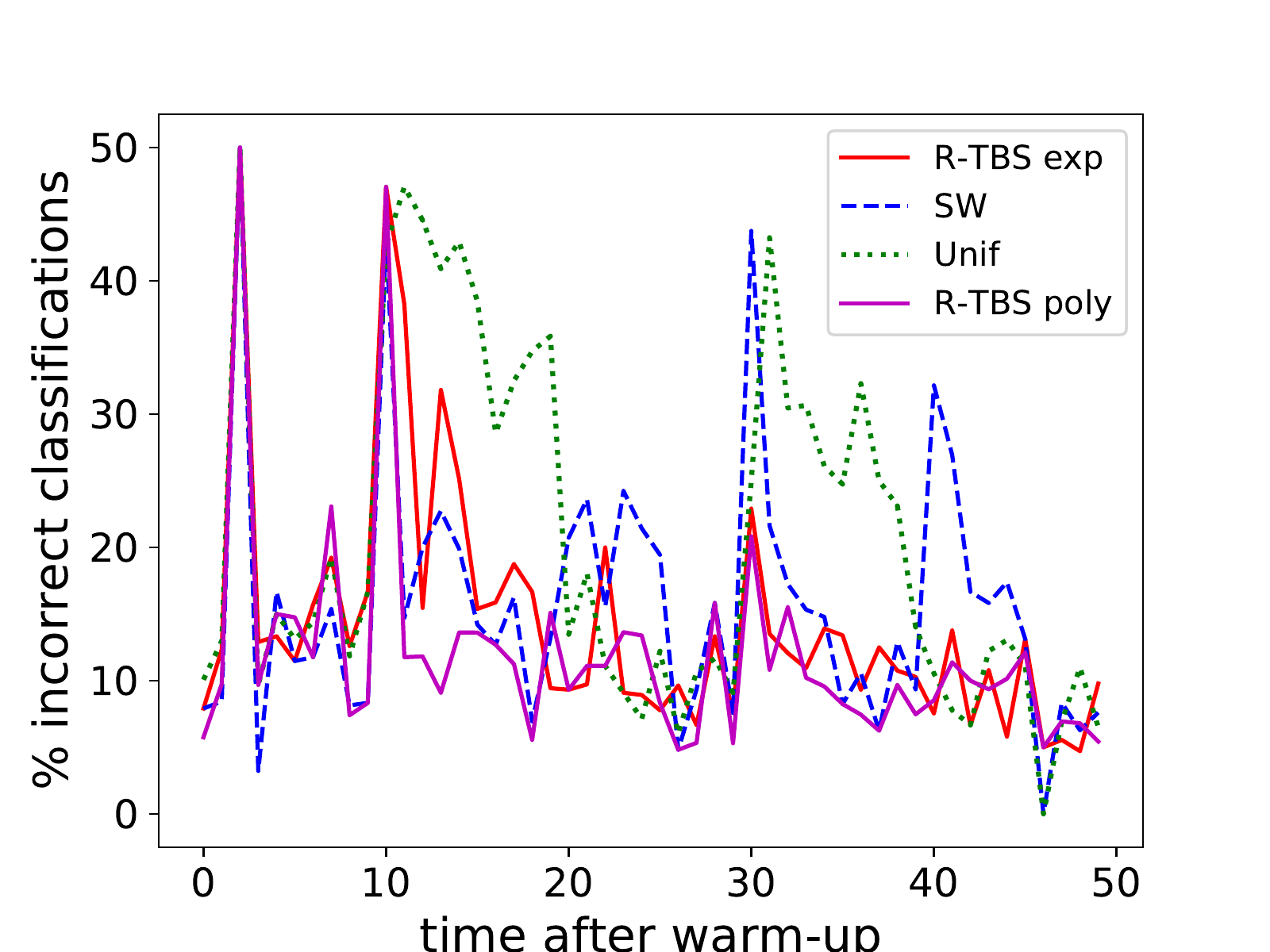}} 
	\subfigure[Growing Batch Size]{
		\label{fig:batchgrow}\includegraphics[width=0.35\linewidth]{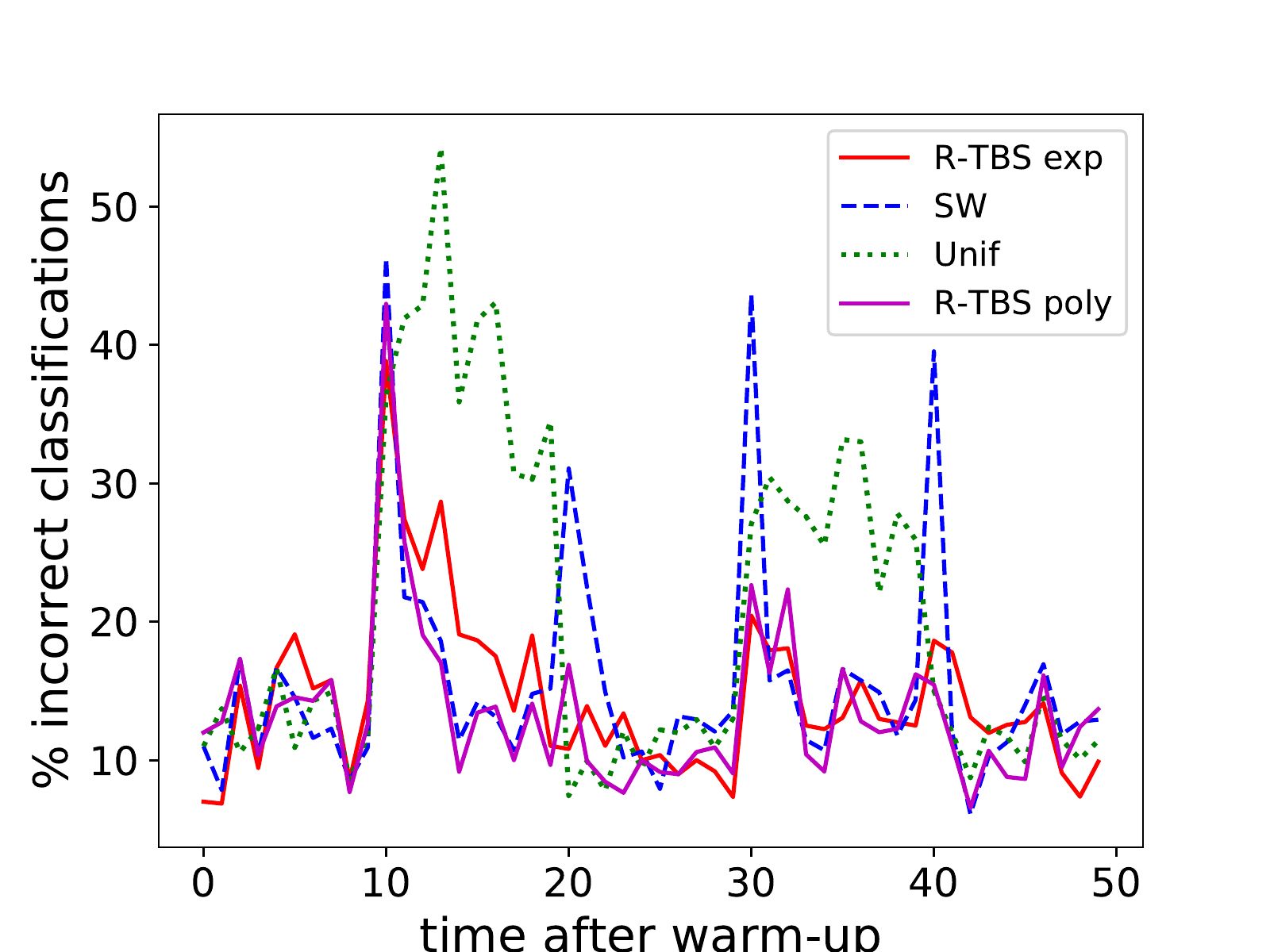}} 
	\caption{\label{fig:rates} Varying batch sizes for kNN classifier}
\end{figure}


\subsection{Application: Linear Regression}\label{ssec:regression}

We now assess the effectiveness of R-TBS for retraining regression models. The experimental setup is similar to kNN, with data generated in ``normal'' and ``abnormal'' modes. In both modes, data items are generated from the standard linear regression model $y = b_1x_1 + b_2x_2 + \eps$, with the noise term $\eps$ distributed according to a $N(0,1)$ distribution. In normal mode, $(b_1,b_2)=(4.2,-0.4)$ and in abnormal mode, $(b_1,b_2)=(-3.6,3.8)$. In both modes, $x_1$ and $x_2$ are generated according to a $\text{Uniform}(0,1)$ distribution. As before, the experiment starts with a warm-up of 100 ``normal'' mode batches and each batch contains 100 items. \revision{For each sampling scheme, the model is trained using a ridge-regression penalty hyperparameter that minimizes mean-squared error. }

\textbf{Saturated samples:} Figure~\ref{fig:regression} shows the performance of R-TBS, SW, and Unif under the $\text{Periodic}(10, 10)$ pattern with a maximum sample size of 1000 for each technique, and $\lambda = 0.10$ for R-TBS. We note that, for this sample size and temporal pattern, the R-TBS sample is always saturated. (This is also true for all of the prior experiments.) The results echo that of the previous section, with R-TBS exhibiting slightly better prediction accuracy on average, and significantly better robustness, than the other methods. The mean square errors (MSEs) across all data points for R-TBS, Unif, and SW are \revision{3.28, 4.41, 3.98} respectively, and their 10\% ES of the MSEs are \revision{8.16, 11.17, 10.65} respectively.


\begin{figure*}[tbh]
	\centering
	\subfigure[n=1000, Periodic(10,10)]{
		\label{fig:regression}\includegraphics[width=0.32\linewidth]{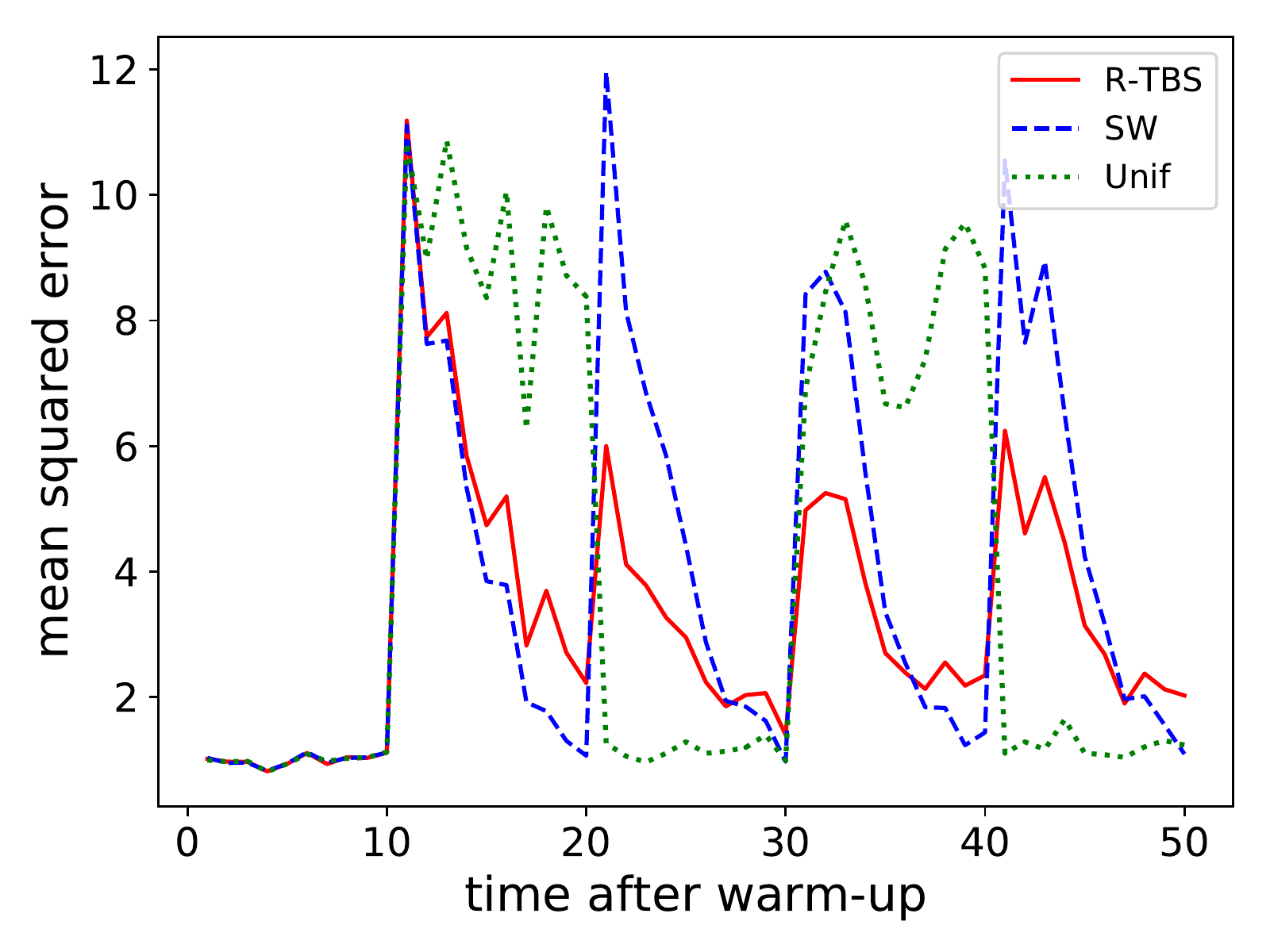}} 
	\subfigure[n=1600, Periodic(10,10)]{
		\label{fig:regression1500_10}\includegraphics[width=0.32\linewidth]{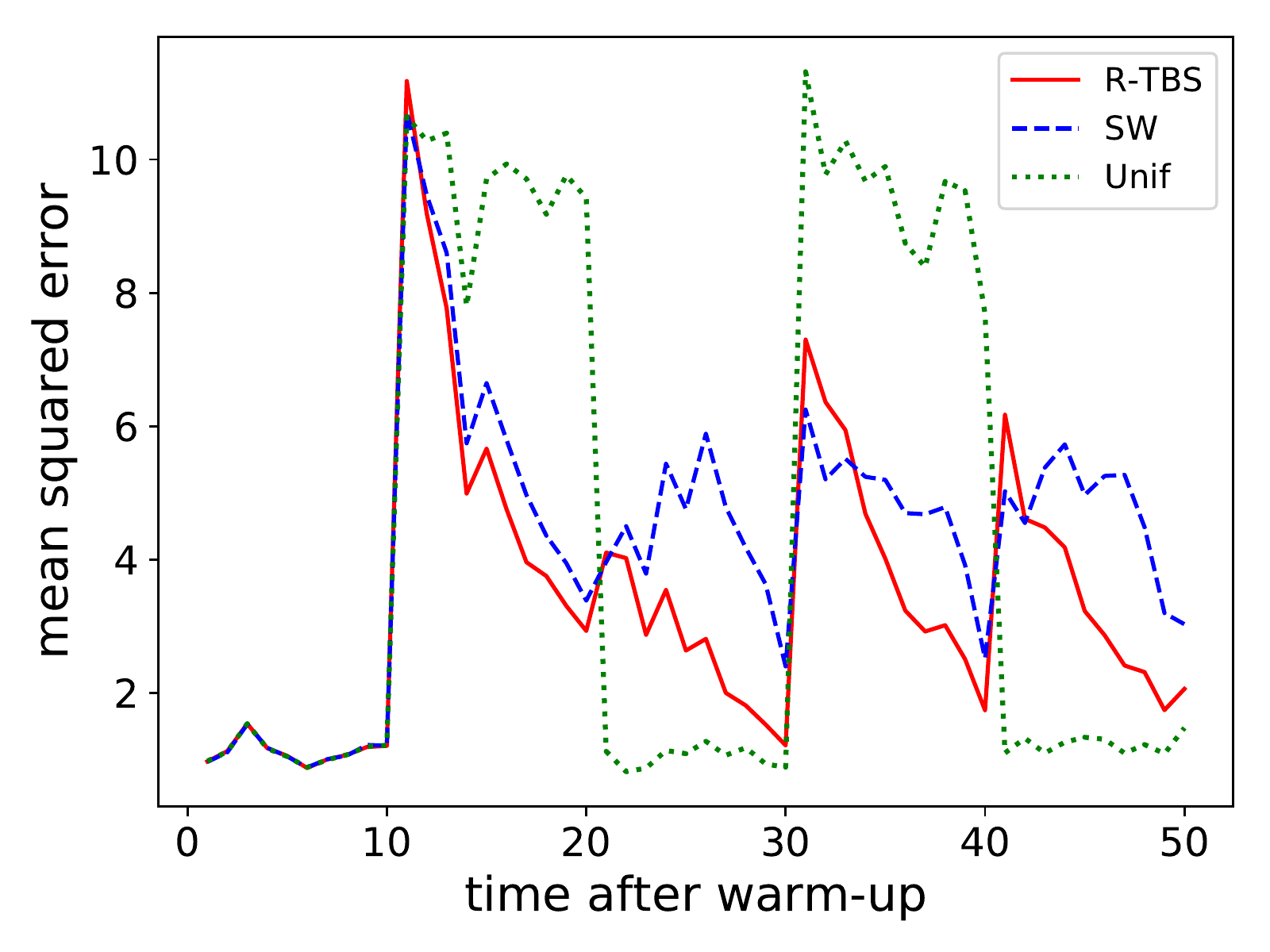}}
	\subfigure[n=1600, Periodic(16,16)]{
		\label{fig:regression1500_15}\includegraphics[width=0.32\linewidth]{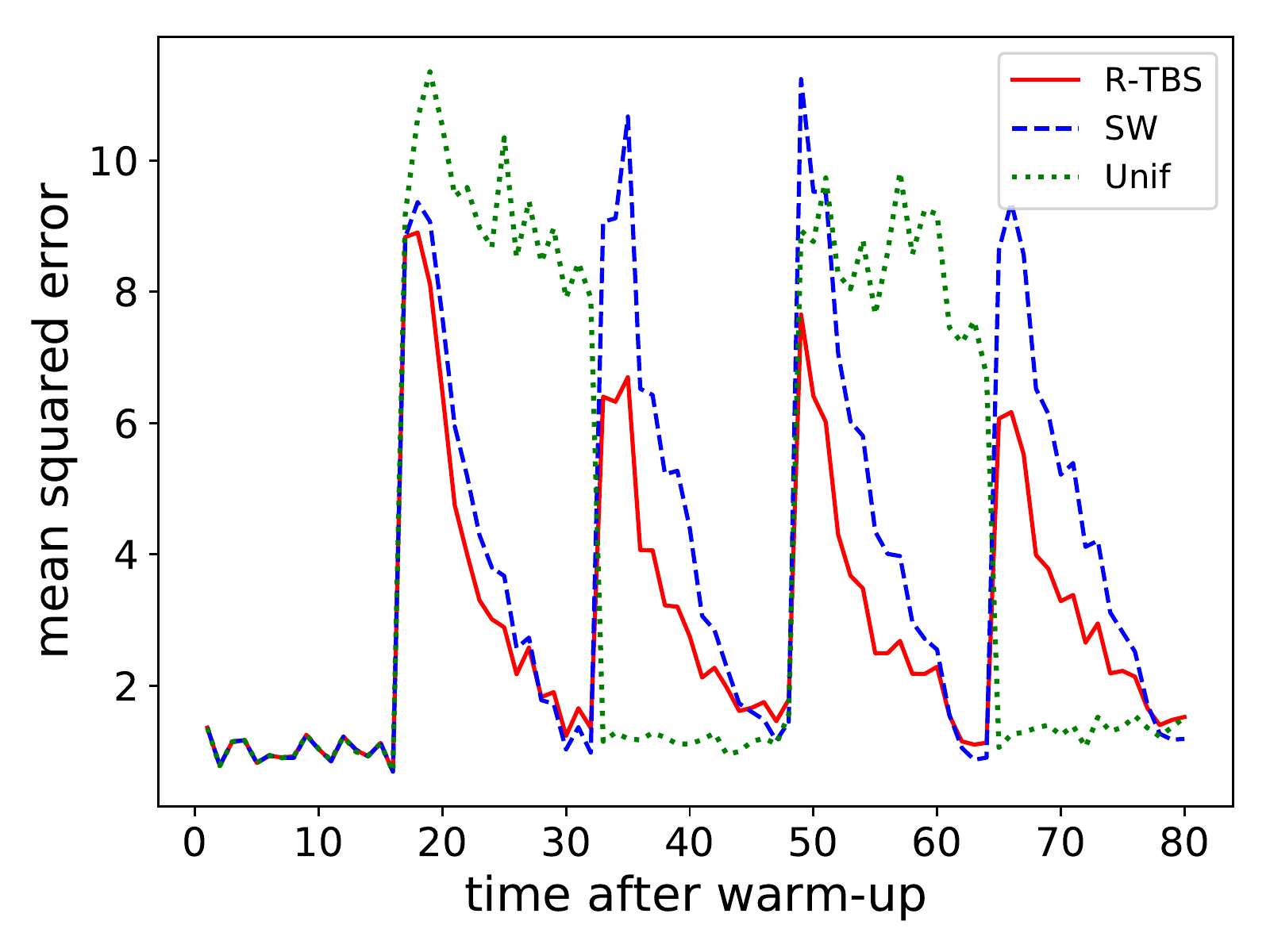}}  
	\caption{Mean square error for linear regression}
\end{figure*}

\begin{figure}[bht]
	\centering
	\subfigure[2 parameters, drastic data switch]{
		\label{fig:onlineBetter}\includegraphics[width=0.32\linewidth]{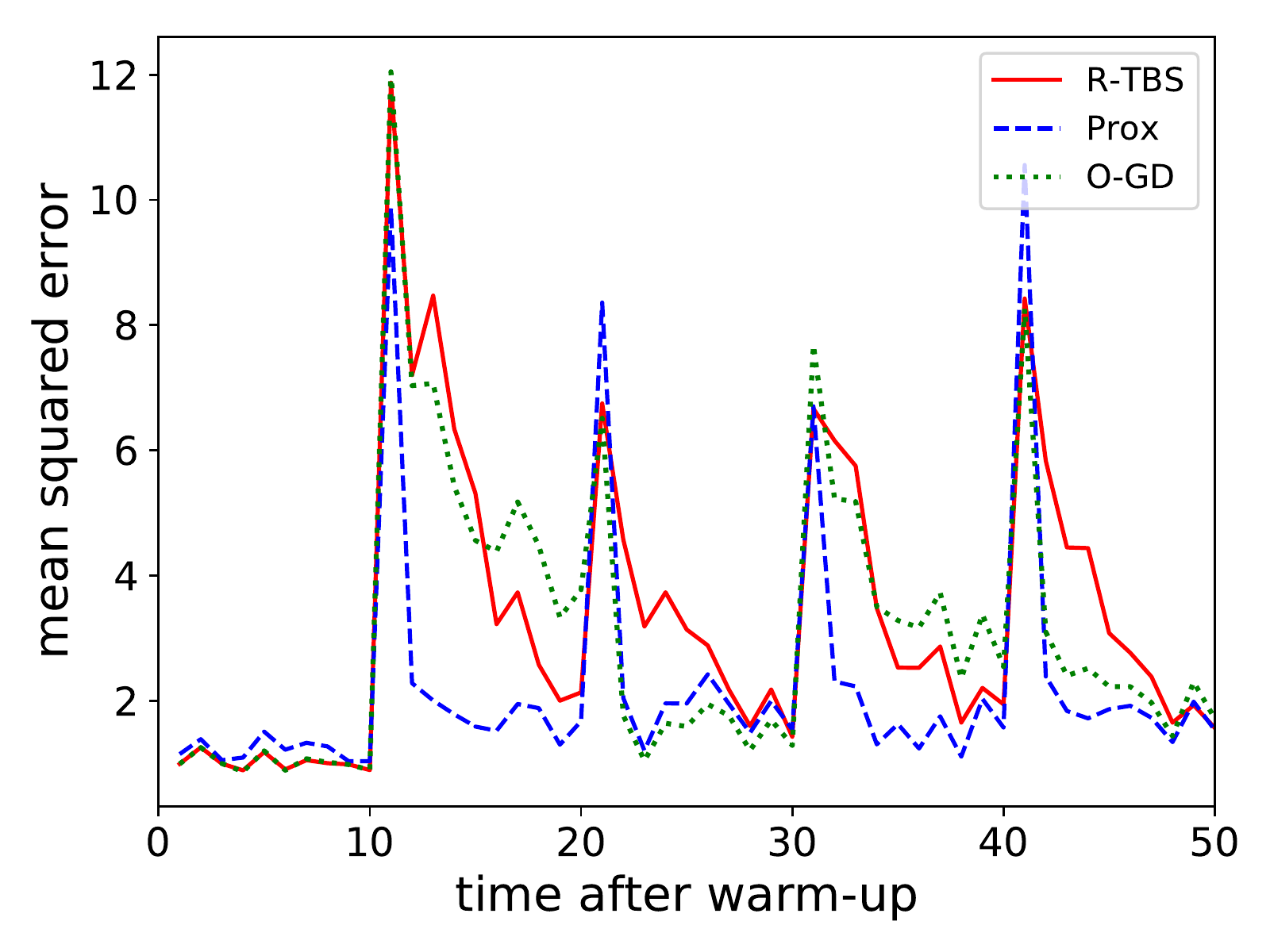}} 
	\subfigure[200 parameters, minor data switch]{
		\label{fig:tbsBetter}\includegraphics[width=0.32\linewidth]{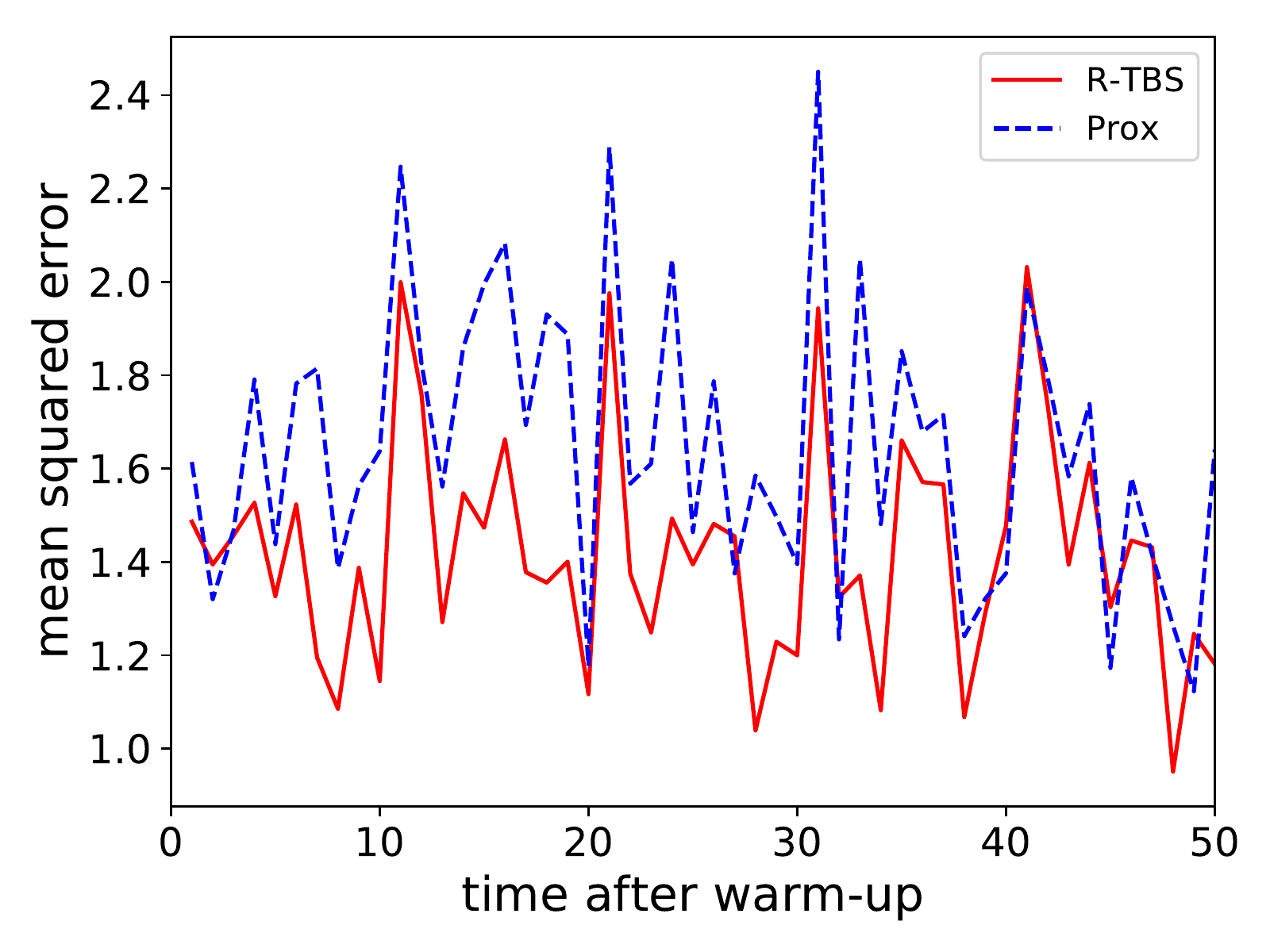}} 
	\subfigure[minor data switch, parameter sweep]{
		\label{fig:paramIncrease}\includegraphics[width=0.32\linewidth]{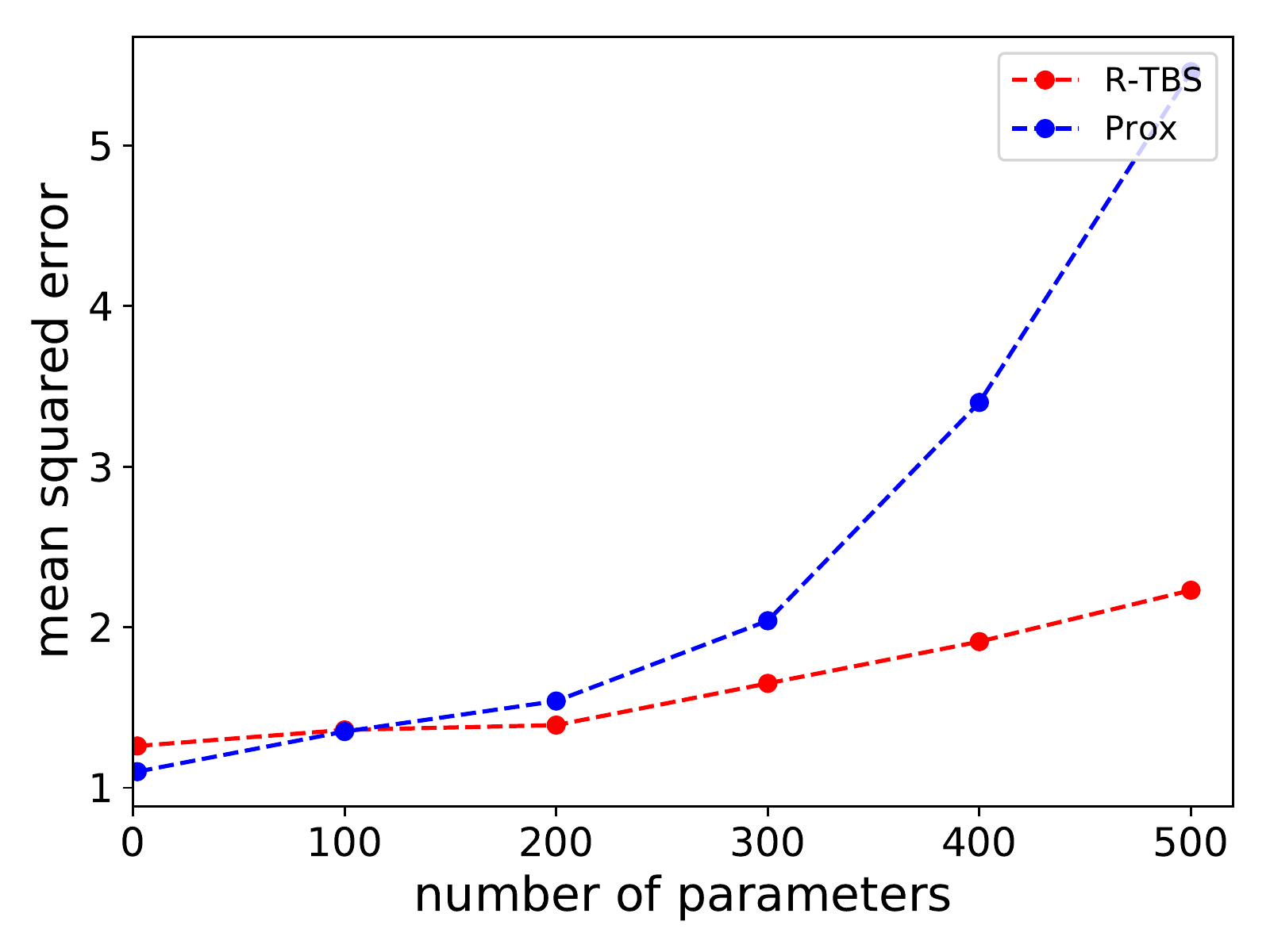}} 
	\caption{\label{fig:onlineVsTBS}Online methods react faster to changes but perform poorly as models get more complex}
\end{figure}

\textbf{Unsaturated Samples:} We now investigate the case of unsaturated samples for R-TBS. We increase the target sample size to $n=1600$. With a constant batch size of 100, and a decay rate $\lambda=0.10$, the reservoir of R-TBS is never full, stabilizing at 1051 items, whereas Unif and SW both have a full sample of 1600 items. 

For the Periodic$(10, 10)$ pattern, shown in Figure~\ref{fig:regression1500_10}, SW has a window size large enough to keep some data from older time periods (up to 16 batches ago), making SW's robustness comparable to R-TBS (ES of \revision{5.86} for SW and \revision{6.01} for R-TBS). However, this amalgamation of old data also hurts its overall accuracy, with MSE rising to \revision{4.16}, as opposed to \revision{3.43} for R-TBS. In comparison, the shape of R-TBS remains almost unchanged from Figure~\ref{fig:regression}, and Unif behaves as poorly as before. When the pattern changes to Periodic$(16, 16)$ as shown in Figure~\ref{fig:regression1500_15}, SW doesn't contain enough old data, making its prediction performance suffer from huge fluctuations again, and the superiority of R-TBS is more prominent. In both cases, R-TBS provides the best overall performance, despite having a smaller sample size. This backs up our earlier claim that more data is not always better. A smaller but more balanced sample with good ratios of old and new data can yield better prediction performance than a large but unbalanced sample.

\revision{\textbf{Comparison to Online ML Approaches: } For parametric models, online approaches provide an alternative to data sampling in order to periodically adapt the model without retraining on all prior data. We consider two techniques: online gradient descent (OGD), where the current batch is used in a single step of \textit{mini-batch gradient descent}, and proximal gradient methods (Prox), where a model is retrained on the current batch, but with a penalty term on the distance between the parameter values for the new and previous models; the idea is to prevent the model from changing too drastically at any given time step. We adapt both algorithms to batched streaming input, i.e., models are updated based on all the items in an incoming batch, instead of one item at a time.

As a first experiment, we use the same regression setup as before. We individually tune hyperparameters for each of the three models---$\lambda$ for R-TBS, the learning rate for OGD, and the proximal distance parameter for the proximal gradient approach---to minimize average $L2$ error. Figure \ref{fig:onlineBetter} shows a sample result. Prox performs the best, as it learns the new pattern in just a single round. In comparison, R-TBS and OGD learn the new pattern a bit more slowly. The average mean squared error for R-TBS, Prox, and OGD are 3.28, 2.04, and 3.05, respectively. 

Our second experiment incorporates two changes. First, we scale down the regression coefficients $b$ for each mode by a factor of $1/3$, which has the effect of reducing the jump in error at a mode change from 10x normal error to only 2x. Additionally, we add another 198 parameters to the regression model, thus making the model harder to learn. 
As seen in Figure~\ref{fig:tbsBetter}, both R-TBS and Prox still struggle at mode changes, but in non-change rounds R-TBS doesn't vary as wildly as Prox. OGD, not shown, performs so poorly that we didn't include it in the graph. This is due to the fact that it is hard to learn a pattern with many parameters in just a single pass over the data. The MSE values for this experiment are 1.39, 1.54, and 6.82 for R-TBS, Prox, and OGD, respectively. Figure~\ref{fig:paramIncrease} shows the effect of increasing the number of stable parameters in the model, with Prox performing better with fewer parameters and R-TBS performing better with a larger number of parameters. OGD again performs so poorly that it is left off the graph. Overall, from these experiments, along with others not shown, we observe that online models are more sensitive and react faster to drastic changes, but a data sampling approach based on R-TBS is more robust for milder changes and complex models with a large number of parameters.}

\begin{figure*}[tbh]
	\centering
		\includegraphics[width=0.35\linewidth]{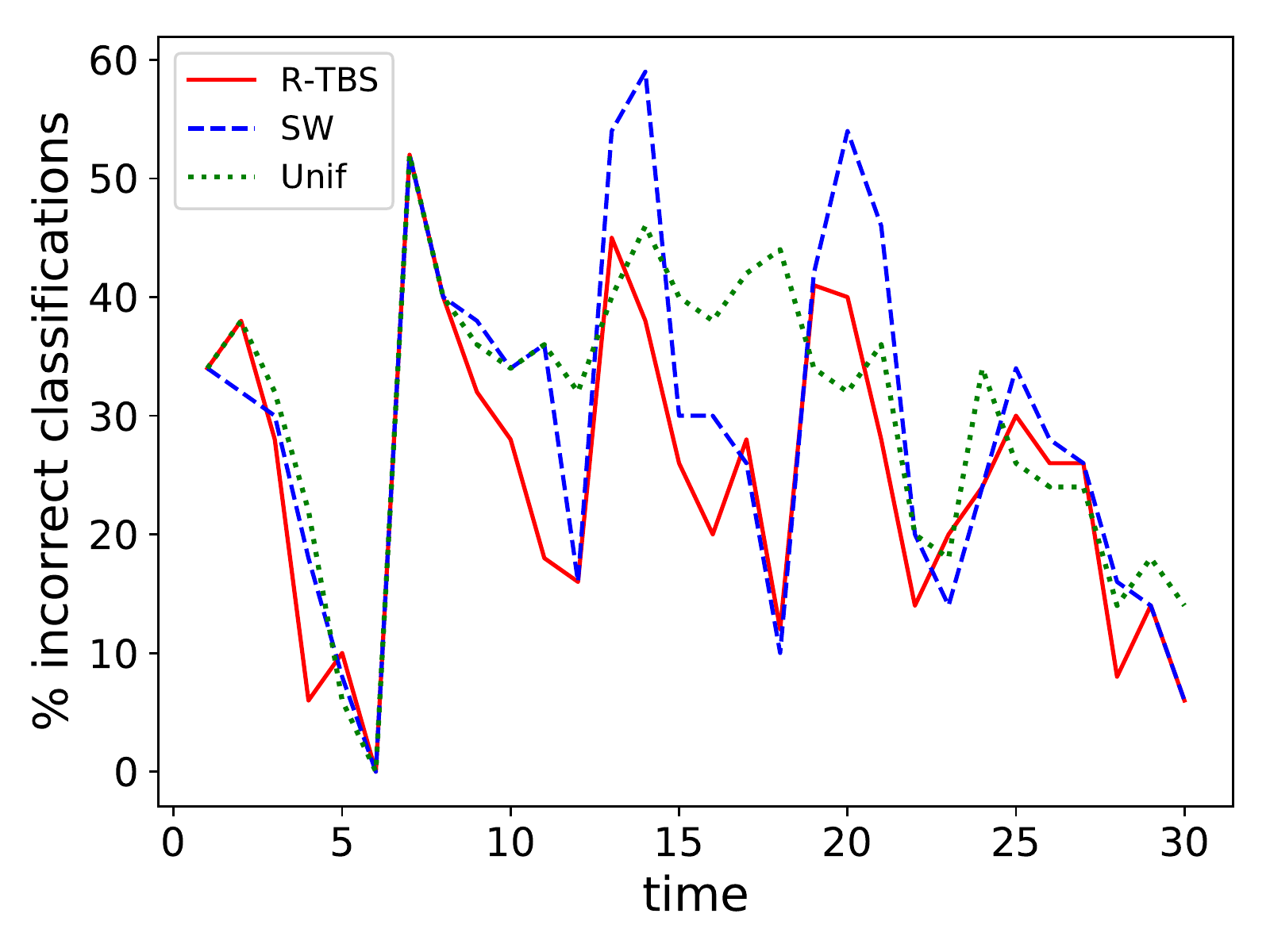} 
	\caption{Misclassification rate (percent) for Naive Bayes}\label{fig:NBaccuracy}
\end{figure*}

\subsection{Application: Naive Bayes Classifier}\label{ssec:naivebayes}

In our final experiment, we evaluate the performance of R-TBS for retraining Naive Bayes models with the Usenet2 dataset (\url{mlkd.csd.auth.gr/concept_drift.html}), which was used in~\cite{KatakisTV08} to study classifiers coping with recurring contexts in data streams. This dataset contains a stream of 1500 messages on different topics from the 20 News Groups Collections~\cite{Lichman2013}. They are sequentially presented to a simulated user who marks whether a message is interesting or not. The user's interest changes after every 300 messages. More details of the dataset can be found in~\cite{KatakisTV08}.



Following~\cite{KatakisTV08}, we use Naive Bayes with a bag-of-words model, and set the optimal parameters for SW with maximum sample size of 300 and batch size of 50. Since this dataset is rather small and contexts change frequently, we use the optimal value of $0.35$ for $\lambda$. We find through experiments that R-TBS displays higher prediction accuracy for all $\lambda$ in the range of $[0.1,0.5]$, so precise tuning of $\lambda$ is not critical. In addition, there is not enough data to warm up the models on different sampling schemes, so we report the model performance on all 30 batches. Similarly, we report 20\% ES for this dataset, due to the limited number of batches. \revision{For each sampling scheme, the smoothing parameter that leads to the best misclassification rate is used}. 

The results are shown in Figure~\ref{fig:NBaccuracy}. The misclassification rate for R-TBS, SW, and Unif are \revision{26.2\%, 28.7\%, and 29.1\%} and the 20\% ES values are \revision{43.2\%, 48.8\%, and 42.5\%}. Importantly, for this dataset the changes in the underlying data patterns are less pronounced than in the previous two experiments. Despite this, SW fluctuates wildly, yielding inferior accuracy and robustness. In contrast, Unif barely reacts to the context changes. As a result, Unif is very slightly better than R-TBS with respect to robustness, but at the price of lower overall accuracy. Thus, R-TBS is generally more accurate under mild fluctuations in data patterns, and its superior robustness properties manifest themselves as the changes become more pronounced.

\section{Related Work}\label{sec:relwork}

\textbf{Time-decay and sampling:} Work on sampling with unequal probabilities goes back to at least Lahiri's 1951 paper~\cite{Lahiri51}; see~\cite[Section~4]{OlkenR95} for some additional discussion of early work. A growing interest in streaming scenarios with weighted and decaying items began in the mid-2000's, with most of that work focused on computing specific aggregates from such streams, such as heavy-hitters, subset sums, and quantiles; see, e.g., \cite{AlonDLT05a,CohenS06,CormodeKT08}. The first papers on time-biased reservoir sampling with exponential decay are due to Aggarwal~\cite{Aggarwal06} and Efraimidis and Spirakis~\cite{efraimidisS06}; batch arrivals are not considered in these works. As discussed in Section~\ref{sec:intro}, the sampling schemes in~\cite{Aggarwal06} are tied to item sequence numbers rather than the wall clock times on which we focus; the latter are more natural when dealing with time-varying data arrival rates.

\textbf{Forward decay:} Cormode et al.~\cite{CormodeSSX09} consider an alternative framework for temporally-biased sampling based on an increasing ``forward decay'' function $g$, where the appearance probability at time $t_k$ of an item arriving at time $t_i\le t_k$ is proportional to $g(t_i)/g(t_k)$. This approach can be used to convert any recursively-defined weighted sampling scheme for a finite population to a streaming algorithm. Note that the decay behavior is quite different between forward and backward schemes. For example, a decay function $f(\alpha)=1/(1+\alpha)^2$ would naturally lead to forward decay function $g(t)=(1+t)^2$. Observe, however, that in the backward scheme with uniform item arrival times and $\Delta=1$, an item arriving at time~$t_i$ decays by a factor of $f(0)/f(1)=1/4$ one time unit after it arrives, whereas in the forwards scheme, the decay factor is $g(i)/g(i+1)= i^2/(i+1)^2$, which becomes close to 1 for large~$i$. As discussed in \cite{CormodeSSX09}, forward decay schemes support a notion of ``relative decay'' where the weight of an item is determined by its fractional distance between an initial ``landmark'' time and the current time. In this paper we focus on backward decay because this latter notion is widely used and, we believe, easier for users to understand. Moreover, as indicated above, items decay relatively slowly in the forward scheme, which can cause ML algorithms to adapt too slowly to changes in the data.
Finally, having experimented with forward decay versions of T-TBS and R-TBS, we observed that another issue with the forward decay in our setting is numerical stability. The work in \cite{CormodeSSX09} focused on analytical queries that are defined for items arriving during a specified, not-too-long time interval. In this setting, one can choose the forward-decay landmark time to be the time when the first item of interest arrives. For a sequence of such queries, the landmark time can be repeatedly shifted forward so that none of the item timestamps (measured relative to the landmark time) becomes too large.  In our setting, our sample can retain items that have arrived arbitrarily long ago, so the only feasible landmark time is $t=0$. This means that as time passes, the (absolute) time stamps become very large. The large timestamps pose numerical difficulties that are not amenable to normalization strategies. So again, we are motivated to focus on backward decay.

In the special case of an exponential decay function, the notions of backward and forward decay coincide, and so algorithms based on forward decay fall within our framework. (It is not hard to show that exponential decay functions are the only functions for which these notions coincide.) In this exponential setting, the authors in \cite{CormodeSSX09} provide a time biased reservoir sampling algorithm based on the A-Res weighted sampling scheme proposed in~\cite{efraimidisS06}. Rather than enforcing \eqref{eq:expratio}, however, the algorithm enforces the (different) A-Res biasing scheme. In more detail, if $s_{i}$ denotes the element at slot~$i$ in the reservoir, then the algorithm in~\cite{efraimidisS06} implements a scheme where an item $x$ is chosen to be at slot $i + 1$ in the reservoir with probability $w_x/(\sum_{j= 1}^x w_j - \sum_{j=1}^{i} w_{s_j})$. From the form of this equation, it becomes clear that the resulting sampling algorithm violates \eqref{eq:expratio}. Indeed, Efraimidis~\cite{Efraimidis15} gives some numerical examples illustrating this point (in his comparison of the A-Res and A-Chao algorithms). We would argue that the constraint on appearance probabilities in \eqref{eq:expratio} is easier to understand in the setting of model management than the foregoing constraint on initial acceptance probabilities. The closest solution to the exponential version of R-TBS adapts the weighted sampling algorithm of Chao~\cite{Chao82} to batches and forward decay; we call the resulting algorithm B-Chao and describe it in Appendix~\ref{sec:chao}.
Unfortunately, as discussed in the appendix, the relation in \eqref{eq:expratio} is violated both during the initial fill-up phase and whenever the data arrival rate becomes slow relative to the decay rate, so that the sample contains ``overweight'' items. Including overweight items causes over-representation of older items during initial fill-up or newer items during low inflow, thus potentially degrading predictive accuracy. The root of the issue is that the sample size is nondecreasing over time. The R-TBS algorithm is the first algorithm to correctly (and, for exponential decay, optimally) deal with ``underflows'' by allowing the sample to shrink---thus handling data streams whose flow rates vary unrestrictedly over continuous time. The current paper also explicitly handles batch arrivals and explores parallel implementation issues. The VarOpt sampling algorithm of Cohen et al.~\cite{CohenDKLT11}---which was developed to solve the specific problem of estimating ``subset sums''---can also be modified via forward decay. The resulting algorithm is more efficient than Chao, but as stated in \cite{CohenDKLT11}, it has the same statistical properties, and hence does not satisfy \eqref{eq:expratio}.

\textbf{Model management:} A key goal of our work is to support model management; see~\cite{GamaZBPB14} for a survey on methods for detecting changing data---also called ``concept drift'' in the setting of online learning---and for adapting models to deal with drift. As mentioned previously, one possibility is to re-engineer the learning algorithm. This has been done, for example, with support-vector machines (SVMs) by developing incremental versions of the basic SVM algorithm~\cite{CauwenberghsP00} and by adjusting the training data in an SVM-specific manner, such as by adjusting example weights as in Klinkenberg~\cite{Klinkenberg04}. Klinkenberg also considers using curated data selection to learn over concept drift, finding that weighted data selection also improves the performance of learners. Our approach of model retraining using time-biased samples follows this latter approach, and is appealing in that it is simple and applies to a large class of machine-learning models. The recently proposed Velox system for model management~\cite{CrankshawBGLZFG15} ties together online learning and statistical techniques for detecting concept drift. After detecting drift through poor model performance, Velox kicks off batch learning algorithms to retrain the model. Our approach to model management is complementary to the work in~\cite{CrankshawBGLZFG15} and could potentially be used in a system like Velox to help deployed models recover from poor performance more quickly. The developers of the recent MacroBase system~\cite{BailisGMNRS17} have incorporated a time-biased sampling approach to model retraining, for identifying and explaining outliers in fast data streams. MacroBase essentially uses Chao's algorithm, and so could potentially benefit from the R-TBS algorithm to enforce the inclusion criterion~\eqref{eq:expratio} in the presence of highly variable data arrival rates.

\section{Conclusion}\label{sec:concl}

Our experiments with classification and regression algorithms, together with the prior work on graph analytics in \cite{XieTSBH15}, indicate the potential usefulness of periodic retraining over time-biased samples to help analytics algorithms deal with evolving data streams without requiring algorithmic re-engineering. To this end we have developed and analyzed several time-biased sampling algorithms that are of independent interest.

In particular, the R-TBS algorithm allows simultaneous control of both the item-inclusion probabilities and the sample size, even when the data arrival rate is unknown and can vary arbitrarily. We have generalized our preliminary algorithms, analyses, and experiments in \cite{HentschelHT18} to arbitrary decay functions. Both theory and empirical results lead us to recommend exponentially and subexponentially decreasing decay functions to achieve reasonable storage and performance.  We found that the runtime performance, as well as the resulting accuracy and robustness of ML models, was comparable for the various decay functions that we studied, so we expect that the choice of decay function will be driven by the application setting, as discussed in Section~\ref{sec:intro}. For exponential decay functions, R-TBS maximizes the expected sample size and minimizes sample-size variability. For non-exponential decay functions, we have provided techniques to trade off storage with sample-size behavior in a principled and controllable manner; the user can similarly trade off storage requirements and control of inclusion probabilities. 

We have also provided techniques for distributed implementation of R-TBS and T-TBS, and have shown that use of time-biased sampling together with periodic model retraining can improve model robustness in the face of abnormal events and periodic behavior in the data. In settings where (i) the mean data arrival rate is known and (roughly) constant, as with a fixed set of sensors, and (ii) occasional sample overflows can be easily dealt with by allocating extra memory, we recommend use of T-TBS to precisely control item-inclusion probabilities. In many applications, however,  we expect that either (i) or (ii) will violated, in which case we recommend the use of R-TBS. Our experiments showed that R-TBS is superior to sliding windows over a range of parameter values, and hence does not require highly precise parameter tuning; this may be because time-biased sampling avoids the all-or-nothing item inclusion mechanism inherent in sliding windows.

An interesting future direction is to apply and extend our sampling schemes to other types of streaming analytics. Another goal is to combine our methods with drift-detection techniques to achieve end-to-end model management solutions.

\begin{appendix}

\section{Proofs}\label{sec:proofs}

\textbf{Proof of Theorem~\ref{th:recurr}}
Denote by $B_k=|\xB_k|$ the (random) size of $\xB_k$ for $k\ge 1$. We therefore assume that $\{B_k\}_{k\ge 1}$  are mutually independent and identically distributed as a random variable $B$ having finite mean $b\ge n\gamma$. To prove assertion~(i) of the theorem, write
\begin{equation}\label{eq:csum}
C_k=\sum_{i=0}^{k-1} N_{i,k},
\end{equation}
where $N_{i,k}$ is the number of sample items from batch $\xB_{k-i}$, i.e, the number of sample items of age $i\Delta$. Observe that the set of sample items of age~$i\Delta$ comprises those items in batch $\xB_{k-i}$ that survive $k-i+1$ rounds of Bernoulli sampling with respective success probabilities of $q,p_{i,i+1},\ldots,p_{i,k}$. As is well known, such a sample is probabilistically equivalent to a single Bernoulli sample with success probability $q\times p_{i,i+1}\times \cdots\times p_{i,k}=qf_k$. Thus we have
\[
\begin{split}
\mean[C_k]&=\mean\bigl[\mean[C_k\mid B_1,\ldots,B_k]\bigr]=\mean\biggl[\sum_{i=0}^{k-1}\mean[N_{i,k}\mid B_{k-i}]\biggr]\\
&=\mean\biggl[\sum_{i=0}^{k-1} qf_i B_{k-i}\biggr] = \sum_{i=0}^{k-1} qf_i\mean[B_{k-i}]=qbF_{k-1}=nF_{k-1}/F_\infty,
\end{split}
\]
since the $N_{i,k}$ are mutually independent given $B_1,\ldots,B_k$, and each $N_{i,k}$ depends on the batch sizes only through $B_{k-i}$. Assertion~(i) now follows immediately. The proof of assertion~(ii) is similar, and uses the fact that $\var[N_{i,k}\mid B_{k-i}]=qf_i(1-qf_i)B_{k-i}$.

Assertion~(iii) follows from \eqref{eq:csum} and Hoeffding's inequalities~\cite{Hoeffding63}. Indeed, a direct application yields the result in assertion~(ii)(a):
$
\prob{C_k\ge (1+\eps)n}\le \exp(-2kn^2\eps^2/\bbar^2)
$
for $\eps,k>0$. To prove assertion~(ii)(b), fix $\eps>0$ and $\delta\in(0,\eps)$ and observe that, by assertion~(i), we have $\mu_k\triangleq \mean[C_k]\to n$ as $k\to\infty$, so that $\mu_k\ge (1-\delta)n$ for sufficiently large $k$. Again applying Hoeffding's inequalty, we have that
\[
\prob{C_k\le (1-\eps)n}\le \prb\Bigl[C_k\le \frac{(1-\eps)}{(1-\delta)}\mu_k\Bigr]
=\prob{C_k\le (1-\eps_\delta)\mu_k}\le e^{-2kn^2\eps_\delta^2/\bbar^2},
\]
where $\eps_\delta=(\eps-\delta)/(1-\delta)$.

\vskip\baselineskip
\textbf{Proof of Theorem~\ref{th:recurGam}} Fix $f$ and define the process $\{\xS_k\}_{k\ge 0}$ as in Section~\ref{sec:ttbs}. We claim that $\{\xS_k\}_{k\ge 0}$ is an irreducible, aperiodic, time-homogeneous Markov chain with state space $\Sigma=2^{[0..\bbar]\times[0,1,\ldots]}$. Indeed the time-homogenous Markov property follows from the one-step recursive nature of the sample-update process. To prove the rest of the claim, set $\alphabar(s)=\max\{\,i:(n,i)\in s\,\}$ for $s\in\Sigma$, so that $\alphabar(s)$ is the age of the oldest item(s) in the sample. Next observe that, for any $s,s'\in\Sigma$ there is a positive probability of going from $s$ to $\emptyset$ in one step, and then a positive probability of going from $\emptyset$ to $s'$ in $l$ steps for any $l\ge \alphabar(s')$.

To prove assertion~(i) of the theorem,  fix $m\ge 0$. If $\bbar=\infty$, so that the batch size is unbounded, then set $s=\{(m',0)\}$, where $m'=\min\{\,i\ge m: \prob{B=i}>0\,\}$, and observe that $\prob{\xS_0=s}\ge q^{m'}\prob{B=m'}>0$. If $\bbar<\infty$, then set $k=\lceil{m/\bbar}\rceil$ and $s=\{(\bbar,0),\ldots,(\bbar,k-1)\}$, and observe that $\prob{\xS_1=s}\ge \prod_{i=0}^{k-1}(qf_i)^{\bbar}\prob{B=\bbar}>0$.

To prove assertion~(ii), it suffices to show that the chain is \emph{recurrent} in that $\prob{\xS_k=s\text{ i.o.}}=1$ for all $s\in\Sigma$. To this end, we apply an extended version of Foster's Theorem due to Meyn and Tweedie~\cite[Theorem~2.1(i)]{MeynT94}. This result asserts that a sufficient condition for recurrence is the existence of a nonnegative unbounded function $V$ on $\Sigma$, a function $k: \Sigma\mapsto \{1,2,\ldots\}$, and a finite subset $A\subseteq\Sigma$ such that
\begin{equation}\label{eq:extfoster}
\mean_s[V(\xS_{k(s)})] - V(s)\le 0
\end{equation}
for all $s\in\Sigma\setminus A$, where, in general, $\mean_s[g(\xS_k)]=\mean[g(\xS_k)\mid \xS_0=s]$.\footnote{For purposes of analyzing the chain, we let $s$ be any valid state that lies in $\Sigma\setminus A$, even though, when actually sampling, the (random) initial state is $\xS_0=\{(N_0,0)\}$ where $N_0$ is a Binomial$(B,q)$ random variable.} Let $\psi(s)$ denote the sample size corresponding to state $s\in\Sigma$, i.e., $\psi(s)=\sum_{(n,i)\in s}n$. Then we set $V(s)=(\psi(s)-n)^2+\alphabar(s)$. We now develop expressions for the left side of \eqref{eq:extfoster}, which then determine the required values for $k(s)$ and $A$. First consider a fixed integer $k\ge 1$ and state $s\in\Sigma$, and observe that we can write $\mean_s[V(\xS_k)] - V(s) = \mean_s\bigl[\bigl(\psi(\xS_k)-n\bigr)^2-\bigl(\psi(s)-n\bigr)^2\bigr] + \mean_s[\alphabar(\xS_k)-\alphabar(s)]$. We analyze each of the two terms on the right separately.

For the first term, denote by $n_i$ the number of age-$i$ sample items when the sample is in state $s$ and set $r_{i,k}=1-(f_{i+k}/f_i)$. (Here and elsewhere we suppress the explicit dependence upon $f$ in our notation.) We can write $\psi(\xS_k)-\psi(s) = D_1-D_2$, where $D_1=\sum_{i=0}^{k-1}N^{(1)}_{i,k}$ and $D_2=\sum_{i=0}^{\alphabar(s)}N^{(2)}_{i,k}$, with $N^{(1)}_{i,k}$ and $N^{(2)}_{i,k}$ denoting Binomial$(B_{k-i},qf_i)$ and Binomial$(n_i,r_{i,k})$ random variables, respectively. Here $D_1$ is the net number of items (after decay) inserted into the sample during the first $k$ steps and $D_2$ is the total number of initial items that have been removed from the sample during these $k$ steps. Then $D=D_1-D_2$ is the overall change in the sample size. Observe that $F^{(2)}_\infty\le F_\infty\le \sum_{i=0}^\infty if_i<\infty$ by assumption. Recalling that $q=n/(bF_\infty)$, straightforward calculations similar to those given previously show that $\mean[D^2]=\mean[D^2_1]-2\mean[D_1]\mean[D_2]+\mean[D^2_1]$, where
\[
\begin{split}
&\mean[D_1]= nF_{k-1}/F_\infty,\quad
\mean[D^2_1]=nF_{k-1}/F_\infty-n^2F^{(2)}_{k-1}/(bF^2_\infty)+\mean^2[D_1],\\
&\mean[D_2]=\sum_{i=1}^{\alphabar(s)}n_ir_{i,k},\quad
\mean[D^2_2]= \sum_{i=1}^{\alphabar(s)}n_ir_{i,k}(1-r_{i,k})+\mean^2[D_2].
\end{split}
\]
Then we have 
\begin{equation}\label{eq:psiRel}
\mean_s\bigl[\bigl(\psi(\xS_k)-n\bigr)^2-\bigl(\psi(s)-n)^2\bigr]
=\mean_s\bigl[\bigl(\psi(s)+D-n\bigr)^2-\bigl(\psi(s)-n\bigr)^2\bigr] =2\mean[D]\bigl(\psi(s)-n\bigr)+\mean[D^2].
\end{equation}
Because $\lim_{k\to\infty}r_{i,k}=1$ for all $i$, it follows that
\begin{equation}\label{eq:limdrifta}
\lim_{k\to\infty}\mean_s\bigl[\bigl(\psi(\xS_k)-n\bigr)^2-\bigl(\psi(s)-n)^2\bigr]
=-\bigl(\psi(s)-n\bigr)^2+\phi_{n,b},
\end{equation}
where $\phi_{n,b}=n\bigl(1-nF^{(2)}_\infty/(bF^2_\infty)\bigr)$.

For the second term, we have
$
\mean_s[\alphabar(\xS_k)]
=\sum_{l=0}^{k-1}\prob{\alphabar(\xS_k) >l}
=\sum_{l=0}^{k-1}\bigl[1-\prob{\alphabar(\xS_k) \le l}\bigr].
$
The event $\alphabar(\xS_k)\le l$ occurs if and only if, after $k$ steps, every group---both initial and subsequent---with age $>l$ has lost all of its members due to the decay process. Conditioning on the batch sizes, we have
\[
\begin{split}
&\prob{\alphabar(\xS_k) \le l}
=\mean\bigl[\prob{\alphabar(\xS_k) \le l\mid B_0,\ldots,B_{k-1}}\bigr]\\
&\quad=\mean\biggl[\prod_{i=l+1}^{\alphabar(s)+k}(1-qf_i)^{Z_i}\biggr]
\ge \prod_{i=l+1}^{\alphabar(s)+k}(1-qf_i)^{\bbar}\ge  \prod_{i=l+1}^{\alphabar(s)+k}(1-f_i)^{\bbar}
\end{split}
\]
where $Z_i=B_i$ if $i<k$ and $Z_i=n_{i-k}$ if $i\ge k$. Thus
\begin{equation}\label{eq:meanAgeBd}
\mean_s[\alphabar(\xS_k)]\le \Gamma\bigl(\alphabar(s)+k;\bbar),
\end{equation}
where we define
$
\Gamma(j;v)=\sum_{l=0}^{j-1}\bigl[1-\prod_{i=l+1}^j(1-f_i)^v\bigr]
$
for $j\ge 0$ and $v>0$. For $v>0$, denote by $v^*$ the smallest even integer greater than or equal to $v$. Since $(1-x)(1-y)\ge (1-x-y)$ for any $x,y\in[0,1]$, we have by induction that 
\[
\prod_{i=l+1}^j(1-f_i)^v=\biggl[\prod_{i=l+1}^j(1-f_i)\biggr]^v\ge \biggl[\prod_{i=l+1}^j(1-f_i)\biggr]^{v^*} 
 \ge (1-F_{l+1,j})^{v^*}\ge 1-v^*F_{l+1,j},
\]
where $F_{a,b}=\sum_{i=a}^bf_i$ and we have used Bernoulli's inequality. It follows that
\[
\Gamma(j;v)
\le v^*\sum_{l=0}^{j-1} F_{l+1,j}=v^*\sum_{l=0}^{j-1}\sum_{m=l+1}^{j-1}f_m
=v^*\sum_{i=0}^j if_i\le v^*\sum_{i=0}^\infty if_i\triangleq\Gamma^*(v)<\infty,
\]
for $v>0$; the second equality follows by interchanging the order of summation. Thus, by \eqref{eq:meanAgeBd},
\begin{equation}\label{eq:limdriftb}
\mean_s[\alphabar(\xS_k)-\alphabar(s)]\le \Gamma^*(\bbar)-\alphabar(s)
\end{equation}
for $k\ge 0$. Combining \eqref{eq:limdrifta} and \eqref{eq:limdriftb}, we have
$
\lim_{k\to\infty}\mean_s[V(\xS_k)]-V(s)\le -\bigl(\psi(s)-n\bigr)^2-\alphabar(s)+\tilde\phi_{n,b},
$
where $\tilde\phi_{n,b}=\phi_{n,b}+\Gamma^*(\bbar)$.

Fix $\eps>0$ and, for $s\in\Sigma$, choose $k(s)$ to be large enough so that
\[
\mean_s[V(\xS_{k(s)})]-V(s)\le -\bigl(\psi(s)-n\bigr)^2-\alphabar(s)+\tilde\phi_{n,b}+\eps.
\]
Next choose $c$ and $a$ large enough so that $\min\{(c-n)^2,a\}\ge \tilde\phi_{n,b}+\eps$, and define the finite set
\begin{equation}\label{eq:defA}
A=\{s\in\Sigma:\psi(s)\le c\text{ and }\alphabar(s)\le a\}.
\end{equation}
It is now straightforward to verify that \eqref{eq:extfoster} holds. Indeed, if $s\in\Sigma\setminus A$, then either $\psi(s)>c$, or $\alphabar(s)>a$, or both. If $\phi(s)>c$, then, since $-\alphabar(s)\le 0$, 
\[
\mean_s[V(\xS_k)]-V(s)\le -\bigl(\psi(s)-n\bigr)^2+\tilde\phi_{n,b}+\eps\le -(c-n\bigr)^2+\tilde\phi_{n,b}+\eps\le 0,
\]
and if $\alphabar(s)>a$, then, since $-\bigl(\psi(s)-n\bigr)^2\le 0$,
\[
\mean_s[V(\xS_k)]-V(s)\le -\alphabar(s)+\tilde\phi_{n,b}+\eps\le -a+\tilde\phi_{n,b}+\eps\le 0,
\]
so that \eqref{eq:extfoster} holds in both these cases. If both $\phi(s)>c$ and $\alphabar(s)>a$, then, of course, \eqref{eq:extfoster} holds as well.  Assertion~(ii) now follows from \cite[Theorem~2.1(i)]{MeynT94}.

To prove assertion~(iii), we show that if \eqref{eq:supratio} holds, then the chain is \emph{positive} recurrent in that the expected time between successive visits to any fixed state $m$ is finite. Note that, by the Markov property, the times between successive visits to $m$ are independent and identically distributed. Thus positive recurrence immediately implies assertion~(iii)(a). Positive recurrence also implies assertion~(iii)(b). Specifically, the chain $\{\xS_k\}_{t\ge 0}$---being irreducible, aperiodic, and positive recurrent---is ergodic, and so has a stationary distribution $\pi$ \cite[Thm.~3.3.1]{bremaud99}. This distribution is also a limiting distribution of the chain; in other words, $\xS_k\Rightarrow \xS_\infty$ for any fixed initial state, where $\Rightarrow$ denotes convergence in distribution and $\xS_\infty$  has distribution $\pi$ \cite[Thm.~4.2.1]{bremaud99}. Setting $C_\infty=\psi(\xS_\infty)$, we then have $C_k\Rightarrow C_\infty$. Moreover, a simple calculation as in the proofs of assertions~(i) and (ii) of Theorem~\ref{th:recurr} shows that $\mean[C_k^2]=qbF_k-q^2bF^{(2)}_k+(qbF_k)^2$, so that $\sup_k\mean[C_k^2]<\infty$. Thus $\{C_k\}_{k\ge 0}$ is uniformly integrable, and hence $\mean[C_\infty]=\lim_{k\to\infty}E[C_k]=n$~\cite[p.~338]{Billingsley95}. Finally, by the strong law of large numbers for Markov chains---see, e.g., \cite[Thm.~3.4.1]{bremaud99}---we have $\lim_{k\to\infty}k^{-1}\sum_{i=0}^k C_i=E[C_\infty]=n$ w.p.1.

To show positive recurrence, we establish a stronger drift condition using the function $V$ defined above, namely that 
\begin{equation}\label{eq:extfosterPos}
\mean_s[V(\xS_{k(s)})] - V(s)\le -k(s)
\end{equation}
for $x\in\Sigma\setminus A$, where $A$ has the same structure as in \eqref{eq:defA}, but uses different constants $c$ and $a$. We then apply Theorem~2.1(ii) in~\cite{MeynT94}. (The theorem also requires that $V$ be bounded on $A$, which is obvious.) First observe that our assumption on $f$ implies that, for any $k\ge 1$, we have $\inf_{i\ge 0}r_{i,k}\ge r_k$ for $k\ge 1$, where $r_k=(1-g_k)\to 1$ as $k\to\infty$. Expanding the terms in \eqref{eq:psiRel} and using the inequality on $r_{i,k}$, we obtain
\begin{equation}\label{eq:expCalc}
\begin{split}
&\mean_s\bigl[\bigl(\psi(\xS_k)-n\bigr)^2-\bigl(\psi(s)-n\bigr)^2\bigr]\\
&\ \le 2\bigl(\beta_k n-r_k\psi(s)\bigr)\bigl(\psi(s)-n)+ \beta_k n
+\beta_k^2n^2-2n\beta_kr_k\psi(s)+\psi(s)+\psi^2(s)\\
&\ = u_k + v_k\psi(s) -(2r_k-1)\psi^2(s),
\end{split}
\end{equation}
where $\beta_k=F_{k-1}/F_\infty$, $u_k=n\beta_k+n^2\beta_k(\beta_k-2)\le 0$, and $v_k=1+2n(\beta_k+r_k-\beta_kr_k)$. Thus
\[
\mean_s\bigl[V(\xS_k)]-V(s)
\le v_k\psi(s) -(2r_k-1)\psi^2(s)+\Gamma^*(\bbar)-\alphabar(s)
\]
Fix $m$ such that $r_m>1/2$. From straightforward calculus, we see that
\[
v_m\psi(s) -(2r_m-1)\psi^2(s)+\Gamma^*(\bbar)-\alphabar(s)\le \frac{v_m}{(8r_m-4)}+\Gamma^*(\bbar)-\alphabar(s)
\]
for all $s\in\Sigma$. Now define $A$ as in \eqref{eq:defA}, choosing $c$ and $a$ so that
\[
 \max\bigl\{v_mc -(2r_m-1)c^2 +\Gamma^*(\bbar),v_m(8r_m-4)^{-1}+\Gamma^*(\bbar)-a\bigr\}< -m.
\]
The inequality in \eqref{eq:extfosterPos} now follows with $k(s)\equiv m$.

\vskip\baselineskip
\textbf{Proof of Theorem~\ref{th:downsamp}} We first assume that $\pi=\{x^*\}$, so that there exists a partial item in $L$, and prove the result for $x=x^*$ and then for $x\not= x^*$. We then prove the result when $\pi=\emptyset$. 

\vskip 0.5\baselineskip
\textit{Proof for $x=x^*$}: Observe that when $\pi=\{x^*\}$, we have $\prob{x^*\in S}=\frc(C)$. First suppose that $\floor{C'}=0$, so that $\frc(C')=C'$. Either the partial item $x^*$ is swapped and ejected in lines~\ref{ln:swap0} and \ref{ln:killA} or is retained as a partial item: $\pi'=\{x^*\}$. Thus
\[
\begin{split}
\prob{x^*\in S'}&=\prob{x^*\in S'\mid x^*\in L'}\prob{x^*\in L'}
=\frc(C')\prob{\text{no swap}}=\frc(C')\bigl(\frc(C)/C\bigr)\\
&=(C'/C)\frc(C)=\theta\prob{x^*\in S}.
\end{split}
\]
Next suppose that $0<\floor{C'}=\floor{C}$. Then the partial item may or may not be converted to a full item via the swap in line~\ref{ln:convert}. Denoting by
$r=\bigl(1-(C'/C)\frc(C)\bigr)/\bigl(1-\frc(C')\bigr)$ the probability that this swap does not occur, we have
\[
\begin{split}
&\prob{x^*\in S'}=\prob{x^*\in S'\mid x^*\in\pi'}\prob{x^*\in\pi'}
+\prob{x^*\in S'\mid x^*\not\in\pi'}\prob{x^*\not\in\pi'}\\
&\quad=\frc(C')\cdot\prob{\text{no swap}}+1\cdot \prob{\text{swap}}
=1-r\bigl(1-\frc(C')\bigr)\\
&\quad=(C'/C)\frc(C)=\theta\prob{x^*\in S}.
\end{split}
\]
Finally, suppose that $\floor{C'}<\floor{C}$. Either the partial item $x^*$ is swapped into $A$ in line~\ref{ln:swapB} or ejected in line~\ref{ln:move}. Thus
$\prob{x^*\in S'}=\prob{\text{swap}} =(C'/C)\frc(C)=\theta\prob{x^*\in S}$, establishing the assertion of the lemma for $x=x^*$ when the partial item~$x^*$ exists.

\vskip 0.5\baselineskip
\textit{Proof for $x\not=x^*$:} Still assuming the existence of $x^*$, set $I_x=1$ if item~$x$ belongs to $S'$ and $Y=I_x=0$ otherwise. Also set $p_x=\prob{x\in S'}=\mean[I_x]$. Since all full items in $S$ are treated identically, we have $p_x\equiv p$ for  $x\in A$, and
\[
\mean[|S'|]=\mean\Bigl[\sum_{x\in A}I_x+I_{x^*}\Bigr]=\sum_{x\in A}\mean[I_x]+\mean[I_{x^*}]=\floor{C}p+p_{x^*}
\]
so that, using \eqref{eq:meansize} and the above result, 
\[
\begin{split}
\prob{x\in S'} &=(\mean[|S'|]-p_{x^*})/\floor{C}=\bigl(C'-(C'/C)\frc(C)\bigr)/\floor{C}
=(C'/C)\bigl(C-\frc(C)\bigr)/\floor{C}\\
&=C'/C=\theta\prob{x\in S}
\end{split}
\]
for any full item $x\in A$.

\vskip 0.5\baselineskip
\textit{Proof when $\pi=\emptyset$:} We conclude the proof by observing that, if $\pi=\emptyset$, then
$
C'=\mean[|S'|]=\sum_{x\in A}p_x=\floor{C}p=Cp
$
and again $\prob{x\in S'}=C'/C=(C'/C)\prob{x\in S}$.

\vskip\baselineskip
\textbf{Proof of Theorem~\ref{th:union}}
First observe that $L$ is indeed a latent sample: $|A|+|\pi|=\ceil{C}$, and $|\pi|\le 1$. Since $C=C_1+C_2$ by Line~\ref{ln:plusC}, the remainder of assertion~(i) of the theorem follows from \eqref{eq:meansize}. To prove assertion~(ii), observe that for every $x\in A_1$, we have that $x\in A$, so that  $\prob{x\in S}=\prob{x\in S_1}=1$. If there is a partial item~$x^*\in\pi_1$, we have three cases. If $\frc(C_1)+\frc(C_2)<1$, then 
\[
\begin{split}
\prob{x^*\in S}
&=\prob{x^*\in\pi}\cdot\prob{x^*\in S\mid x^*\in\pi}=\frac{\frc(C_1)}{\frc(C_1)+\frc(C_2)}\cdot\bigl(\frc(C_1)+\frc(C_2)\bigr)\\
&=\frc(C_1)=\prob{x^*\in S_1}.
\end{split}
\]
If $\frc(C_1)+\frc(C_2)=1$, then
\[
\prob{x^*\in S}
=\prob{x^*\in A}\cdot\prob{x^*\in S\mid x^*\in A}=\frc(C_1)\cdot 1=\prob{x^*\in S_1}.
\]
Finally, if $\frc(C_1)+\frc(C_2)>1$, then
\[
\begin{split}
\prob{x^*\in S}
&=\prob{x^*\in\pi}\cdot\prob{x^*\in S\mid x^*\in\pi}+\prob{x^*\in A}\cdot\prob{x^*\in S\mid x^*\in A}\\
&=\frac{1-\frc(C_1)}{2-\frc(C_1)-\frc(C_2)}\cdot\bigl(\frc(C_1)+\frc(C_2)-1\bigr)+ \frac{1-\frc(C_2)}{2-\frc(C_1)-\frc(C_2)}\cdot 1\\
&=\frc(C_1)=\prob{x^*\in S_1}.
\end{split}
\]
This proves assertion~(ii), and the proof of assertion~(iii) is almost identical.

\vskip\baselineskip
\textbf{Proof of Theorem~\ref{th:rtbsIncl}}
Our proof of assertions~(i) and (ii) of the theorem is by induction on $k$. For $k=1$ and $x\in\xB_1$, we see from line~\ref{ln:dsampleBE} that $x$ initially appears in the latent sample $L'_0=(\xB_1,\emptyset,|\xB_1|)$ with probability~1. Denote by $S'_0$, $S_0$, and $S_1$ random samples generated from $L'_0$, $L_0$, and $L_1$ via Algorithm~\ref{alg:getSample} and by $C'_0$, $C_0$, and $C_1$ the sample weights of $L'_0$, $L_0$, and $L_1$. By Theorem~\ref{th:downsamp}, we have $\prob{x\in S_0}=\rho_1\prob{x\in S'_0}= \rho_1\cdot 1=\rho_1f(0)=\rho_1f(\alpha_{1,1})$. After being unioned with the empty latent sample $(A,\pi,C)=(\emptyset,\emptyset,0)$, we have by Theorem~\ref{th:union} that $\prob{x\in S_1}=\prob{x\in S_0}=\rho_1f(\alpha_{1,1})$. By definition of the downsampling operation, we have $C_0=\rho_1|\xB_1|=\rho_1 W_1$, and thus, again by Theorem~\ref{th:union}, $C_1=C_0+0=\rho_1 W$. This proves (i) and (ii) for $k=1$. 
Assume for induction that (i) and (ii) hold for $k-1$. Then the inductive step for (i) is as given in Section~\ref{sec:rtbsAlgE}. To prove the inductive step for (ii), denote by $L_{k,1}$ the latent sample obtained by downsampling $L_{k-1}$ in line~\ref{ln:dsampleE},  by $C_{k,0}=\rho_k|\xB_k|$ the sample weight resulting from downsizing $(\xB_k,\emptyset,|\xB_k|)$ in line~\ref{ln:dsampleBE}, and by $C_{k,1}$ the sample weight of $L_{k,1}$. Using Theorem~\ref{th:union} and the inductive hypothesis that $\rho_{k-1}=C_{k-1}/W_{k-1}$, we have
\[
C_k=C_{k,0}+C_{k,1}=\rho_k|\xB_k|+\frac{\rho_k}{\rho_{k-1}}\theta_k C_{k-1}=\rho_k(|\xB_k|+\theta W_{k-1})=\rho_k W_k,
\]
and the desired result follows.


To prove assertion~(iii), fix $k$ and $i$, and first observe that $f(\alpha_{i,k})\le f(\alpha_{i,k-1})$ since $f$ is nonincreasing. If $W_k\ge W_{k-1}$ or $W_{k-1}\le n$, then $\rho_k\le\rho_{k-1}$ and the desired result follows. If $W_{k-1}>n$ and $W_k<W_{k-1}$, then we have, setting $\theta_k=e^{-\lambda(t_k-t_{k-1})}$,
\[
\begin{split}
\rho_kf(\alpha_{i,k})&=\min\Bigl(1,\frac{n}{W_k}\Bigr)f(\alpha_{i,k})=\min\Bigl(1,\frac{n}{\theta_k W_{k-1}+|\xB_k|}\Bigr)\theta_k f(\alpha_{i,k-1})
\le \frac{n}{W_{k-1}}f(\alpha_{i,k-1})\\
&=\rho_{k-1}f(\alpha_{i,k-1}).
\end{split}
\]

\vskip\baselineskip
\textbf{Proof of Theorem~\ref{th:consol}}
We first show that $\fa_k(\alpha)\le f(\alpha)$ for all $\alpha$ of the form $\alpha=\alpha_{i,k}$. For $\alpha_{i,k}$ with $i>m(k)$, the assertion  follows immediately  from the definition of $\fa_k$, so we assume that $\alpha=\alpha_{i,k}$ with $i\le m(k)$ and argue by induction. The assertion is trivially true for $i=m(k)$ by definition of $\fa$. Write $m=m(k)$ and assume for induction that the assertion holds for some $i\le m$. Using the inductive assumption and \eqref{eq:bdLambda}, we have that
\[
f(\alpha_{i-1,m})
=f(\alpha_{i,m})\frac{f(\alpha_{i-1,m})}{f(\alpha_{i,m})}
\ge \fa_k(\alpha_{i,m})\frac{f(\alpha_{i,m}+\Delta)}{f(\alpha_{i,m})}
\ge  \fa_k(\alpha_{i,m})e^{-\lambda\Delta}
=\fa_k(\alpha_{i-1,m}),
\]
and the assertion holds.

To establish assertion~(i) of the theorem, observe that $|f(\alpha)-\fa_k(\alpha)|=f(\alpha)-\fa_k(\alpha)\le f(\alpha)$. Because of the check in line~\ref{ln:entryCheck}, a batch of age $\alpha\ge\alpha^*_k$ that joins the consolidated latent sample satisfies $f(\alpha)<\delta_1$. From that point on, the value of $\alpha$ increases for the batch, and hence $f(\alpha)$ decreases, so the relation $f(\alpha)<\delta_1$ continues to hold, proving (i).

We now prove assertion~(ii) by induction. For $k=1$, we have $m(k)=1$ and thus $\sum_{i=1}^0 |\xB_i| f(\alpha_{i,k})=0<\delta_2$. By the induction hypothesis, we have $\sum_{i=1}^{m(k-1)-1} |\xB_i| f(\alpha_{i,k-1})<\delta_2$. Since $f$ is monotonically decreasing, we have that  $\sum_{i=1}^{m(k-1)-1} |\xB_i| f(\alpha_{i,k})<\delta_2$ prior to line~\ref{ln:entryCheck} during the processing of the new batch $\xB_k$. It follows that, at line~\ref{ln:bigUnion}, we have $F_\infty-\sum_{i=m(k)}^k f(\alpha_{i,k})<\delta_2/B^*_k$, where $B^*_k=\max_{i\le k}|\xB_i|$. Since $\sum_{i=1}^{m(k)-1} f(\alpha_{i,k})+\sum_{i=m(k)}^k f(\alpha_{i,k})<F_\infty$, we have $\sum_{i=1}^{m(k)-1} f(\alpha_{i,k})<\delta_2/B^*_k$. Thus $\sum_{i=1}^{m(k)-1} |\xB_i|f(\alpha_{i,k})\le B^*_k\sum_{i=1}^{m(k)-1} f(\alpha_{i,k})<\delta_2$.

To prove assertion~(iii), recall that $f_i=f(i\Delta)$ and observe that $\sum_{i=m(k)}^k f(\alpha_{i,k})=\sum_{i=0}^{k-m(k)}f_i$, so that $F_\infty-\sum_{i=m(k)}^k f(\alpha_{i,k})=\sum_{i=k-m(k)+1}^\infty f_i$. Thus, at line~\ref{ln:bigUnion}, $k-m(k)$ is the smallest integer such that $\sum_{i=k-m(k)+1}^\infty f_i<\delta_2/B^*_k$. Since $N$ is the smallest integer such that $\sum_{i=N}^\infty f_i<\delta_2/\bbar\le \delta_2/B^*_k$, we have $k-m(k)+1\le N$, and the assertion follows.

\vskip\baselineskip
\textbf{Proof of Proposition~\ref{prop:ssBDx}}
To prove assertion~(i) of the proposition, observe that $n'/W^*_k\ge 1$, and hence $n'/W_i\ge 1$ for all $i\le k$. Since $\rho_1=1$, it follows that $\rho_i=1$ for $i\le k$.

To prove assertion~(ii), first observe that since $f$ is nonincreasing, we have $\rho^*_{i,k}\ge \rho_{k-1}$ for $i\le k$ and hence $\rho^*_k\ge \rho_{k-1}$. Now suppose that $\rho_k$ achieves its minimum value in $[t_1,t_k]$ at time $t_{i'}$. Then $\rho_{i'}$ equals 1, $n'/W^*_{i'}$, or $\rho^*_{i'}$. Since $\rho^*_{i'}\ge\rho_{i'-1}\ge\rho_{i'}$, and since $n'/W_j<1$ for some $t_j\le t_k$ (so that $\rho_{i'}\le\rho_j< 1$) we have that $\rho_{i'}=n'/W_{i'}=n'/W^*_k$, the latter equality holding since $\rho_{i'}$ is the smallest value of $\rho$ in $[t_1,t_k]$. Thus for any $i\le k$ we have $\rho_i\ge \rho_{i'}\ge n'/W^*_k$.

To prove assertion~(iii), note that, under the conditions of the proposition, $\rho_k=\rho^*_k$ by \eqref{eq:setRho}. As shown above, we have $\rho^*_k>\rho_{k-1}$---with strict inequality because $f$ is strictly decreasing---and hence $\rho_k>\rho_{k-1}$.


\section{Analysis for Section~2}\label{sec:batchproofs}

\textbf{Bernoulli downsampling.} We first show that downsampling using the binomial distribution, as in Algorithm~1,
is statistically equivalent to simple sequential downsampling via Bernoulli coin flips. Consider a set $S$ and a subset $S'\subseteq S$ with $|S|=n$ and $|S'|=k$ (with $k\le n$). The probability of producing $S'$ from $S$ via $n$ coin flips with retention probability $p$ is $P_1(S')=p^k(1-p)^{n-k}$. Now consider the probability $P_2(S')$ of producing $S'$ from $S$ by first generating a binomial number $M$ of items to retain and then uniformly selecting $M$ specific items uniformly from $S$. The probability of selecting $M=k$ items is $\binom{n}{k}p^k(1-p)^{n-k}$ and the probability of selecting the specific set $S'$ of $k$ elements, given that $M=k$, is $\binom{n}{k}^{-1}$. Thus the overall probability $P_2(S')$ is the product of these terms, which equals $P_1(S')$. Thus, for any subset $S'$, both sampling schemes produce $S'$ with the same probability, and hence the schemes are statistically identical.

\vskip\baselineskip
\noindent\textbf{Batch reservoir sampling.} We now prove that Algorithm~2
does in fact produce uniform samples. As before, for $k\ge 1$, let $\xU_k=\bigcup_{j=1}^k \xB_j$ be the set of items arriving up through time $t_k$ and set $W_k=|\xU_k|$; we take $W_0=0$. Also write $B_k=|\xB_k|$.
Observe that $\{W_k\}_{k\ge 1}$ is nondecreasing and set $K=\min\{\,k\ge 1: W_k> n\,\}$. We first show that $S_k$ is a uniform sample from $\xU_k$ for $k\in[1..K]$. For $k<K$, we have $W_{k-1}\le n$ and $W_k=W_{k-1}+B_k\le n$. In this case, $S_{k-1}=\xU_{k-1}$ and $M=B_k$ with probability~1, since $M$ is hypergeometric$(B_k+W_{k-1},B_k,W_{k-1})$, so that $S_k=S_{k-1}\cup\xB_k=\xU_k$ and hence is trivially a uniform sample from $\xU_k$. For $k=K$, we have that $W_{k-1}\le n$ and $W_{k-1}+B_k> n$. Again, $S_{k-1}=\xU_{k-1}$. Fix $m\in[n-W_{k-1}..n]$ and consider a set $S=B\cup R$, where $B\subseteq \xB_k$ with $|B|=m$ and $R\subseteq \xU_{k-1}$ with $|R|=n-m$. In this case, we will have $S_k=S$ if (i) $M=m$, where $M$ is hypergeometric$(n,B_k,W_{k-1})$, (ii) the set of $m$ items accepted into the sample is exactly the set $B$, and (iii) the set of $m-(n-W_{k-1})$ items chosen to be overwritten is exactly the set $\xU_{k-1}-R$. Multiplying the probabilities of these three events together, we find that
\[
\prob{S_k=S}=\frac{\binom{B_k}{m}\binom{W_{k-1}}{n-m}}{\binom{W_{k-1}+B_k}{n}}\cdot \frac{1}{\binom{B_k}{m}}\cdot \frac{1}{\binom{W_{k-1}}{m-n+W_{k-1}}}=\frac{1}{\binom{W_{k-1}+B_k}{n}}=\frac{1}{\binom{W_k}{n}}.
\]
Since $m$ and $S$ were chosen arbitrarily, and there are $\binom{W_k}{n}$ possible choices for $S$, it follows that $S_K$ is a uniform sample of $\xU_K$. We establish the desired result for $k>K$ by induction. Suppose that $S_j$ is a uniform random sample of $\xU_j$ for $j\le k-1$. Since $k>K$, we have that $W_{k-1}> n$. Consider a set $S=B\cup R$ as above, but with $m\in[0..\min(n,B_k)]$. Also let $\xE$ denote the set of all subsets of $\xU_{k-1}-R$ of size $m$. Thus $\xE$ contains all possible sets of items that might be overwritten when accepting  $m$ items from $\xB_k$ into the sample. Fix $E\in\xE$ and consider the case where $S_{k-1}=E\cup R$. Then we will have $S_k=S$ if (i) $M=m$, where $M$ is hypergeometric$(n,B_k,W_{k-1})$, (ii) the set of $m$ items accepted into the sample is exactly the set $B$, and (iii) the set of $m$ overwritten items is exactly the set $E$. Also, by induction, we have $\prob{S_{k-1}=E\cup R}=1/\binom{W_{k-1}}{n}$ for all $E\in\xE$. Putting everything together, we have
\[
\begin{split}
&\prob{S_k=S}=\sum_{E\in\xE}\prob{S_{k-1}=E\cup R}\prob{S_k=S\mid S_{k-1}=E\cup R}\\
&\qquad=\sum_{E\in\xE} \frac{1}{\binom{W_{k-1}}{n}}\cdot\frac{\binom{B_k}{m}\binom{W_{k-1}}{n-m}}{\binom{W_{k-1}+B_k}{n}}\cdot \frac{1}{\binom{B_k}{m}}\cdot \frac{1}{\binom{n}{m}}
= |\xE|\frac{1}{\binom{W_{k-1}}{n}}\cdot\frac{\binom{B_k}{m}\binom{W_{k-1}}{n-m}}{\binom{W_{k-1}+B_k}{n}}\cdot \frac{1}{\binom{B_k}{m}}\cdot \frac{1}{\binom{n}{m}}\\
&\qquad =\textstyle{\binom{W_{k-1}-(n-m)}{m}}\cdot\frac{1}{\binom{W_{k-1}}{n}}\cdot\frac{\binom{B_k}{m}\binom{W_{k-1}}{n-m}}{\binom{W_{k-1}+B_k}{n}}\cdot \frac{1}{\binom{B_k}{m}}\cdot \frac{1}{\binom{n}{m}}=\frac{1}{\binom{W_k}{n}}.
\end{split}
\]
Again, since $m$ and $S$ are arbitrary, the desired result follows.


\section{Chao's Algorithm}\label{sec:chao}

In this section, we provide pseudocode for a batch-oriented, time-decayed version of Chao's algorithm~\cite{Chao82} for maintaining a weighted reservoir sample of $n$ items, which we call B-Chao. Recall that the goal of time-biased sampling is to enforce the relationship
\begin{equation}
\prob{x\in S_k}/\prob{y\in S_k}=f(\alpha_{i,k})/f(\alpha_{j,k}) \nonumber
\end{equation}
for arbitrary batch arrival times $t_i\le t_j\le t_k$ and arbitrary items $x\in\xB_i$ and $y\in\xB_j$, where $f$ is the decay function and $\alpha_{i,k}=t_k-t_i$ the age at time $t_k$ of an item belonging to batch $\xB_i$. For simplicity, we focus on exponential decay functions.

The pesudocode is given as Algorithm~\ref{alg:chao}. In the algorithm, the function $\textsc{Get1}(x,A)$ randomly chooses an item $i$ in a set $A$, and then sets $x\gets i$ and  $A\gets A\setminus\{x\}$. We explain the function \textsc{Normalize} below.

\begin{algorithm}[t]
\caption{Batched version of Chao's scheme (B-Chao)}\label{alg:chao}
{\footnotesize
$\lambda$: decay factor ($\ge 0$)\;
$n$: reservoir size\;
\BlankLine
\Comment{Initialize}
$S\gets \emptyset$\;
$W\gets 0$ \Comment*[r]{$W=$ agg. weight of non-overweight items}
$V\gets\emptyset$\Comment*[r]{$V$ holds overweight items}
$A\gets\emptyset$\Comment*[r]{$A$ hold newly non-overweight items}
\Comment{Process batches}
\For{$i\gets1,2,\ldots$}{
  \Comment{update weights}
  $W\gets e^{-\lambda}W$\;
   \lFor{$(z,w_z)\in V$}{$w_z\gets e^{-\lambda}w_z$}
 \Comment{Process items in batch}
  \For{$j\gets 1,2,\ldots,|\xB_t|$}{
    $\textsc{Get1}(x,\xB_i)$ \Comment*[r]{get new item to process}
    \If(\Comment*[f]{reservoir not full yet}){$|S|<n$}{
      $S\gets S\cup \{x\}$; $W\gets W+1$\;\label{ln:nonFull}
      }
    \Else(\Comment*[f]{reservoir is full}){
      $\textsc{Normalize}(x,V,A,W,\pi_x)$ \Comment*[r]{categorize items}
     \If{$\textsc{\emph{Uniform}}()\le \pi_x$}{
        \Comment{accept $x$ and choose victim to eject}
        $\alpha=0$; $y\gets$ {\bf null}; $U\gets\textsc{Uniform}()$\;
        \For(\Comment*[f]{attempt to choose from $A$}...){$(z,w_z)\in A$}{
          $\alpha\gets\alpha + \bigl(1-\frac{(n-|V|)w_z}{W}\bigr)/\pi_x$\;
          \If{$U\le \alpha$}{$A\gets A\setminus\{(z,w_z)\}$; $y\gets z$; {\bf break}}
          }
        \lIf(\Comment*[f]{... else remove victim from $S$}){$y==$ {\bf null}}{$\textsc{Get1}(y,S)$}
        \lIf(\Comment*[f]{Add new item to sample if not overweight}){$(x,1)\notin V$}{$S\gets S\cup\{x\}$}

        }
      $S\gets S\cup \{z:(z,w_z)\in A\}$; $A\gets\emptyset$\Comment*[r]{if no longer overweight, stop tracking}
      }
    }
  output $S\cup \{z:(z,w_z)\in V\}$
  }
}
\end{algorithm}

\begin{algorithm}[t]
\caption{Normalization of appearance probabilities}\label{alg:normp}
{\footnotesize
$x$: newly arrived item (has weight $=1$)\;
$V$: set of items that remain overweight (and their weights)\;
$A$: set of items that become non-overweight (and their weights)\;
$W$: aggregate weight of non-overweight items\;
$\pi_x$: inclusion probability for $x$\;
$n$: reservoir size\;
\BlankLine
$W\gets W+1+\sum_{(z,w_z)\in V}w_z$ \Comment*[r]{agg. wt. of new \& sample items}
\If(\Comment*[f]{$x$ is not overweight}){$n/W\le 1$}{
  $A\gets V$; $V\gets\emptyset$ \Comment*[r]{no item is now overweight}
  $\pi_x\gets n/W$
  }
\Else(\Comment*[f]{$x$ is overweight}){
  $\pi_x\gets 1$;
  $W\gets W-1$\;
  $D\gets \{(x,1)\}$ \Comment*[r]{$D=$ set of overweight items so far}
  \Repeat{$(n-|D|)w_z/W\le1$\Comment*[f]{first non-overweight item}}{
    $(z,w_z)\gets\textsc{GetMax}(V)$\;
    \If(\Comment*[f]{$z$ remains overweight}){$(n-|D|)w_z/W>1$}{
      $D\gets D\cup\{(z,w_z)\}$; $W\gets W-w_z$
      }
    \Else(\Comment*[f]{$z$ no longer overweight}){
      $A\gets A\cup \{(z,w_z)\}$}
      }
    $A\gets A\cup V$; $V\gets D$ \Comment*[r]{no more overweight items in $V$}
  }
}
\end{algorithm}

Note that the sample size increases to~$n$ and remains there, regardless of the decay rate. During the initial period in which the sample size is less than $n$, arriving items are included with probability~1 (line~\ref{ln:nonFull}); if more than one batch arrives before the sample fills up, then clearly the relative inclusion property in \eqref{eq:expratio} will be violated since all items will appear with the same probability even though the later items should be more likely to appear. Put another way, the weights on the first $n$~items are all forced to equal~1.

After the sample fills up, B-Chao encounters additional technical issues due to ``overweight'' items. In more detail, observe that $\mean[|S|]=\sum_{i\in S}\pi_i$, where $\pi_i=P[i\in S]$. At any given moment we require that $\mean[|S|]=\sum_{i\in S}\pi_i=n$. If we also require for each~$i$ that $\pi_i\propto w_i$, then we must have  $\pi_i=n w_i/W$, where $W=\sum_{i\in S}w_i$. It is possible, however, that $w_i/W>1/n$, and hence $\pi_i>1$, for one or more items $i\in S$. Such items are called \emph{overweight}. As in \cite{Chao82}, B-Chao handles this by retaining the most overweight item, say $i$,  in the sample with probability~1. The algorithm then looks at the reduced sample of size $n-1$ and weight $W-w_i$, and identifies the item, say $j$, having the largest weight $w_j$. If item~$j$ is overweight in that the modified relative weight $w_j/(W-w_i)$ exceeds $1/(n-1)$, then it is included in the sample with probability~1 and the sample is again reduced. This process continues until there are no more overweight items, and can be viewed as a method for categorizing items as overweight or not, as well as normalizing the appearance probabilities to all be less than 1. The \textsc{Normalize} function in Algorithm~\ref{alg:normp} carries out this procedure; Algorithm~\ref{alg:normp} gives the pseudocode. In Algorithm~\ref{alg:normp}, the function $\textsc{GetMax}(V)$ returns the pair $(z,w_z)\in V$ having the maximum value of $w_z$ and also sets $V\gets V\setminus\{(z,w_z)\}$; ties are broken arbitrarily. An efficient implementation would represent $V$ as a priority queue.

When overweight items are present, it is impossible to both maintain a sample size equal to $n$ and to maintain the property in \eqref{eq:expratio}. Thus, as discussed in Section~2.1 of \cite{Chao82}, the algorithm only enforces the relationship in \eqref{eq:expratio} for items that are not overweight. When the decay rate $\lambda$ is high, newly arriving items are typically overweight, and transform into non-overweight items over time due to the arrival of subsequent items. In this setting, recently-arrived items are overrepresented. The R-TBS algorithm, by allowing the sample size to decrease, avoids the over\-weight-item problem, and thus the violation of the relative inclusion property \eqref{eq:expratio}, as well as the complexity arising from the need to track overweight items and their individual weights (as is done in the pseudocode via $V$). We note that prior published descriptions of Chao's algorithm tend to mask the complexity and cost incurred by the handling of overweight items.

\section{Implementation of D-R-TBS on Spark}\label{sec:spark-impl}

In this section we discuss aspects of our implementation that are specific to Spark. Spark is a natural platform for implementing D-R-TBS because it supports streaming, machine learning, and efficient distributed data processing, and is widely used.  Efficient implementation is relatively straightforward for T-TBS but decidedly nontrivial for R-TBS because of both Spark's idiosyncrasies and the coordination needed between nodes.

\subsection{Spark Overview}
Spark is a general-purpose distributed processing framework based on a functional programming paradigm. Spark provides a distributed memory abstraction called a Resilient Distributed Dataset (RDD). An RDD is divided into partitions that are then distributed across the cluster for parallel processing. RDDs can either reside in the aggregate main memory of the cluster or in efficiently serialized disk blocks. An RDD is immutable and cannot be modified, but a new RDD can be constructed by transforming an existing RDD. Spark utilizes both lineage tracking and checkpointing of RDDs for fault tolerance. A Spark program consists of a single driver and many executors. The driver of a Spark program orchestrates the control flow of an application, while the executors perform operations on the RDDs, creating new RDDs. 

\subsection{Distributed Data Structures}

We leverage Spark Streaming to ingest batches of arriving data, thereby supporting input sources such as HDFS, Kafka, Flume, and so on. Each incoming batch $\xB_k$ is thus naturally stored as an RDD. We can store the reservoir using either a key-value store or a co-partitioned reservoir---see Section~5.2---but  prefer using a co-partitioned reservoir because it has lower overhead (since incoming batch partitions align with local reservoir partitions). We would like use Spark's distributed fault-tolerant RDD data structure to implement the co-partitioned reservoir. A problem arises, however, if we try to store the reservoir as a vanilla RDD: because RDDs are immutable, the large numbers of reservoir inserts and deletes at each time point would trigger the constant creation of new RDDs, quickly saturating memory. We therefore augment the RDD with with the in-place update technique proposed by Xie, et al.~\cite{XieTSBH15}. The key idea is to share objects across different RDDs. In particular, we store the reservoir as an RDD, each partition of which contains only one object, a (mutable) vector containing the items in the corresponding reservoir partition. A new RDD created from an old RDD via a batch of inserts and deletes references the same vector objects as the old RDD. We keep the lineage of RDDs intact by notifying Spark of changes to old RDDs by calling the \textsc{Unpersist} function. In case of failure, old RDDs (with old samples) can be recovered from checkpoints, and Spark's recovery mechanism based on lineage will regenerate the sample at the point of failure.

\subsection{Choosing Items to Delete and Insert}

Section~5.3
has detailed the centralized and distributed decision mechanisms for choosing items to delete and insert. Here, we add some Spark-related details for centralized decisions.

All of the transient large data structures are stored as RDDs in Spark; these include the set of item locations for the insert items $\mathcal{Q}$, the set of retrieved insert items $\mathcal{S}$, and the set of item locations for the delete items $\mathcal{R}$. To ensure the co-partitioning of these RDDs with the incoming batch RDD---and the reservoir RDD when the co-partitioned reservoir is used---we use a customized partitioner. For the join operations between RDDs, we use by default Spark's standard repartition-based join. When RDDs are co-partitioned and co-located, however, we implement a customized join algorithm that performs only local joins on corresponding partitions. 

\subsection{Fault Tolerance of Distributed Implementations}\label{sec:checkpoint}

\revision{We rely primarily on Spark's lineage tracking and checkpointing mechanisms to ensure the fault tolerance of our distributed algorithms. Spark Streaming's checkpointing mechanism is used to ensure the resiliency of the incoming batches. If the co-partitioned reservoir approach is used, we simply leverage Spark's built-in lineage and checkpointing mechanisms for the reservoir RDD. If the key-value store approach is used, then, because such stores are non-native to Spark, we need to do our own checkpointing, writing the reservoir content to the distributed file system; this adds more implementation overhead to the distributed algorithms. We also have to do our own checkpointing for any variables not stored in the foregoing distributed data structures, such as the current total weight and the current sample weight, for the distributed algorithms.}

\revision{Finally, we distinguish between checkpointing the reservoir for fault tolerance and materializing the reservoir for the use of external ML applications. Because failure doesn't happen very often, checkpointing occurs much less frequently than the arrival of incoming batches. On the other hand, no matter how the reservoir is implemented, its content needs to be materialized in a consumable format after processing each incoming batch, thereby enabling an external ML application to access the sample for model retraining. Because the changes to the reservoir between subsequent incoming batches are usually small, the system can write a small delta for each new batch, and write full snapshots  periodically.}

\end{appendix}

\section*{Acknowledgment}
The authors wish to thank Valerie Caro for her help with the numerical experiments.

\bibliographystyle{ACM-Reference-Format}
\bibliography{tbs-arxiv}

\end{document}